\documentclass[11pt,reqno]{amsart}

\newtheorem{theorem}{Theorem}[section]
\newtheorem{definition}[theorem]{Definition}
\newtheorem{proposition}[theorem]{Proposition}
\newtheorem{assumption}[theorem]{Assumption}
\newtheorem{corollary}[theorem]{Corollary}
\newtheorem{lemma}[theorem]{Lemma}

\allowdisplaybreaks
\makeatletter

\title[Dynamic Pareto Optima in Multi-Period Pure-Exchange Economies]{Dynamic Pareto Optima\vspace{0.25cm}\\in Multi-Period Pure-Exchange Economies\vspace{0.3cm}}

\author[Brandon Tam, Mario Ghossoub, and Silvana M.\ Pesenti]{Brandon Tam\vspace{0.15cm}\\ University of Toronto\vspace{0.65cm}\\
 {Mario Ghossoub\vspace{0.15cm}\\ University of Waterloo\vspace{0.65cm}\\ 
 {Silvana M.\ Pesenti\vspace{0.15cm}\\ University of Toronto\vspace{0.65cm}\\  
 \today}}}

\address{{\bf Brandon Tam}: University of Toronto -- Department of Statistical Sciences -- 700 University Ave, 9th Floor, Toronto, ON, M5G 1X6 -- Canada\vspace{0.1cm}}\email{\href{mailto:brandontam.tam@mail.utoronto.ca}{brandontam.tam@mail.utoronto.ca}\vspace{0.2cm}}

\address{{\bf Mario Ghossoub}: University of Waterloo -- Department of Statistics and Actuarial Science -- 200 University Ave.\ W.\ -- Waterloo, ON, N2L 3G1 -- Canada\vspace{0.1cm}}
\email{\href{mailto:mario.ghossoub@uwaterloo.ca}{mario.ghossoub@uwaterloo.ca}\vspace{0.2cm}}

\address{{\bf Silvana M. Pesenti}: University of Toronto -- Department of Statistical Sciences -- 700 University Ave, 9th Floor, Toronto, ON, M5G 1X6 -- Canada\vspace{0.1cm}}
\email{\href{mailto:silvana.pesenti@utoronto.ca}{silvana.pesenti@utoronto.ca}}

\thanks{\vspace{0.2cm}}
\thanks{ \vspace{0.2cm}}
\thanks{\vspace{0.2cm}}

\thanks{Brandon Tam acknowledges financial support from the Ontario Graduate Scholarship. Mario Ghossoub acknowledges financial support from the Natural Sciences and Engineering Research Council of Canada (RGPIN-2024-03744). Silvana Pesenti acknowledges financial support from the Natural Sciences and Engineering Research Council of Canada (RGPIN-2025-05847 and ALLRP 580632-22).}

\usepackage[T1]{fontenc}
\usepackage[round]{natbib}
\usepackage{amsmath}
\usepackage{amssymb}
\usepackage{amsthm}
\usepackage{fancyhdr}
\usepackage{enumitem}
\usepackage{comment}
\usepackage{relsize}
\usepackage{nicefrac}
\usepackage{bm}
\usepackage{mathrsfs}
\usepackage{graphicx}
\usepackage[utf8]{inputenc}
\usepackage{cancel}
\usepackage{mathtools}
\usepackage{tikz,xcolor}
\usepackage{dirtytalk}
\usetikzlibrary{cd,decorations.pathreplacing,matrix,arrows,positioning,automata,shapes,shadows,calc,fadings,decorations,through,intersections}
\usepackage{float}
\usepackage{subcaption}
\usepackage{multirow}
\usepackage[left=2.6cm,top=2.6cm,right=2.6cm,bottom=2.6cm,head=2.8cm,headsep=1cm, foot=3cm]{geometry}
\usepackage{hyperref}
\usepackage[skip=3mm,indent=12pt]{parskip}

\usepackage{booktabs}
\usepackage{siunitx}
\usepackage[colorinlistoftodos,textsize=small]{todonotes}
\usepackage{cleveref}
\usepackage{caption}
\usepackage{anyfontsize}
\usepackage{yhmath}
\usepackage{comment}
\usepackage{dsfont}

\definecolor{darkblue}{rgb}{0.1,0.1,0.9}
\definecolor{darkred}{rgb}{0.9,0.1,0.1}

\hypersetup{colorlinks, allcolors= darkblue}

\newcommand{\B}{\mathcal{B}}
\newcommand{\F}{\mathcal{F}}
\newcommand{\T}{\mathcal{T}}
\newcommand{\K}{\mathcal{K}}

\renewcommand{\L}{{\mathcal{L}}}
\newcommand{\bL}{{\boldsymbol{\mathcal{L}}}}

\newcommand{\A}{\mathcal{A}}
\newcommand{\IR}{\mathcal{I}\mathcal{R}}
\newcommand{\PO}{\mathcal{P}\mathcal{O}}
\newcommand{\CPO}{\mathcal{C}\mathcal{P}\mathcal{O}}

\newcommand{\MIR}{\mathcal{M}\mathcal{I}\mathcal{R}}
\newcommand{\MPO}{\mathcal{M}\mathcal{P}\mathcal{O}}
\newcommand{\CMPO}{\mathcal{C}\mathcal{M}\mathcal{P}\mathcal{O}}

\newcommand{\DIR}{\mathcal{D}\mathcal{I}\mathcal{R}}
\newcommand{\DPO}{\mathcal{D}\mathcal{P}\mathcal{O}}
\newcommand{\CDPO}{\mathcal{C}\mathcal{D}\mathcal{P}\mathcal{O}}

\newcommand{\V}{{\boldsymbol{V}}}
\newcommand{\X}{{\boldsymbol{X}}}
\newcommand{\W}{{\boldsymbol{W}}}

\newcommand{\Y}{{\boldsymbol{Y}}}

\newcommand{\Z}{{\boldsymbol{Z}}}
\newcommand{\0}{{\boldsymbol{0}}}

\newcommand{\mfg}{{\mathfrak{g}}}
\newcommand{\G}{\mathfrak{G}}
\newcommand{\mfT}{\mathfrak{T}}

\renewcommand{\epsilon}{{\varepsilon}}

\newcommand{\Id}{{\mathds{1}}}

\newcommand{\Rmnum}[1]{\expandafter\@slowromancap\romannumeral #1@}

\newcommand{\eqd}{\,{\overset{d}{=}}\,}

\newcommand{\brho}{\overline{\rho}}

\newcommand{\bR}{{\overline{R}}}

\newcommand{\mN}{{\mathcal{N}}}

\DeclareMathOperator*{\essinf}{ess\,inf}

\newcommand{\E}{{\mathbb{E}}}
\newcommand{\N}{{\mathbb{N}}}
\renewcommand{\P}{{\mathbb{P}}}

\newcommand{\R}{{\mathbb{R}}}

\renewcommand{\tilde}{\widetilde}
\renewcommand{\hat}{\widehat}

\newcommand{\cvx}{\preccurlyeq_{\hspace{- 0.4 mm} _{cx}}}

\newcommand{\scvx}{\prec_{\hspace{- 0.4 mm} _{cx}}}

\DeclareMathOperator*{\argmin}{\arg\min}

\newcommand{\ES}{{\mathrm{ES}}}

\begin{document}

\begin{abstract}
We study a problem of optimal allocation in a discrete-time multi-period pure-exchange economy, where agents have preferences over stochastic endowment processes that are represented by strongly time-consistent dynamic risk measures. We introduce the notion of dynamic Pareto-optimal allocation processes and show that such processes can be constructed recursively starting with the allocation at the terminal time. We further derive a comonotone improvement theorem for allocation processes, and we provide a recursive approach to constructing comonotone dynamic Pareto optima when the agents' preferences are coherent and satisfy a property that we call \emph{equidistribution-preserving}. In the special case where each agent's dynamic risk measure is of the distortion type, we provide a closed-form characterization of comonotone dynamic Pareto optima. We illustrate our results in a two-period setting.
\end{abstract} 

\maketitle

\sloppy


\bigskip

\section{Introduction} \label{sec:intro}

The study of efficiency in pure-exchange economies has served as the foundation for the characterization of Pareto-optimal (PO) allocations of an aggregate endowment among a population of economic agents. The historical roots of this field can be traced back to the seminal work of \citet{borch1960safety, borch1962equilibrium} and \citet{wilson1968theory}, within the framework of Expected Utility (EU) Theory with risk-averse agents. Under this framework, the celebrated Borch rule states that an allocation is PO if and only if the ratio of marginal utilities between the agents is constant. Furthermore, optimal allocations are non-decreasing deterministic functions of the aggregate endowment. In other words, PO allocations are comonotone when agents are risk averse. This has motivated a series of developments related to comonotone improvement results in the literature, including \citet{landsberger1994co}, \cite{chateauneuf1997review}, \citet{carlier2012pareto}, and \citet{denuit2025comonotonicity}, for instance.

Characterization results for Pareto optima have also been obtained for models beyond the EU framework. For instance, economies with ambiguity-sensitive preferences of the Choquet-Expected Utility (CEU) type \citep{schmeidler1989subjective} were examined by \cite{chateauneuf2000optimal}, \cite{dana2004ambiguity}, \cite{de2011ambiguity}, and \cite{beissner2023optimal}. \cite{dana2002equilibria, dana2004ambiguity} and \cite{de2011ambiguity} examine economies with Maxmin Expected Utility (MEU) preferences axiomatized by \cite{gilboa1989maxmin}. Exchange economies with variational preferences (e.g., \cite{maccheroni2006ambiguity}) were studied by \cite{dana2010overlapping} and \cite{ravanelli2014comonotone}, for instance. \cite{strzalecki2011efficient} consider general ambiguity-sensitive convex preferences. \cite{embrechts2018quantile, embrechts2020quantile} study PO risk exchange for agents with quantile preferences. Some recent work has focused on obtaining crisp characterizations that can be numerically implemented \citep{liu2020weighted, liu2022inf, ghossoub2026efficiency}. \citet{embrechts2018quantile} examine preferences given by quantile-based risk measures (specifically, the Range-Value-at-Risk) and the work is generalized to tail risk measures in \citet{liu2022inf}. In \citet{liu2020weighted}, the author studies a weighted risk-sharing problem and gives an explicit characterization when agent preferences are given by distortion risk measures. This characterization is further generalized to positively homogeneous monetary utilities (or equivalently, law-invariant coherent risk measures) in \citet{ghossoub2026efficiency}. \cite{jouini2008optimal} and \cite{filipovic2008optimal} also studied the case of economies with law-invariant monetary utilities. Recently, \cite{hara2026sharing} study Pareto optimality in pure-pure exchange economies with heterogeneous smooth ambiguity preferences (e.g., \cite{klibanoff2005smooth,denti2022model}).

While the aforementioned literature considered a single-period framework, which we henceforth refer to as the static setting, the nature of investment and consumption is dynamic in reality. As a result, the literature has also devoted some effort to the study of the multi-period optimal-allocation problem. For instance, in a discrete-time multi-period setting with agents having recursive utility preferences, the optimal consumption problem is studied in \citet{benhabib1988dynamics}, and a generalized market with multiple consumption goods is considered in \citet{dana1990structure}. For more sophisticated preferences that account for uncertainty in future consumption, see \citet{epstein1989substitution}, \citet{chew1990nonexpected}, \citet{chew1991recursive}, and \citet{kan1995structure}, for instance. Further work in the discrete-time setting include \citet{epstein2004intertemporal}, who study equilibrium asset prices for consumption and investment problems; and \citet{anderson2005dynamics}, who define and characterize steady states in infinite horizon problems. For a broad overview and summary of results in the discrete-time setting, we refer to \citet{backus2004exotic}.

In the continuous-time setting, recursive preferences are given by stochastic differential utilities (SDU), which have been introduced in \citet{duffie1992stochastic}, and their applications to asset pricing are first discussed in \citet{duffie1992asset}. In \citet{duffie1994efficient}, the existence of Pareto optima under SDU is established, and a characterization is derived. This characterization is improved in \citet{dumas1998efficient}, who show that Pareto optima can be computed by optimizing one value function, instead of one for each agent. For problems with consumption and investment, optimal strategies in various market settings are studied in a series of papers by Schroder and Skiadas. The case of unconstrained trading is considered in \citet{schroder1999optimal}, and its results are extended to agents with trading constraints in \citet{schroder2003optimal} and \citet{skiadas2007dynamic}. Non-tradable income is introduced in \citet{schroder2005lifetime}, and discontinuous market information is studied in \citet{schroder2008optimality}.

In the above literature on optimal consumption, or optimal allocations, in a dynamic framework, Pareto efficiency is defined from the perspective of time-$0$ preferences. Future consumption streams, future endowment streams, or future risk profiles are first discounted to time $0$, and then evaluated. Our contention is that this leads to an inherently myopic view that fails to capture the dynamic structure of endowment profiles. We propose a novel approach that allows to incorporate the dynamic structure of risk profiles into the decision-making process, through dynamic risk measures; and we introduce the associated notions of dynamic individual rationality and dynamic Pareto optimality.

Specifically, in this paper, we assume that each agent's initial endowment is given by a discrete-time stochastic process, and the initial endowments are aggregated and reallocated at each point in time. Moreover, agent preferences over discrete-time stochastic processes are represented by strongly time-consistent dynamic risk measures, which are extensively studied in the risk management literature \citep{cheridito2006dynamic, ruszczynski2010risk,bielecki2018unified, moresco2024uncertainty,coache2024reinforcement, dentcheva2024risk, pesenti2025risk,bielecki2025time}. Unlike the existing literature on recursive utilities, we assume that agent preferences are (conditionally) translation invariant, which is a classical assumption for risk measures. Therefore, the risk of a deterministic endowment is equal to the size of the endowment and agent preferences admit a recursive representation, as in \cite{cheridito2006dynamic} and \cite{ ruszczynski2010risk}. This recursive representation allows us to define a notion of optimality for all time points, contrasting with classical notions of efficiency which are defined from the perspective of time $0$ only, and hence are inherently myopic. Our new notion of optimality, which we call \emph{dynamic Pareto optimality} (or just Pareto optimality for simplicity), not only considers the risk of an allocation process at time $0$, but also incorporates the dynamic structure of risk profiles into the decision-making process.

In our first main result (Theorem \ref{theorem:char_dpo}), we show that an allocation is PO if and only if it is the solution to a series of recursive (backward in time) optimization problems. When the problem is static, we recover the classical result that PO allocations optimize a given social welfare function. In Section \ref{sec:cdpo}, we extend the classical notion of comonotonicity to the dynamic setting and derive a comonotone improvement theorem for dynamic risk measures that preserve the convex order. We further study Pareto optima in the comonotone market, and we show that each PO allocation in the comonotone market is also PO in the unconstrained market, therefore providing justification for studying the comonotone market. Our second main result (Theorem \ref{theorem:char_cdpo2}) provides a crisp characterization of comonotone Pareto optima when agent preferences are coherent and satisfy a property that we call \emph{equidistribution-preserving}. This allows for a useful algorithmic implementation, which we demonstrate for a two-period setting in Section \ref{sec:examples}. Our numerical example highlights the differences between dynamic and static models, and it demonstrates how future risks affect allocations in earlier periods.

The rest of this paper is structured as follows. Section \ref{sec:prelim} introduces the necessary notation and recalls some relevant properties of dynamic risk measures. Section \ref{sec:po_dynamic} presents the dynamic optimal-allocation problem and defines a notion of Pareto optimality for this framework. Section \ref{sec:char_dpo} characterizes Pareto optima, and our main characterization of comonotone Pareto optima is provided in Section \ref{sec:cdpo}. A numerical example is provided in Section \ref{sec:num_ex}, and Section \ref{sec:conclusion} concludes. Some background material is presented in Appendix \ref{appendix:single_period}.

\bigskip

\section{Notation and Preliminaries} \label{sec:prelim}

\medskip

\subsection{Notation} \label{sec:notation}

We consider a finite (known) time horizon $T\in\N^+$ and a filtered non-atomic probability space $(\Omega, \F, \{\F_t\}_{t\in\{0, 1, \ldots, T\}}, \P)$. For each $t\in\{0, 1, \ldots, T\}$, define $\L_t^{\infty}:=L^{\infty}(\Omega, \F_t, \P)$ to be the set of essentially bounded and $\F_t$-measurable random variables, and denote by $\L_{t+}^{\infty}\subseteq L^{\infty}(\Omega, \F_t, \P)$ the subset of non-negative random variables. Furthermore, let $\L_{0,T}^{\infty}:=\L_0^{\infty}\times \L_1^{\infty}\times \cdots \times \L_T^{\infty}$ be the space of essentially bounded discrete-time stochastic processes, and denote by $X_{0:T}:=\{X_0, X_1, \ldots, X_T\}\in \L_{0,T}^{\infty}$ a real-valued stochastic process. For $0\leq t<s\leq T$, we use the notation $\L_{t,s}^{\infty}:=0\times \cdots \times 0\times \L_t^{\infty}\times \cdots \times \L_s^{\infty}\times 0 \times \cdots \times 0$ and denote by $X_{t:s}:=\{0, \ldots, 0, X_t, \ldots, X_s, 0, \ldots, 0\}\in \L_{t,s}^{\infty}$ the projection of $X_{0:T}$ onto the space $\L_{t,s}^{\infty}$. Similarly, we denote the space of $n$-dimensional stochastic processes by $\bL_{0,T}^{\infty}$. We further denote random vectors with capital boldface letters, real-valued vectors with lowercase boldface letters, and use $\X_{0:T}$ for vector valued $n$-dimensional stochastic processes, where for each $t\in\{0, 1, \ldots, T\}$, $\X_t:=(X_t^{(1)}, \ldots, X_t^{(n)})$ with $X_t^{(i)} \in \L_t^\infty$ for all $i \in\mN:=\{1, \ldots, n\}$. Throughout, (in)equalities between random vectors and stochastic processes are to be understood component-wise and in a $\P$-almost sure sense. Finally, we denote by $\eqd$ equality in distribution.

\medskip

\subsection{Risk Measures} \label{sec:rm}

We assume that each agent has a preference given by a dynamic risk measure. Before defining a dynamic risk measure, we first recall the definition of a conditional risk measure. 

\medskip

\begin{definition} [Conditional Risk Measure] \label{definition:cond_rm} 
A conditional risk measure is a functional $\rho_{t,s} \colon  \L_{t+1,s}^{\infty}\to\L_{t}^{\infty}$, where $t\in\mfT:=\{0, 1, \ldots, T-1\}$ and $t< s\leq T$. 
\end{definition}

\smallskip

A conditional risk measure may satisfy some of the following properties. These properties were first introduced in \citet{artzner1999coherent, kusuoka2001law, follmer2002convex} in the context of static risk measures.

\medskip

\begin{definition} [Properties of Conditional Risk Measures] \label{definition:properties}

For $t\in\mfT$ and $t<s\leq T$, let $Y_{t+1:s}$, $\Z_{t+1:s}\in\L^{\infty}_{t+1,s}$. A conditional risk measure $\rho_{t:s}$ is:

\begin{enumerate}[label = $(\roman*)$]
\item Monotone if $\rho_{t,s}(Y_{t+1:s})\leq \rho_{t,s}(Z_{t+1:s})$ whenever $Y_{t+1:s}\leq Z_{t+1:s}$. If, in addition, $Y_{t+1:s}\leq Z_{t+1:s}$ and $\P(Y_{t+1:s}<Z_{t+1:s})>0$ implies $\P\big(\rho_{t,s}(Y_{t+1:s})<\rho_{t,s}(Z_{t+1:s})\big)>0$, then $\rho_{t,s}$ is strictly monotone. 
\item Positive homogeneous if for all $c\in\L_{t+}^{\infty}$ and $Y_{t+1:s}\in\L^{\infty}_{t+1,s}$, $\rho_{t,s}(c\, Y_{t+1:s})=c\, \rho_{t,s}(Y_{t+1:s})$. 
\item Subadditive if for all $Y_{t+1:s}, \, Z_{t+1:s}\in\L^{\infty}_{t+1,s}$, $\rho_{t,s}(Y_{t+1:s}+Z_{t+1:s})\leq \rho_{t,s}(Y_{t+1:s})+\rho_{t,s}(Z_{t+1:s})$. 
\item Translation invariant if for all $Y_{t:s}\in\L_{t,s}^{\infty}$, $\rho_{t,s}\big((Y_{t} + Y_{t+1}, Y_{t+2:s})\big)=Y_t+\rho_{t,s}(Y_{t+1:s})$.
\item Convex if for all $\lambda\in\L_t^{\infty}$ with $0\leq\lambda\leq 1$ and $Y_{t+1:s}, \, Z_{t+1:s}\in\L^{\infty}_{t+1,s}$, $\rho_{t,s}\big(\lambda\,  Y_{t+1:s}+(1-\lambda) \, Z_{t+1:s}\big)\leq\lambda \, \rho_{t,s}(Y_{t+1:s})+(1-\lambda) \rho_{t,s}(Z_{t+1:s})$. 
\item Normalized if $\rho_{t,s}(\0)=0$, where $\0\in\R^{s-t}$ is the zero vector. 
\item Equidistribution-preserving $\rho_{t,s}(Y_{t+1:s})\, \eqd \rho_{t,s}(Z_{t+1:s})$ whenever $Y_{t+1:s}\eqd Z_{t+1:s}$. 
\end{enumerate}
\end{definition}

When $s=t+1$, we omit the second subscript of the conditional risk measure, i.e. $\rho_{t, t+1}:= \rho_t$, and call $\rho_t$ a one-step conditional risk measure. For $t=0$ and $s=1$, the risk measure $\rho_0$ is a static risk measure, and the property of equidistribution-preserving reduces to the classical notion of law invariance. Moreover, a conditional risk measure is coherent if it is monotone, positive homogeneous, subadditive, and translation invariant. Next, we recall the definition of a dynamic risk measure. 

\medskip

\begin{definition} [Dynamic Risk Measure] \label{definition:dynamic_rm} A dynamic risk measure on $\mfT$ is a sequence of conditional risk measures $\{\rho_{t,T}\}_{t \in \mfT}$. 
\end{definition}

A dynamic risk measure satisfies one of the properties in Definition \ref{definition:properties} if $\rho_{t,T}$ satisfies the property for all $t\in\mfT$. Next, we recall the notion of strong time consistency (which we henceforth refer to as just time consistency) and the recursive relation for time-consistent dynamic risk measures \citep{cheridito2006dynamic, ruszczynski2010risk}. Time consistency means that if one stochastic process has a larger risk than another at a future time $\theta$ and the two processes are almost surely equal at all times $t$ satisfying $\tau+1\leq u\leq \theta$, then the former process should also have a larger risk at time $\tau$. This ensures that in this case the timing of risk evaluation does not change the ordering of the risk assessment.

\medskip

\begin{definition} [Time Consistency] \label{definition:strong_tc} A dynamic risk measure $\{\rho_{t,T}\}_{t\in\mfT}$ is time consistent if for all $0\leq\tau<\theta<T$ and $X_{\tau+1:T}, \, Y_{\tau+1:T}\in\L_{\tau+1,T}^{\infty}$, the conditions
\smallskip
\begin{equation*} 
 X_t=Y_t, \, \text{ for all } \ \tau+1\leq t\leq \theta,
 \ \ \text{ and } \ \ 
 \rho_{\theta, T}(X_{\theta+1,T})\leq \rho_{\theta, T}(Y_{\theta+1,T})
\end{equation*} 
\noindent imply that 
\begin{equation*}
\rho_{\tau, T}(X_{\tau+1,T})\leq \rho_{\tau, T}(Y_{\tau+1,T}). 
\end{equation*}
\end{definition}

It is well known that time-consistent dynamic risk measures admit the following recursive representation. A proof of this result can be found in \citet{cheridito2006dynamic} and \citet{ruszczynski2010risk}.

\medskip

\begin{theorem} [Recursive Relation] \label{theorem:recursion} Let $\{\rho_{t,T}\}_{t\in\mfT}$ be a normalized, monotone, and translation-invariant dynamic risk measure. Then, $\{\rho_{t,T}\}_{t\in\mfT}$ is time consistent if and only if there exists a family of normalized, monotone, and translation-invariant one-step conditional risk measures $\{\rho_t\}_{t\in\mfT}$ such that for all $t\in\mfT$ and $X_{t+1:T}\in\L_{t+1,T}^{\infty}$,
\begin{equation}
\label{eq:drm-recursion}
\rho_{t,T}(X_{t+1:T}) =\rho_t\Big(X_{t+1}+\rho_{t+1}\big(X_{t+2}+\cdots+
\rho_{T-1}(X_T)\big)\Big). 
\end{equation}
\end{theorem}

Theorem \ref{theorem:recursion} allows us to write any conditional risk measure as a sequence of one-step conditional risk measures. Moreover, any sequence of one-step conditional risk measures gives rise to a time-consistent dynamic risk measure through \eqref{eq:drm-recursion}. This recursive structure is essential to construct optimal allocation processes recursively backwards in time. 

\bigskip

\section{Pareto Optimality in a Dynamic Setting} \label{sec:po_dynamic}

In this section, we extend the static notions of individual rationality and Pareto optimality to the dynamic setting. We first describe the problem in Section \ref{sec:problem_dynamic} and define two different notions of Pareto optimality in Section \ref{sec:myopic_po} and Section \ref{sec:dpo}. The first notion, which we call myopic Pareto optimality, is motivated by the existing literature on recursive preferences, and only considers the risk of the process at time $0$. We discuss the shortfalls of this definition when agent preferences are translation invariant, and propose a second notion of optimality that is compatible with translation-invariant and time-consistent dynamic risk measures in Section \ref{sec:dpo}.  

\medskip

\subsection{Problem Setup} \label{sec:problem_dynamic}

We consider a problem with $n\in\N^+$ agents, and for each $i\in\mN$, we assume that the $i$-th agent has initial endowment $X_{1:T}^{(i)}:=\{X_1^{(i)}, \ldots, X_T^{(i)}\}\in\L_{1,T}^{\infty}$. Furthermore, for each $t\in\T:=\{1, \ldots T\}$, we denote by $\X_t:=(X_t^{(1)}, \ldots, X_t^{(n)})$ the vector of all agents' endowments at time $t$ and assume that $T\geq 2$ unless otherwise stated. The case of $T=1$ (i.e., the static problem), is well studied in the literature and we refer to Appendix \ref{appendix:single_period} for a collection of key results from the static case. Additionally, the information available to the agents is encapsulated in the filtration $\{\F_t\}_{t\in\{0, 1, \ldots, T\}}$, and we assume that $\F_0=\{\emptyset, \Omega\}$ is the trivial sigma algebra. At each time $t\in\T$, the $n$ agents pool together their endowments to create the so-called aggregate endowment $S_t:=\sum_{i\in\mN}X_t^{(i)}\in\L_t^{\infty}$, and we denote by $S_{1:T}:=\{S_1, \ldots, S_T\}$ the aggregate endowment process (or aggregate risk process in a risk-sharing problem). Under this setting, we have the following definition of an allocation. 

\medskip

\begin{definition} [Allocation - Dynamic] \label{definition:allocation_dyn}
A $n$-dimensional process $\Y_{1:T}\in\bL_{1,T}^{\infty}$ is an allocation of $S_{1:T}$ if 
\begin{equation} \label{eq:allocation}
    \sum_{i \in \mN} Y_t^{(i)} = S_t, \ \text{for all} \ t\in\T\,.
\end{equation}

\medskip

\noindent We denote by $\A_{1:T}$ the set of allocation processes. Moreover, if the (sub)-process $\Y_{r:s}$ satisfies \eqref{eq:allocation} for all $t\in\{r, \ldots, s\}$, then we say that $\Y_{r:s}$ is an allocation of $S_{r:s}$ and denote by $\A_{r:s}$ the set of allocations of $S_{r:s}$.  
\end{definition}

Note that for each $i\in\mN$ and $t\in\T$, the random variable $Y_t^{(i)}$ represents the portion of the time-$t$ aggregate endowment that is allocated to the $i$-th agent. The agents wish to choose $\Y_{1:T}$ based on individual preferences characterized by time-consistent dynamic risk measures with representation \eqref{eq:drm-recursion}, that is, for $i\in\mN$, agent-$i$ uses a dynamic risk measure $\{\rho_{t,T}^{(i)}\}_{t\in\mfT}$. Then, by Theorem \ref{theorem:recursion}, the $i$-th agent's preference at time $t$ is given by a normalized, monotone, and translation-invariant one-step conditional risk measure $\rho_t^{(i)}$ for all $t\in\mfT$ and $i\in\mN$. Consequently, we make the following assumption all throughout.

\medskip

\begin{assumption} \label{assumption:tc} For $i\in\mN$, the $i$-th agent's risk preference is given by a normalized, monotone, translation-invariant, and time-consistent dynamic risk measure $\{\rho_{t,T}^{(i)}\}_{t\in\mfT}$. 
\end{assumption} 

\medskip

\subsection{Myopic Pareto Optimality} \label{sec:myopic_po}

In this section, we define a notion of optimality that only considers the time $0$ conditional risk measure of each agent. We refer to this notion of optimality as ``myopic'' since the agents do not care about risk valuation at any future time point. 

\medskip

\begin{definition} [Myopic Individual Rationality] \label{definition:myopic_ir} An allocation $\Y_{1:T}\in\A_{1:T}$ is myopic individually rational (IR) if 
\begin{equation} \label{eq:myopic_ir}
    \rho_{0,T}^{(i)}(Y_{1:T}^{(i)})
        \leq
    \rho_{0,T}^{(i)}(X_{1:T}^{(i)})\,,
    \ \text{for all} \ i\in\mN\,.
\end{equation}

\smallskip

\noindent We denote by $\MIR$ the set of all myopic IR allocations.
\end{definition}

Condition \eqref{eq:myopic_ir} states that agents will not be worse off with respect to their conditional risk measure $\rho_{0,T}^{(i)}$ after participating in the market. Thus, agents have an incentive to participate in reallocation if and only if the optimal allocation is myopic IR. Furthermore, note that $\MIR\neq\emptyset$ since $\X_{0:T}\in\MIR$. Next, we define a notion of Pareto optimality for this myopic setting, which we call myopic Pareto optimality. Myopic Pareto optimality means that an agent cannot deviate from the optimal allocation and be strictly better off without making another agent strictly worse off.

\medskip

\begin{definition} [Myopic Pareto Optimality] \label{definition:myopic_po} An allocation $\Y_{1:T}^*\in\A_{1:T}$ is myopic Pareto optimal (PO) if 
\begin{enumerate}[label = $(\roman*)$]
    \item $\Y_{1:T}^*\in\MIR$
    \item there does not exist an allocation process $\Y_{1:T}\in\MIR$ that satisfies $\rho_{0,T}^{(i)}(Y_{1:T}^{(i)})\leq \rho_{0,T}^{(i)}(Y_{1:T}^{*(i)})$ for all $i\in\mN$, and at least one inequality is strict. 
\end{enumerate}

\smallskip

\noindent We denote by $\MPO$ the set of all myopic PO allocations. 
\end{definition}

Note that myopic individual rationality and myopic Pareto optimality reduce to the classical notions of individual rationality and Pareto optimality when $T=1$. Furthermore, Assumption \ref{assumption:tc} allows us to find myopic Pareto optima by finding Pareto optima in a related static problem instead. For each $i\in\mN$, 
\begin{align*}
    \rho_{0,T}^{(i)}(Y_{1:T}^{(i)})
    &= \rho_0^{(i)}\Big(Y_{1}^{(i)}+\rho_{1}^{(i)}\left(Y_{2}^{(i)}+\cdots+\rho_{T-1}^{(i)}(Y_{T}^{(i)})\right)\Big) 
    \\
    &= \rho_{0}^{(i)}\left(\rho_{1}^{(i)}\Big(\cdots\rho_{T-1}^{(i)}\big(\sum_{t=1}^{T}Y_t^{(i)}\big)\Big)\right)
    =: \brho^{(i)}\left(\sum_{t=1}^{T}Y_t^{(i)}\right)\,.
\end{align*}

\noindent We observe from the above equation that by translation invariance, the $i$-th agent's risk evaluation of the process $Y_{1:T}^{(i)}$ is equivalent to evaluating the risk of $\sum_{t=1}^{T}Y_t^{(i)}$ with a real-valued risk measure $\brho^{(i)}\colon \L_T^\infty\to \R$. We denote by $\bar{Y}^{(i)}:=\sum_{t=1}^{T}Y_t^{(i)}$ the total allocation of the $i$-th agent and $\bar{X}^{(i)}:=\sum_{t=1}^{T}X_t^{(i)}$ the initial total endowment of the $i$-th agent. Then, \eqref{eq:myopic_ir} holds if and only if
\begin{equation} \label{eq:myopic_ir_3}
    \brho^{(i)}(\bar{Y}^{(i)})
        \leq
    \brho^{(i)}(\bar{X}^{(i)})\,,
    \ \text{for all} \ i\in\mN\,.
\end{equation}

\noindent In other words, under Assumption \ref{assumption:tc}, a process $\Y_{1:T}$ is myopic IR if and only if the associated total allocation of the agents $\bar{\Y}$ is IR (in the static sense) when the $i$-th agent uses risk measure $\brho^{(i)}$ and has initial endowment $\bar{X}^{(i)}$. Furthermore, myopic IR (and hence myopic PO) is a condition on the total allocation across all time periods $\bar{\Y}$ and not the individual allocations $Y_t^{(i)}$. We summarize the above discussion with the following proposition.

\medskip

\begin{proposition} \label{prop:mpo} Let $\Y_{1:T}\in\A_{1:T}$ and $\bar{Y}^{(i)} := \sum_{t=1}^{T} Y_t^{(i)}$ for $i\in\mN$. Then, the following are equivalent:
\begin{enumerate}[label = $(\roman*)$]
    \item $\bar{\Y}$ is PO (in the static sense) when for all $i\in\mN$ the $i$-th agent uses risk measure $\brho^{(i)}$ and has initial endowment $\sum_{t=1}^{T} X_t^{(i)}$.
    \item $\Y_{1:T}\in\MPO$.  
\end{enumerate}
\end{proposition}

Proposition \ref{prop:mpo} implies that we can find myopic Pareto optima using tools from the static case. In particular, by combining Proposition \ref{prop:mpo} with the well-known (static) inf-convolution result (see Proposition 3.3 from \citet{ghossoub2026efficiency} or Proposition \ref{prop:static} $(i)$ in Appendix \ref{appendix:single_period}), we obtain that an allocation $\Y_{1:T}$ is myopic PO if and only if the associated total allocation vector of the agents $\bar{\Y}$ attains the infimum
\begin{equation} \label{eq:myopic_po}
\underset{\Z\in\tilde{\IR}} \inf \
    \sum_{i\in\mN}\brho^{(i)}(Z^{(i)}),
\end{equation}
\noindent where $\tilde{\IR}$ is the set of total allocation vectors such that \eqref{eq:myopic_ir_3} holds.

Furthermore, observe that the representation \eqref{eq:myopic_po} is equivalent to the classical notion of $\alpha$-efficiency for recursive utilities with $\alpha=(1, \ldots, 1)$, where here the recursive utility is replaced by $\brho$. The assumption of translation invariance, which is typically not used in the recursive utility literature, has a significant impact on the set of optimal allocations. In particular, given a myopic PO allocation $\Y_{1:T}^*$, translation invariance implies that any allocation $\Y_{1:T}'$ satisfying $\sum_{t=1}^{T}Y_t'^{(i)}=\sum_{t=1}^{T}Y_t^{*(i)}$ for all $i\in\mN$ is also myopic PO. In other words, myopic Pareto optimality is a condition on the total allocation for each agent over time and never defines a unique stochastic process $\Y_{1:T}$. Thus, it is necessary to consider a different notion of optimality for translation-invariant dynamic preferences. 

\medskip

\subsection{Dynamic Pareto Optimality} \label{sec:dpo}

In this section, we remove the assumption that agents only care about evaluating the risk of an allocation process at time $0$ and assume that each agent is concerned with their future risk $\rho_{t,T}^{(i)}(Y_{t+1:T}^{(i)})$ at all time points $t\in\mfT$. This leads to a different notion of individual rationality and Pareto optimality, which we call dynamic individual rationality and dynamic Pareto optimality. For simplicity, we drop the word dynamic whenever there is no confusion and simply use the terms IR and PO. 

\medskip

\begin{definition} [Dynamic Individual Rationality] \label{definition:dir} An allocation $\Y_{1:T}\in\A_{1:T}$ is dynamic IR if for all $t\in\mfT$
\begin{equation*}
    \rho_{t,T}^{(i)}(Y_{t+1:T}^{(i)})
        \leq
    \rho_{t,T}^{(i)}(X_{t+1:T}^{(i)})\,,
    \ \text{for all} \ i\in\mN\,.
\end{equation*}

\noindent We denote by $\DIR$ the set of all dynamic IR allocations. 
\end{definition}

\medskip

\begin{definition} [Dynamic Pareto Optimality] \label{definition:dpo} An allocation $\Y_{1:T}^*\in\A_{1:T}$ is dynamic PO if 
\begin{enumerate}[label = $(\roman*)$]
    \item $\Y_{1:T}^*\in\DIR$
    \item there does not exist an allocation process $\Y_{1:T}\in\DIR$ such that for all $i\in\mN$ and $t\in\mfT$,   
    \[
    \rho_{t,T}^{(i)}(Y_{t+1:T}^{(i)})\leq \rho_{t,T}^{(i)}(Y_{t+1:T}^{*(i)}),
    \]   
    and $\P\left(\rho_{t,T}^{(i)}(Y_{t+1:T}^{(i)})<\rho_{t,T}^{(i)}(Y_{t+1:T}^{*(i)})\right)>0$ for some $i\in\mN$ and $t\in\mfT$. 
\end{enumerate}

\smallskip

\noindent We denote by $\DPO$ the set of all dynamic PO allocations. 
\end{definition}

Unlike myopic Pareto optimality, we assume here that agents care about risk evaluation at intermediate time steps. Hence, even when agent preferences are translation invariant, it does not suffice to study the sum of each agent's allocation over time. To characterize PO allocations, we first define the notion of Pareto optimality at a single time point $t\in\mfT$. Then, we show that an allocation is PO if and only if it is PO at time $t$ for all $t\in\mfT$. This allows us to characterize PO allocations via backwards recursion. To simplify the notation in our recursive arguments, we introduce the risk-to-go process. 

\medskip

\begin{definition} [Risk-to-Go] \label{definition:risk_to_go} The risk-to-go process associated with $\{\rho_{t,T}\}_{t\in\mfT}$, denoted $R_{0:T}\in \L_{0,T}^\infty$, is defined recursively as follows:

\begin{enumerate}[label = $(\roman*)$]
    \item $R_T=0$.
    \item For $t\in\mfT$, we define $R_t(Y_{t+1:T}):=
    \rho_t\big(Y_{t+1}+R_{t+1}(Y_{t+2:T})\big)$.
\end{enumerate}
\end{definition}

\medskip

\begin{definition} [Individual Rationality at Time $t$] \label{definition:ir_t} For $t\in\mfT$, an allocation $\Y_{t+1:T}\in\A_{t+1:T}$ is IR at time $t$ if for all $\in\mN$,
\[
R_{t}^{(i)}(Y_{t+1:T}^{(i)})\leq R_{t}^{(i)}(X_{t+1:T}^{(i)}).
\] 

\smallskip

\noindent Furthermore, we denote by $\IR_t$ the set of all processes $\Y_{t+1:T}$ that are IR at time $t$, i.e., 

\[
\IR_t:=\left\{\Y_{t+1:T}~|~R_{t}^{(i)}(Y_{t+1:T}^{(i)})\leq R_{t}^{(i)}(X_{t+1:T}^{(i)})\ \text{for all}\ i\in\mN\right\}.
\]
\end{definition}

It follows from the definition that $\Y_{1:T}\in\DIR$ if and only if $\Y_{t+1:T}\in\IR_t$ for all $t\in\mfT$. Moreover, $\IR_0=\MIR$ and all notions of individual rationality in this paper are equivalent to the static notion of individual rationality when $T=1$. Next, we define Pareto optimality at time $t$ in a similar manner.

\medskip

\begin{definition} [Pareto Optimality at Time $t$] \label{definition:po_t} For $t\in\mfT$, an allocation $\Y_{t+1:T}^*\in\A_{t+1:T}$ is PO at time $t$ if
\begin{enumerate}[label = $(\roman*)$]
    \item $\Y^*_{t+1:T}\in\IR_t$
    \item there does not exist a random vector $\Y_{t+1}$ with $\{\Y_{t+1}, \Y_{t+2:T}^*\}\in\IR_t$ such that for all $i\in\mN$,   
    \[
    \rho_{t}^{(i)}\big(Y_{t+1}^{(i)}+R_{t+1}^{(i)}(Y_{t+2:T}^{*(i)})\big)
        \leq 
    \rho_{t}^{(i)}\big(Y_{t+1}^{*(i)}+R_{t+1}^{(i)}(Y_{t+2:T}^{*(i)})\big),
    \]   
    
    and $\P\left(\rho_{t}^{(i)}\big(Y_{t+1}^{(i)}+R_{t+1}^{(i)}(Y_{t+2:T}^{*(i)})\big)<\rho_{t}^{(i)}\big(Y_{t+1}^{*(i)}+R_{t+1}^{(i)}(Y_{t+2:T}^{*(i)})\right)>0$ for some $i\in\mN$.
\end{enumerate}

\smallskip

\noindent We denote by $\PO_t$ the set of all processes $\Y_{t+1:T}$ that are PO at time $t$. 
\end{definition}

We use the convention that $\{\Y_T, \Y_{T+1:T}^*\}=\Y_T$. In other words, an allocation $\Y_T^*\in\IR_{T-1}$ is PO at time $t$ if there does not exist a random vector $\Y_T\in\IR_{T-1}$ such that for all $i\in\mN$, $R_{T-1}^{(i)}(Y_T^{(i)}) \leq R_{T-1}^{(i)}(Y_{T}^{*(i)}) 
\ \text{and} \ 
\P\left(R_{T-1}^{(i)}(Y_T^{(i)})<R_{T-1}^{(i)}(Y_{T}^{*(i)})\right)>0 
\ \text{for some} \ i\in\mN.$ 
Our next result links $\PO_t$ with $\DPO$. In particular, an allocation $\Y_{1:T}$ is PO if and only if each sub-allocation $\Y_{t+1:T}$ is PO at time $t$ for all $t\in\mfT$. Hence, we can characterize PO allocations by characterizing allocations that are PO at time $t$ instead. 

\medskip

\begin{lemma}\label{lemma:dpo} $\Y_{1:T}^*\in\DPO$ if and only if $\Y_{t+1:T}^*\in\PO_t$ for all $t\in\mfT$. 
\end{lemma}


\begin{proof}
We first prove that $\Y_{t+1:T}^*\in\PO_t$ for all $t\in\mfT$ implies $\Y_{1:T}^*\in\DPO$. Suppose $\Y_{t+1:T}^*\in\PO_t$ for all $t\in\mfT$. Then, $\Y_{t+1:T}^*\in\IR_t$ for all $t\in\mfT$, which implies that $\Y_{1:T}^*\in\DIR$. Hence, $\Y_{1:T}^*$ satisfies the first condition of Pareto optimality and it remains to check the second condition. Let $\Y_{1:T}\in\DIR$ such that for all $t\in\mfT$ and $i\in\mN$,
\begin{equation} \label{eq:lemma_dpo_1}
R_{t}^{(i)}(Y_{t+1:T}^{(i)})\leq R_{t}^{(i)}(Y_{t+1:T}^{*(i)}).
\end{equation}

\noindent To verify the second condition of Pareto optimality, it suffices to show that for all $i\in\mN$ and $s\in\mfT$,
\begin{equation} \label{eq:lemma_dpo_2}
R_{s}^{(i)}(Y_{s+1:T}^{(i)})=R_{s}^{(i)}(Y_{s+1:T}^{*(i)}). 
\end{equation}

\noindent We prove \eqref{eq:lemma_dpo_2} by induction. For $s=T-1$, we have $\Y_T\in\IR_{T-1}$ since $\Y_{1:T}\in\DIR$. Hence, the assumption that $\Y_T^*\in\PO_{T-1}$ implies that at least one of the following holds:
\begin{align}
&\P\left(R_{T-1}^{(i)}(Y_{T}^{(i)})\leq R_{T-1}^{(i)}(Y_{T}^{*(i)})\right)
<1 \ \text{for some} \ i\in\mN,
\label{eq:lemma_dpo_3} \\
&\P\left(R_{T-1}^{(i)}(Y_{T}^{(i)})<R_{T-1}^{(i)}(Y_{T}^{*(i)})\right)=0 \ \text{for all} \ i\in\mN.
\label{eq:lemma_dpo_4}
\end{align}

\noindent Since \eqref{eq:lemma_dpo_3} contradicts \eqref{eq:lemma_dpo_1}, it must be the case that \eqref{eq:lemma_dpo_4} holds. Combining \eqref{eq:lemma_dpo_1} and \eqref{eq:lemma_dpo_4} yields \eqref{eq:lemma_dpo_2} for time $T-1$. Next, we proceed by induction backwards in time. For this, assume that \eqref{eq:lemma_dpo_2} holds for all $s\geq t$. It suffices to show that \eqref{eq:lemma_dpo_2} also holds for $t-1$. Note that for all $i\in\mN$, we have 
\begin{align*}
\rho_{t-1}^{(i)}\big(Y_{t}^{(i)}+R_{t}^{(i)}(Y_{t+1:T}^{*(i)})\big) &= \rho_{t-1}^{(i)}\big(Y_{t}^{(i)}+R_{t}^{(i)}(Y_{t+1:T}^{(i)})\big) \\
&\leq \rho_{t-1}^{(i)}\big(Y_{t}^{*(i)}+R_{t}^{(i)}(Y_{t+1:T}^{*(i)})\big) \\
&\leq \rho_{t-1}^{(i)}\big(X_{t}^{(i)}+R_{t}^{(i)}(X_{t+1:T}^{(i)})\big),
\end{align*}

\noindent where the first equality holds by the induction hypothesis, the first inequality follows from \eqref{eq:lemma_dpo_1}, and the last inequality follows from the fact that $\Y_{t:T}^*\in\IR_{t-1}$. In other words, $\{\Y_{t}, \Y_{t+1:T}^*\}\in\IR_{t-1}$. Hence, the assumption that $\Y_{t:T}^*\in\PO_{t-1}$ implies that at least one of the following holds:
\begin{align}
&\P\left(\rho_{t-1}^{(i)}
    \big(Y_{t}^{(i)}+R_{t}^{(i)}(Y_{t+1:T}^{*(i)})\big)
        \leq\rho_{t-1}^{(i)}
    \big(Y_{t}^{*(i)}+R_{t}^{(i)}(Y_{t+1:T}^{*(i)})\big)\right)
<1 \ \text{for some} \ i\in\mN,
\label{eq:lemma_dpo_5} \\
&\P\left(\rho_{t-1}^{(i)}
    \big(Y_{t}^{(i)}+R_{t}^{(i)}(Y_{t-1:T}^{*(i)})\big)
        <\rho_{t-1}^{(i)}
    \big(Y_{t}^{*(i)}+R_{t}^{(i)}(Y_{t+1:T}^{*(i)})\big)\right)
=0 \ \text{for all} \ i\in\mN.
\label{eq:lemma_dpo_6}
\end{align}

\noindent Since \eqref{eq:lemma_dpo_2} holds for $s=t$, \eqref{eq:lemma_dpo_5} is equivalent to 
\[
\P\left(\rho_{t-1}^{(i)}
    \big(Y_{t}^{(i)}+R_{t}^{(i)}(Y_{t+1:T}^{(i)})\big)
        \leq \rho_{t-1}^{(i)}
    \big(Y_{t}^{*(i)}+R_{t}^{(i)}(Y_{t+1:T}^{*(i)})\big)
\right)<1 \ \text{for some} \ i\in\mN,
\] 

\noindent which contradicts \eqref{eq:lemma_dpo_1}. Therefore, it must be the case that \eqref{eq:lemma_dpo_6} holds. Since \eqref{eq:lemma_dpo_2} holds for $s=t$, \eqref{eq:lemma_dpo_6} is equivalent to 
\begin{equation}\label{eq:lemma_dpo}
\P\left(\rho_{t-1}^{(i)}
    \big(Y_{t}^{(i)}+R_{t}^{(i)}(Y_{t+1:T}^{(i)})\big)
        <\rho_{t-1}^{(i)}
    \big(Y_{t}^{*(i)}+R_{t}^{(i)}(Y_{t+1:T}^{*(i)})\big)
\right)=0 \ \text{for all} \ i\in\mN.
\end{equation}

\noindent Combining \eqref{eq:lemma_dpo_1} and \eqref{eq:lemma_dpo} yields \eqref{eq:lemma_dpo_2} for $s=t$.

\smallskip

To prove the converse, let $\Y_{1:T}^*\in\DPO$. Assume, by way of contradiction, that $\Y_{t+1:T}^*\notin\PO_{t}$ for some $t\in\mfT$. Since $\Y_{1:T}^*\in\DPO$, this implies that $\Y_{t+1:T}^*\in\IR_t$ for all $t\in\mfT$. Therefore, $\Y_{t+1:T}^*\notin\PO_t$ implies that there exists $\Y_{t+1}$ with $\{\Y_{t+1}, \Y_{t+2:T}^*\}\in\IR_t$ such that
\begin{align}
&\P\left(\rho_{t}^{(i)}
    \big(Y_{t+1}^{(i)}+R_{t+1}^{(i)}(Y_{t+2:T}^{*(i)})\big)
        \leq \rho_{t}^{(i)}
    \big(Y_{t+1}^{*(i)}+R_{t+1}^{(i)}(Y_{t+2:T}^{*(i)})\big)
\right)=1 \ \text{for all} \ i\in\mN \ \text{and},
\label{eq:lemma_dpo_7} \\
&\P\left(\rho_{t}^{(i)}
    \big(Y_{t+1}^{(i)}+R_{t+1}^{(i)}(Y_{t+2:T}^{*(i)})\big) 
         < \rho_{t}^{(i)}
    \big(Y_{t+1}^{*(i)}+R_{t+1}^{(i)}(Y_{t+2:T}^{*(i)})\big)
\right)>0 \ \text{for some} \ i\in\mN.
\label{eq:lemma_dpo_8} 
\end{align}

\noindent Let $\tilde{\Y}_{1:T}$ be the process $\Y_{1:T}^*$ with the $(t+1)$-st vector replaced by $\Y_{t+1}$. Then, for all $s\in\mfT$ and $i\in\mN$, 
\begin{equation} \label{eq:lemma_dpo_9}
R_{s}^{(i)}(\tilde{Y}_{s+1:T}^{(i)}) \leq 
R_{s}^{(i)}(Y_{s+1:T}^{*(i)}) \leq 
R_{s}^{(i)}(X_{s+1:T}^{(i)}),
\end{equation}

\noindent where the first inequality follows from the monotonicity of $\{\rho_s^{(i)}\}_{s\in\mfT}$ and \eqref{eq:lemma_dpo_7}. Finally, for $s=t$ and $i\in\mN$ satisfying \eqref{eq:lemma_dpo_8}, the first inequality in \eqref{eq:lemma_dpo_9} is strict with positive probability, contradicting the assumption that $\Y_{1:T}^*\in\DPO$. Therefore, $\Y_{t+1:T}^*\in\PO_t$. 
\end{proof}

\bigskip

\section{Characterization of Dynamic Pareto Optimal Allocations} \label{sec:char_dpo} 

In this section, we generalize the well-known static inf-convolution result (\citet{embrechts2018quantile} Proposition 1, \citet{ghossoub2026efficiency} Proposition 3.3) to the dynamic setting. Then, we combine this result with Lemma \ref{lemma:dpo} to derive a recursive approach for finding PO allocations. We start with a lemma that gives an alternative representation for $\PO_t$. 

\medskip

\begin{lemma}\label{lemma:pot_rep} For all $t\in\mfT$,
\begin{align*}
\PO_t=\biggl\{                   
    \Z_{t+1:T}\biggr. &\in\IR_t~\big|~  \P\left(
        \sum_{i\in\mN}\rho_{t}^{(i)}
            \big(Y_{t+1}^{(i)}+R_{t+1}^{(i)}(Z_{t+2:T}^{(i)})\big)
                \leq
        \sum_{i\in\mN}R_{t}^{(i)}(Z_{t+1:T}^{(i)})
    \right)<1 \ \text{or} \\
    & \P\left(
        \sum_{i\in\mN}\rho_{t}^{(i)}
            \big(Y_{t+1}^{(i)}+R_{t+1}^{(i)}(Z_{t+2:T}^{(i)})\big)
                <
        \sum_{i\in\mN}R_{t}^{(i)}(Z_{t+1:T}^{(i)})
    \right)=0 \\
     & \biggl. \qquad\qquad\qquad \text{for all} \ \Y_{t+1}\neq \Z_{t+1} \ 
        \text{such that} \ \{\Y_{t+1}, \Z_{t+2:T}\}\in\IR_t
\biggr\}.
\end{align*}
\end{lemma}

\begin{proof}
Fix $t\in\mfT$ and let $\Gamma$ denote the set on the right-hand side of the equation in the lemma. 

\smallskip

We first show that $\Gamma\subseteq\PO_t$. Let $\Y_{t+1:T}\in\IR_t$ and assume that $\Y_{t+1:T}\notin\PO_t$. Then, by definition of $\PO_t$, there exists a random vector $\tilde{\Y}_{t+1}$ with $\{\tilde{\Y}_{t+1}, \Y_{t+2:T}\}\in\IR_t$ such that
\begin{align*}
&\P\left(\rho_{t}^{(i)}
    \big(\tilde{Y}_{t+1}^{(i)}+R_{t+1}^{(i)}(Y_{t+2:T}^{(i)})\big)
        \leq \rho_{t}^{(i)}
    \big(Y_{t+1}^{(i)}+R_{t+1}^{(i)}(Y_{t+2:T}^{(i)})\big)
\right)=1 \ \text{for all} \ i\in\mN \ \text{and}, \\
&\P\left(\rho_{t}^{(i)}
    \big(\tilde{Y}_{t+1}^{(i)}+R_{t+1}^{(i)}(Y_{t+2:T}^{(i)})\big)
        < \rho_{t}^{(i)}
    \big(Y_{t+1}^{(i)}+R_{t+1}^{(i)}(Y_{t+2:T}^{(i)})\big)
\right)>0 \ \text{for some} \ i\in\mN.
\end{align*}

\noindent Hence, neither condition of the set $\Gamma$ is satisfied, so $\Gamma\subseteq\PO_t$.

\smallskip

Next, we show that $\PO_t\subseteq \Gamma$. Let $\Y_{t+1:T}^*\in\PO_t$. Assume, by way of contradiction, that $\Y_{t+1:T}^*\notin \Gamma$. Then, by definition of $\Gamma$, there exists a random vector $\tilde{\Y}_{t+1}$ with $\{\tilde{\Y}_{t+1}, \Y_{t+2:T}^*\}\in\IR_t$ such that
\begin{align}
&\P\left(\sum_{i\in\mN}\rho_{t}^{(i)}
    \big(\tilde{Y}_{t+1}^{(i)}+R_{t+1}^{(i)}(Y_{t+2:T}^{*(i)})\big)
        \leq \sum_{i\in\mN}\rho_{t}^{(i)}
    \big(Y_{t+1}^{*(i)}+R_{t+1}^{(i)}(Y_{t+2:T}^{*(i)})\big)
\right)=1 \ \text{and} 
\label{eq:sum} \\
&\P\left(\sum_{i\in\mN}\rho_{t}^{(i)}
    \big(\tilde{Y}_{t+1}^{(i)}+R_{t+1}^{(i)}(Y_{t+2:T}^{*(i)})\big) 
        <\sum_{i\in\mN}\rho_{t}^{(i)}
    \big(Y_{t+1}^{*(i)}+R_{t+1}^{(i)}(Y_{t+2:T}^{*(i)})\big)
\right)>0.
\label{eq:sum_2}
\end{align}

\noindent For $i\in\mN$, let $v^{(i)}:=\rho_{t}^{(i)}\big(\tilde{Y}_{t+1}^{(i)}+R_{t+1}^{(i)}(Y_{t+2:T}^{*(i)})\big)-R_{t}^{(i)}(Y_{t+1:T}^{*(i)})\in\L_
{t}^{\infty}$ and $w^{(i)}\in\L_{t}^{\infty}$ such that

\begin{enumerate}[label = $(\roman*)$]
\item For all $i\in\mN$, $0\leq w^{(i)}\leq\max\{0, \, -v^{(i)}\}$.
\item There exists $i\in\mN$ such that $\P\big(w^{(i)}<\max\{0, \, -v^{(i)}\}\big)>0$. 
\item \begin{equation} \label{eq:v_i}
    \sum_{i\in\mN}v^{(i)}\Id_{v^{(i)}\geq 0}
        =
    \sum_{i\in\mN}w^{(i)}\Id_{v^{(i)}<0}. 
\end{equation}
\end{enumerate}

\noindent Note that random variables satisfying $(i)$ and $(ii)$ exist since \eqref{eq:sum_2} implies that $\P\big(-v^{(i)}>0\big)>0$ for some $i\in\mN$. Furthermore, \eqref{eq:sum} and \eqref{eq:sum_2} imply that there exists $\{w^{(i)}\}_{i\in\mN}$ that also satisfies $(iii)$. For $i\in\mN$, let 
\[
\epsilon^{(i)}:=w^{(i)}\Id_{v^{(i)}<0}-v^{(i)}\Id_{v^{(i)}\geq 0}\in\L_{t}^{\infty}, 
    \ 
\hat{Y}_{t+1}^{(i)}:=\tilde{Y}_{t+1}^{(i)}+\epsilon^{(i)}, 
    \ \text{and} \
\hat{\Y}_{t+2:T}=\Y^*_{t+2:T}.
\]

\noindent We first show that $\hat{\Y}_{t+1:T}\in\A_{t+1:T}$. Since $\Y_{t+1:T}^*\in\PO_t$, it follows that $\hat{\Y}_{t+2:T}=\Y^*_{t+2:T}\in\A_{t+2:T}$. Furthermore, by \eqref{eq:v_i} and the fact that $\{\tilde{\Y}_{t+1}, \Y_{t+2:T}^*\}\in\IR_t\subseteq\A_{t+1:T}$,
\[
\sum_{i\in\mN}\hat{Y}_{t+1}^{(i)}
    =
\sum_{i\in\mN}\left(\tilde{Y}_{t+1}^{(i)}+\epsilon^{(i)}\right)
    =
\sum_{i\in\mN}\tilde{Y}_{t+1}^{(i)}=S_{t+1},
\]

\noindent and therefore $\hat{\Y}_{t+1:T}\in\A_{t+1:T}$. Next, we show that $\hat{\Y}_{t+1:T}\in\IR_t$. Since $\Y_{t+1:T}^*\in\IR_t$, we have $R_{t}^{(i)}(Y_{t+1:T}^{*(i)})\leq R_{t}^{(i)}(X_{t+1:T}^{(i)})$ for all $i\in\mN$. Hence, to show that $\hat{\Y}_{t+1:T}\in\IR_t$, it suffices to show that
\[
\P\left(
    \rho_{t}^{(i)}\big(\hat{Y}_{t+1}^{(i)}+R_{t+1}^{(i)}(Y_{t+2:T}^{*(i)})\big)\leq R_{t}^{(i)}(Y_{t+1:T}^{*(i)})
\right)=1.
\] 

\noindent Indeed, for all $i\in\mN$, 
\begin{align*}
& \P\left(\rho_{t}^{(i)}
    \big(\hat{Y}_{t+1}^{(i)}+R_{t+1}^{(i)}(Y_{t+2:T}^{*(i)})\big)
        \leq R_{t}^{(i)}(Y_{t+1:T}^{*(i)})
    \right) \\
= &\ \P\left(\rho_{t}^{(i)}
    \big(\tilde{Y}_{t+1}^{(i)}+
        R_{t+1}^{(i)}(Y_{t+2:T}^{*(i)})\big)+\epsilon^{(i)}
            \leq 
        R_{t}^{(i)}(Y_{t+1:T}^{*(i)})
    \right) \\
= &\ \P\left(
        \left\{\rho_{t}^{(i)}
            \big(\tilde{Y}_{t+1}^{(i)}+R_{t+1}^{(i)}(Y_{t+2:T}^{*   
                (i)})\big)-v^{(i)}
            \leq R_{t}^{(i)}(Y_{t+1:T}^{*(i)})
        \right\} \cap 
        \left\{
            v^{(i)}\geq 0
        \right\}
    \right)\\
&\qquad\ + \P\left(
    \left\{\rho_{t}^{(i)}
        \big(\tilde{Y}_{t+1}^{(i)}+R_{t+1}^{(i)}(Y_{t+2:T}^{*
            (i)})\big)+w^{(i)}
        \leq R_{t}^{(i)}(Y_{t+1:T}^{*(i)})
    \right\}\cap 
    \left\{
        v^{(i)}<0
    \right\}
    \right) \\
= &\ \P\left(v^{(i)}\geq 0\right) + 
        \P\left(\{w^{(i)}\leq -v^{(i)}\} \cap \{v^{(i)}<0\}\right) \\ 
= &\ \P\left(v^{(i)}\geq 0\right) + \P\left(v^{(i)}<0\right) 
=  1,
\end{align*}

\noindent where the first equality follows from translation invariance, the third equality follows from the fact that $v^{(i)}=\rho_{t}^{(i)}\big(\tilde{Y}_{t+1}^{(i)}+R_{t+1}^{(i)}(Y_{t+2:T}^{*(i)})\big)-R_{t}^{(i)}(Y_{t+1:T}^{*(i)})$, and the fourth equality follows from $(i)$. Finally, for $i\in\mN$ satisfying $(ii)$, we have
\begin{align*}
\P\Big(\rho_{t}^{(i)}
        \big(\hat{Y}_{t+1}^{(i)}&+R_{t+1}^{(i)}(Y_{t+2:T}^{*(i)})\big)
    < R_{t}^{(i)}(Y_{t+1:T}^{*(i)})
\Big) \\
&=\ \P\left(\rho_{t}^{(i)}
    \big(\tilde{Y}_{t+1}^{(i)}+R_{t+1}^{(i)}(Y_{t+2:T}^{*(i)})\big)
        +\epsilon^{(i)}
    < R_{t}^{(i)}(Y_{t+1:T}^{*(i)})
\right) \\
&\geq\ \P\left(
    \left\{
        \rho_{t}^{(i)}\big(\tilde{Y}_{t+1}^{(i)} 
            + R_{t+1}^{(i)}(Y_{t+2:T}^{*(i)})\big)+w^{(i)}< 
        R_{t}^{(i)}(Y_{t+1:T}^{*(i)})
    \right\} \cap 
    \left\{
        v^{(i)}<0
    \right\}
\right)\\
&=\P\left(
    \left\{
        w^{(i)}< -v^{(i)}
    \right\} \cap 
    \left\{
        v^{(i)}<0
    \right\}
\right) 
> \ 0,
\end{align*}

\noindent where the first equality follows from translation invariance and the last inequality follows from $(ii)$. Hence, $\Y_{t+1:T}^*\notin\PO_t$, which is a contradiction.
\end{proof}

\medskip

The quantity $\sum_{i\in\mN}R_{t}^{(i)}(Z_{t+1:T}^{(i)})$ is the total risk-to-go for all agents at time $t$ for the allocation process $\Z_{t+1:T}$. Thus, Lemma \ref{lemma:pot_rep} states that $\Z_{t+1:T}\in\PO_t$ if and only if we cannot obtain a strict improvement in the total risk-to-go at time $t$ by changing only the allocation of $S_{t+1}$. Furthermore, Lemma \ref{lemma:pot_rep} allows us to derive an inf-convolution result for allocations that are PO at time $t$. 

\medskip

\begin{proposition} \label{prop:pot_inf} Let $t\in\mfT$ and $\Y_{t+2:T}^*$ be an allocation of $S_{t+2:T}$. Then, $\{\Z_{t+1}, \Y_{t+2:T}^*\}\in\PO_{t}$ if and only if $\Z_{t+1}$ attains the infimum
\begin{equation} \label{eq:inf_t}
\underset{\left\{
    \W_{t+1}~|~\{\W_{t+1}, \Y_{t+2:T}^{*}\}\in\IR_{t}
\right\}} \inf \,                           
\psi\left(\sum_{i\in\mN}\rho_{t}^{(i)}
        \left(W_{t+1}^{(i)}+R_{t+1}^{(i)}
            \left(Y_{t+2:T}^{*(i)}\right)
        \right)
    \right),
\end{equation}

\medskip

\noindent for some monotone $\psi:\L_{t}^{\infty}\to\R$ that is strictly monotone on the set
\[
\aleph:=
    \left\{
        V\in\L_{t}^{\infty} ~|~ V = \sum_{i\in\mN}\rho_{t}^{(i)}
            \big(V_{t+1}^{(i)}+R_{t+1}^{(i)}(Y_{t+2:T}^{*(i)})\big)
        \ \text{ for some } \ \{\V_{t+1}, \Y_{t+2:T}^*\}\in\IR_{t}
    \right\}.
\]
\end{proposition}

\begin{proof}
Let $\{\Z_{t+1}, \Y_{t+2:T}^*\}\in\PO_{t}$ and  $Z:=\sum_{i\in\mN}\rho_{t}^{(i)}\big(Z_{t+1}^{(i)}+R_{t+1}^{(i)}(Y_{t+2:T}^{*(i)})\big)$. Define the functional $\psi_Z: \L_{t}^{\infty} \to \R$ by $\psi_Z(W) := \E[\max\{W, Z\}]$. We first show that $\Z_{t+1}$ attains the infimum \eqref{eq:inf_t} for the mapping $\psi_Z$. For $\{\Y_{t+1}, \Y_{t+2:T}^*\}\in\IR_{t}$, let $Y:=\sum_{i\in\mN}\rho_{t}^{(i)}\big(Y_{t+1}^{(i)}+R_{t+1}^{(i)}(Y_{t+2:T}^{*(i)})\big)$. Then,
\[
\psi_Z(Y) = \E[\max\{Y, Z\}] \geq \E[Z] = \psi_Z(Z).
\]

\noindent It remains to show that $\psi_Z$ is strictly monotone on $\aleph$. Let $\tilde{Y}, \, \hat{Y}\in \aleph$. Then, by definition of $\aleph$, there exist some $\{\tilde{\Y}_{t+1}, \Y_{t+2:T}^*\}, \, \{\hat{\Y}_{t+1}, \Y_{t+2:T}^*\}\in\IR_{t}$, such that 
\[
\tilde{Y}=\sum_{i\in\mN}\rho_{t}^{(i)}
    \big(\tilde{Y}_{t+1}^{(i)}+R_{t+1}^{(i)}(Y_{t+2:T}^{*(i)})\big)
\ \ \hbox{ and } \ \ 
\hat{Y}=\sum_{i\in\mN}\rho_{t}^{(i)}
    \big(\hat{Y}_{t+1}^{(i)}+R_{t+1}^{(i)}(Y_{t+2:T}^{*(i)})\big).
\] 

\noindent Assume, without loss of generality, that $\tilde{Y}\leq \hat{Y}$. Then,
\[
\psi_Z(\hat{Y})-\psi_Z(\tilde{Y}) = 
    \E[\max\{\hat{Y}, Z\}]-\E[\max\{\tilde{Y}, Z\}] 
\geq 0.
\]

\noindent To show strict monotonicity on $\aleph$, we need to show that $\psi_Z(\tilde{Y})<\psi(\hat{Y})$ whenever $\tilde{Y}\leq \hat{Y}$ and $\P(\tilde{Y}<\hat{Y})>0$. Assume, by way of contradiction, that $\P(\tilde{Y}<\hat{Y})>0$ and $\psi_Z(\tilde{Y})=\psi_Z(\hat{Y})$. By definition of $\psi_Z$, we have 
\[
\E[\max\{\hat{Y}, Z\}-\max\{\tilde{Y}, Z\}] = 0.
\]

\noindent Since the non-negative random variable $\max\{\hat{Y}, Z\}-\max\{\tilde{Y}, Z\}$ is strictly positive whenever $\P(\hat{Y} > Z)$ or $\P(\tilde{Y} > Z)$, we have $\P(\hat{Y}\leq Z)=\P(\tilde{Y}\leq Z)=1$. Furthermore, as $\{\Z_{t+1}, \Y_{t+2:T}^*\}\in\PO_{t}$, Lemma \ref{lemma:pot_rep} implies that $\P(\hat{Y}<Z)=\P(\tilde{Y}<Z)=0$. Consequently,
\[
\P(\hat{Y}=Z)=\P(\tilde{Y}=Z)=\P(\hat{Y}=\tilde{Y})=1,
\] 

\noindent thereby contradicting the assumption that $\P(\tilde{Y}<\hat{Y})>0$. Therefore, $\psi$ is strictly monotone on $\aleph$.

\smallskip

Conversely, suppose that $\Z_{t+1}'$ attains the infimum \eqref{eq:inf_t}, for some $\psi$ that is strictly monotone on $\aleph$. Assume, by way of contradiction, that $\{\Z'_{t+1}, \Y^*_{t+2:T}\}\notin\PO_{t}$. By Lemma \ref{lemma:pot_rep}, there exists $\Y_{t+1}$ such that $\{\Y_{t+1}, \Y^*_{t+2:T}\}\in\IR_{t}$,
\begin{align} 
&\P\left(\sum_{i\in\mN}\rho_{t}^{(i)}
    \big(Y_{t+1}^{(i)}+R_{t+1}^{(i)}(Y_{t+2:T}^{*(i)})\big)
        \leq \sum_{i\in\mN}\rho_{t}^{(i)}
    \big(Z_{t+1}'^{(i)}+R_{t+1}^{(i)}(Y_{t+2:T}^{*(i)})\big)
\right)=1, \ \text{and} \label{eq:proof_prop4} \\
&\P\left(\sum_{i\in\mN}\rho_{t}^{(i)}
    \big(Y_{t+1}^{(i)}+R_{t+1}^{(i)}(Y_{t+2:T}^{*(i)})\big)
        < \sum_{i\in\mN}\rho_{t}^{(i)}
    \big(Z_{t+1}'^{(i)}+R_{t+1}^{(i)}(Y_{t+2:T}^{*(i)})\big)
\right)>0. \label{eq:proof_prop5}
\end{align}

\noindent Since $\psi$ is strictly monotone on $\aleph$, \eqref{eq:proof_prop4} and \eqref{eq:proof_prop5} imply that 
\[
\psi\left(\sum_{i\in\mN}\rho_{t}^{(i)}
    \big(Y_{t+1}^{(i)}+R_{t+1}^{(i)}(Y_{t+2:T}^{*(i)})\big)
\right) <
\psi\left(\sum_{i\in\mN}\rho_{t}^{(i)}
    \big(Z_{t+1}'^{(i)}+R_{t+1}^{(i)}(Y_{t+2:T}^{*(i)})\big)
\right),
\]

\noindent which contradicts the optimality of $\Z_{t+1}'$ for \eqref{eq:inf_t}. Thus, $\{\Z'_{t+1}, \Y_{t+2:T}^*\}\in\PO_{t}$.
\end{proof}

\medskip

The set $\aleph$ is the set of attainable risks at time $t$ given an allocation $\Y_{t+2:T}^{*}$. Moreover, when $T=1$, we recover Proposition 3.3 from \citet{ghossoub2026efficiency} by taking $\psi$ to be the identity function. For $T\geq 2$, the total risk-to-go at all future time points is a random variable and thus we need to introduce a real-valued functional $\psi$ into our optimization problem. When the functional $\psi$ is strictly monotone on the set of attainable risks at time $t$, then the solution to the optimization problem \eqref{eq:inf_t} corresponds to a PO allocation at time $t$. Conversely, for any allocation $\{\Z_{t+1}, \Y_{t+2:T}^*\}$ that is PO at time $t$, there exists a functional $\psi$ that is strictly monotone on the set of attainable risks at time $t$ such that $\Z_{t+1}$ solves the optimization problem \eqref{eq:inf_t}. Combining Proposition \ref{prop:pot_inf} and Lemma \ref{lemma:dpo} yields the main characterization result for this section.

\medskip

\begin{theorem} \label{theorem:char_dpo} $\Y_{1:T}\in\DPO$ if and only if for each $t\in\{T, \ldots, 1\}$, $\Y_{t}^*$ is defined recursively as a solution to
\begin{equation}\label{eq:inf_dpo}
\underset{\left\{
    \Y_{t}~|~\{\Y_{t}, \Y_{t+1:T}^{*}\}\in\IR_{t-1}
\right\}} \argmin \, 
\psi_{t-1}\left(\sum_{i\in\mN}\rho_{t-1}^{(i)}
    \big(Y_{t}^{(i)}+R_{t}^{(i)}(Y_{t+1:T}^{*(i)})\big)
\right),
\end{equation}

\medskip

\noindent for some $\psi_{t-1}:\L_{t-1}^{\infty}\to\R$ satisfying the conditions of Proposition \ref{prop:pot_inf}. 
\end{theorem}

\smallskip

\begin{proof}
Suppose that $\Y_{1:T}^*$ is constructed recursively according to the theorem. Then, $\Y_{t+1:T}^*\in\PO_t$ for all $t\in\T$ by Proposition \ref{prop:pot_inf} and therefore $\Y_{1:T}^*\in\DPO$ by Lemma \ref{lemma:dpo}. 

\smallskip

Conversely, suppose that $\Y_{1:T}^*\in\DPO$. Then, $\Y_{t+1:T}^*\in\PO_t$ for all $t\in\T$ by Lemma \ref{lemma:dpo} and hence $\Y_{1:T}^*$ has the desired recursive construction by Proposition \ref{prop:pot_inf}. 
\end{proof}

\medskip

Theorem \ref{theorem:char_dpo} states that an allocation is PO if and only if it is the solution to a series of recursive (backward in time) $\psi$ inf-convolution problems of the form \eqref{eq:inf_dpo}. Moreover, the only requirement on $\psi$ is strict monotonicity on the set of attainable risks at time $t$. Since the expected value is strictly monotone on $\L_t^{\infty}$ for all $t\in\T$, any solution to the expected value inf-convolution problem corresponds to an allocation that is PO at time $t$. We state this formally in the following corollary.

\medskip

\begin{corollary} \label{corollary:char_dpo} For $t\in\{T, \ldots, 1\}$, define $\Y_{t}^*$ recursively as a solution to
\begin{equation}\label{eq:inf_dpo_ev}
\underset{\left\{
    \Y_{t}~|~\{\Y_{t}, \Y_{t+1:T}^{*}\}\in\IR_{t-1}
\right\}} \argmin \, 
\E\left[\sum_{i\in\mN}\rho_{t-1}^{(i)}
    \big(Y_{t}^{(i)}+R_{t}^{(i)}(Y_{t+1:T}^{*(i)})\big)
\right].
\end{equation}

\medskip

\noindent Then, $\Y_{1:T}\in\DPO$. 
\end{corollary}

\bigskip

\section{Comonotone Allocations} \label{sec:cdpo}

In this section, we define a new notion of comonotonicity for allocation processes and study processes that are comonotone dynamic PO (which we just call $c$-PO). After introducing a dynamic comonotone improvement theorem in Section \ref{sec:com}, we define the notions of comonotone myopic Pareto optimality and comonotone Pareto optimality in Sections \ref{sec:cmPO} and \ref{sec:cDPO}. Then, we prove an analogue to Theorem \ref{theorem:char_dpo} for comonotone Pareto optima in Section \ref{sec:cdpo_char} and discuss the connection between the comonotone and unconstrained markets in Section \ref{sec:dpo_cdpo}. Finally, we derive an explicit characterization of $c$-PO allocations in Section \ref{sec:cdpo_ex_char}.

\medskip

\subsection{Comonotone Improvement} 
\label{sec:com}

Before defining a new notion of comonotonicity that is compatible with our dynamic setting, we first recall the classical definition of comonotonicity for random vectors. We say that a random vector $\Z$ is comonotone if there exists a random variable $W$ and non-decreasing functions $f^{(i)}:\R\to\R$ such that $Z^{(i)}=f^{(i)}(W)$ for all $i\in\mN$. Moreover, by \citet{denneberg1994non} Proposition 4.5, comonotonicity of $\Z$ is equivalent to the  existence of non-decreasing functions $g^{(i)}:\R\to\R$ such that for all $i \in \mN$, it holds that $Z^{(i)}=g^{(i)}(\sum_{j\in\mN}Z^{(j)})$. When PO allocations are comonotone in the static setting, then they are also comonotone with the aggregate, and thus agents have no incentive to under-report their initial risks in a risk-sharing problem. In an insurance setting, policyholders have no incentive to over-report their losses in a comonotone market, and hence comonotonicity is often referred to as the no-sabotage condition \citep{carlier2003pareto}.

In the dynamic setting, agents are not only concerned with the risk of the aggregate endowment at a fixed time, but also the risk-to-go. Hence, we propose the following notion of a comonotone process with respect to the dynamic risk measures $\{\rho^{(i)}_{t:T}\}_{t \in \mfT, i \in \mN}$. Unlike the static setting, our notion of comonotonicity depends on the preferences of the agents. However, as all agent preferences are fixed throughout, we simply say that $\Y_{1:T}$ is comonotone and omit the dependence on $\{\rho^{(i)}_{t:T}\}_{t \in \mfT, i \in \mN}$ whenever there is no confusion. 

\medskip

\begin{definition} [Comonotone Process] \label{definition:com_process} A process $\Y_{1:T}$ is comonotone with respect to $\{\rho^{(i)}_{t:T}\}_{t \in \mfT, i \in \mN}$ if for any $t\in\T$, the random vector $\big(Y_t^{(1)}+R_t^{(1)}(Y_{t+1:T}^{(1)}), \ldots, Y_t^{(n)}+R_t^{(n)}(Y_{t+1:T}^{(n)})\big)$ is comonotone (in the classical sense). We denote by $\A_{1:T}^C$ the set of comonotone processes.
\end{definition}

For $\tau\in\T$, if the random vector $\big(Y_s^{(1)}+R_s^{(1)}(Y_{s+1:T}^{(1)}), \ldots, Y_s^{(n)}+R_s^{(n)}(Y_{s+1:T}^{(n)})\big)$ is comonotone for all $s\geq\tau$, then we say that $\Y_{\tau:T}$ is a comonotone allocation of $S_{\tau:T}$ and denote by $\A_{\tau:T}^C$ the set of comonotone allocations of $S_{\tau:T}$. Moreover, as the risk-to-go term vanishes at the terminal time, $\Y_T$ is a comonotone allocation of $S_T$ in the classical sense (with respect to any dynamic risk measure). Hence, comonotonicity ensures that agents have no incentive to misreport their endowment or loss at the terminal time. Before the terminal time, agents do not have the incentive to misreport the sum of their endowment and the risk-to-go. 

Before we state the dynamic comonotone improvement theorem, we first recall the definition of the convex order. We refer to \citet{shaked2007stochastic} for a detailed discussion of stochastic orders.

\medskip

\begin{definition} [Convex Order] \label{definition:convex_order}
Let $Y, \, Z\in\L_T^{\infty}$. We say that $Z$ dominates $Y$ in the convex order, denoted $Y\cvx Z$, if for every convex function $\phi:\R\to\R$, 
\[
\E[\phi(Y)]\leq\E[\phi(Z)],
\] 

\noindent whenever the expectations are well defined. Moreover, if the above inequality is strict for some strictly convex function $\phi$, then we say that $Z$ strictly dominates $Y$ in the convex order, and use the notation $Y\scvx Z$. 
\end{definition}

We are now ready to state the dynamic comonotone improvement theorem. The proof of this result relies heavily on the comonotone improvement from the static case (see Theorem 3.1 of \citet{carlier2012pareto} or Proposition \ref{prop:static} $(ii)$ in Appendix \ref{appendix:single_period}). 

\medskip

\begin{lemma} [Dynamic Comonotone Improvement Theorem]\label{lemma:com_improvement}
For each $\Y_{1:T}\in\A_{1:T}$, there exists $\tilde{\Y}_{1:T}\in\A_{1:T}^C$ such that for all $i\in\mN$ and $s\in\T$,
\begin{equation} \label{eq:com_improvement}
\tilde{Y}_s^{(i)}+R_s^{(i)}(\tilde{Y}_{s+1:T}^{(i)})
    \cvx 
Y_s^{(i)}+R_s^{(i)}(\tilde{Y}_{s+1:T}^{(i)}).
\end{equation}

\smallskip

\noindent Moreover, if $\Y_{1:T}\notin\A_{1:T}^C$, then $\tilde{\Y}_{1:T}$ can be chosen such that for some $s\in\T$, there exists $j\in\mN$ such that
\begin{equation} \label{eq:com_improvement_2}
\tilde{Y}_s^{(j)}+R_s^{(j)}(\tilde{Y}_{s+1:T}^{(j)})
    \scvx 
Y_s^{(j)}+R_s^{(j)}(\tilde{Y}_{s+1:T}^{(j)}).
\end{equation}
\end{lemma}

\medskip

\begin{proof}
For $s=T$, $\Y_{s:T}$ is a random vector, so the existence of $\tilde{\Y}_T\in\A_T^C$ such that \eqref{eq:com_improvement} holds follows from Theorem 3.1 of \citet{carlier2012pareto}. Next, we proceed by induction backwards in time. For this, assume that there exists $\tilde{\Y}_{t:T}$ such that \eqref{eq:com_improvement} holds for all $s\geq t$ and let 
\[
\Z_{t-1}:=\left(
    Y_{t-1}^{(1)}+R_{t-1}^{(1)}(\tilde{Y}_{t:T}^{(1)}), \ldots, 
        Y_{t-1}^{(n)}+R_{t-1}^{(n)}(\tilde{Y}_{t:T}^{(n)})
    \right).
\] 

\noindent Then, $\Z_{t-1}$ is an allocation of $S_{t-1}+\sum_{i\in\mN}R_{t-1}^{(i)}(\tilde{Y}_{t:T}^{(i)})$. By Theorem 3.1 of \citet{carlier2012pareto}, there exists a comonotone allocation vector $\tilde{\Z}_{t-1}$ of $S_{t-1}+\sum_{i\in\mN}R_{t-1}^{(i)}(\tilde{Y}_{t:T}^{(i)})$ such that for all $i\in\mN$, 
\[
\tilde{Z}_{t-1}^{(i)} 
    \cvx 
Y_{t-1}^{(i)}+R_{t-1}^{(i)}(\tilde{Y}_{t:T}^{(i)}).
\] 

\noindent Let $\tilde{Y}_{t-1}^{(i)} := \tilde{Z}_{t-1}^{(i)}-R_{t-1}^{(i)}(\tilde{Y}_{t:T}^{(i)})$. Then, $\tilde{\Y}_{t-1:T}\in\A_{t-1:T}^C$ and $\tilde{\Y}_{t-1:T}$ satisfies \eqref{eq:com_improvement} for all $s\geq t-1$. 

\smallskip

Finally, \eqref{eq:com_improvement_2} follows from the same recursive construction (by applying the second half of Theorem 3.1 in \citet{carlier2012pareto}). 
\end{proof}

\medskip

For $T=1$, the random vector $\tilde{\Y}_{1:T}$ in Lemma \ref{lemma:com_improvement} is referred to as a comonotone improvement in the static case literature. For $T\geq 2$, the comonotone improvement $\tilde{\Y}_{1:T}$ differs from the classical notion due to the additional risk-to-go term. Furthermore, if $\Y_{t+1:T}$ is a comonotone process as defined in Definition \ref{definition:com_process}, then we can take $\tilde{\Y}_{t+1:T}=\Y_{t+1:T}$, in which case \eqref{eq:com_improvement} automatically holds for $s\geq t+1$. We state this formally in the following corollary. 

\medskip

\begin{corollary} \label{cor:com_improvement_2}
Let $t\in\T$ and $\Y_{t:T}\in\A_{t:T}$ with $\Y_{t+1:T}\in\A_{t+1:T}^C$. Then, there exists $\tilde{\Y}_{t}\in\A_t$ such that $\{\tilde{\Y}_t, \Y_{t+1:T}\}\in\A_{t:T}^C$ and
\begin{equation*}
\tilde{Y}_t^{(i)}+R_t^{(i)}(Y_{t+1:T}^{(i)}) 
    \cvx 
Y_t^{(i)}+R_t^{(i)}(Y_{t+1:T}^{(i)}),
\end{equation*}

\noindent for all $i\in\mN$. Moreover, if $\Y_{t:T}\notin\A_{t:T}^C$, then $\tilde{\Y}_{t}$ can be chosen such that there exists $j\in\mN$ with
\[
\tilde{Y}_t^{(j)}+R_t^{(j)}(Y_{t+1:T}^{(j)}) 
    \scvx 
Y_t^{(j)}+R_t^{(j)}(Y_{t+1:T}^{(j)}).
\]
\end{corollary}

\medskip

\begin{proof}
Let $\Z_{t} := \big(Y_{t}^{(1)}+R_{t}^{(1)}(Y_{t+1:T}^{(1)}), \ldots, Y_{t}^{(n)}+R_{t}^{(n)}(Y_{t+1:T}^{(n)})\big)$. Then, construct $\tilde{\Y}_t$ in the same manner as in the proof of Lemma \ref{lemma:com_improvement}. 
\end{proof}

\medskip

To make use of comonotone improvements, we require agent preferences to be monotone with respect to the convex order. Risk measures with this property are said to be convex order preserving, which we define next. 

\medskip

\begin{definition} [Convex Order Preserving] \label{definition:conv_order_pres} A one-step conditional risk measure $\rho_{t}$ is convex order preserving if for all $Y, \, Z\in\L_{t+1}^{\infty}$ such that $Y\cvx Z$, we have $\rho_t(Y)\leq\rho_t(Z)$. If, in addition, $\P\big(\rho_t(Y)<\rho_t(Z)\big)>0$ whenever $Y\scvx Z$, then we say that $\rho_t$ is strictly convex order preserving. 
\end{definition}

\medskip

\begin{assumption} \label{assumption:convex_order} For $i\in\mN$ and $t\in\mfT$, the $i$-th agent's one-step conditional risk measure $\rho_t^{(i)}$ preserves the convex order. 
\end{assumption}

Note that if we further assume that agent preferences are monotone (i.e., Assumption \ref{assumption:tc} holds), then the risk-to-go at any time $t$ of a comonotone improvement is upper bounded by the risk-to-go of the original allocation. We refer to Section \ref{sec:examples} for an example of when this assumption is satisfied.

\subsection{Comonotone Myopic Pareto Optima} 
\label{sec:cmPO} 

In this section, we revisit the notion of myopic Pareto optimality introduced in Section \ref{sec:myopic_po}. When agents are restricted to allocation processes that are comonotone in the sense of Definition \ref{definition:com_process}, we have the following definition of comonotone myopic Pareto optimality.

\medskip

\begin{definition} [Comonotone Myopic Pareto Optimality] \label{definition:myopic_cpo} An allocation $\Y_{1:T}^*\in\A_{1:T}$ is comonotone myopic PO if 
\begin{enumerate}[label = $(\roman*)$]
    \item $\Y_{1:T}^*\in\MIR\cap\A_{1:T}^C$
    \item there does not exist an allocation process $\Y_{1:T}\in\MIR\cap\A_{1:T}^C$ that satisfies $\rho_{0,T}^{(i)}(Y_{1:T}^{(i)})\leq \rho_{0,T}^{(i)}(Y_{1:T}^{*(i)})$ for all $i\in\mN$, and at least one inequality is strict. 
\end{enumerate}

\smallskip

\noindent We denote by $\CMPO$ the set of all comonotone myopic PO allocations. 
\end{definition}

For $T=1$, comonotone myopic Pareto optimality reduces to the classical notion of comonotone Pareto optimality. Moreover, in Section \ref{sec:myopic_po}, we showed that myopic Pareto optima solve an inf-convolution problem of the form \eqref{eq:myopic_po}. Similarly, consider the following problem: 
\begin{equation}\label{eq:inf_cmpo}
\underset{\Y_{1:T} \, \in \, \MIR \, \cap \, \A_{1:T}^C} \inf \
    \sum_{i\in\mN}\rho_{0,T}^{(i)}(Y_{1:T}^{(i)}) 
=
\underset{\Y_{1:T} \, \in \, \MIR \, \cap \, \A_{1:T}^C} \inf \
    \sum_{i\in\mN}\brho^{(i)}\left(\sum_{t=1}^{T}Y_t^{(i)}\right). 
\end{equation}

\noindent The next result shows that comonotone myopic Pareto optima solve the inf-convolution problem \eqref{eq:inf_cmpo}. The arguments in the proof are similar to the static case, which we include below for completeness. 

\medskip

\begin{proposition} \label{prop:cmpo}
An allocation $\Y_{1:T}^*$ is comonotone myopic PO if and only if $\Y_{1:T}^*$ solves \eqref{eq:inf_cmpo}. 
\end{proposition}

\begin{proof}
We first show that any allocation attaining the infimum \eqref{eq:inf_cmpo} is comonotone myopic PO. Suppose that $\Y_{1:T}^*\notin\CMPO$. Then, there exists $\tilde{\Y}_{1:T}\in\MIR\cap\A_{1:T}^C$ such that $\rho_{0,T}^{(i)}(\tilde{Y}_{1:T}^{(i)})\leq \rho_{0,T}^{(i)}(Y_{1:T}^{*(i)})$ for all $i\in\mN$, and at least one inequality is strict. Hence, 
\[
\sum_{i\in\mN}\rho_{0,T}^{(i)}(\tilde{Y}_{1:T}^{(i)}) < 
  \sum_{i\in\mN}\rho_{0,T}^{(i)}(Y_{1:T}^{*(i)}). 
\]

\noindent Thus, $\Y_{1:T}^*$ does not attain the infimum \eqref{eq:inf_cmpo} and we conclude that any allocation attaining the infimum \eqref{eq:inf_cmpo} must be comonotone myopic PO.

\smallskip

Conversely, assume, by way of contradiction, that there exists $\Y_{1:T}^*\in\CMPO$ that does not attain the infimum \eqref{eq:inf_cmpo}. Then, there exists $\tilde{\Y}_{1:T}\in\MIR\cap\A_{1:T}^C$ such that 
\begin{equation} \label{eq:inf_cmpo_proof}
\sum_{i\in\mN}\rho_{0,T}^{(i)}(\tilde{Y}_{1:T}^{(i)}) < 
  \sum_{i\in\mN}\rho_{0,T}^{(i)}(Y_{1:T}^{*(i)}),  
\end{equation}

\noindent Define a partition $\mN_1, \mN_2, \mN_3$ of $\mN$ such that 
\begin{align*}
&\rho_{0,T}^{(i)}(\tilde{Y}_{1:T}^{(i)}) > 
    \rho_{0,T}^{(i)}(Y_{1:T}^{*(i)}), 
\ \text{for all} \ i\in\mN_1, \\
&\rho_{0,T}^{(i)}(\tilde{Y}_{1:T}^{(i)}) = 
    \rho_{0,T}^{(i)}(Y_{1:T}^{*(i)}), 
\ \text{for all} \ i\in\mN_2, \\
&\rho_{0,T}^{(i)}(\tilde{Y}_{1:T}^{(i)}) < 
    \rho_{0,T}^{(i)}(Y_{1:T}^{*(i)}), 
\ \text{for all} \ i\in\mN_3. \\
\end{align*}

\noindent Note that $\mN_3\neq\emptyset$ by \eqref{eq:inf_cmpo_proof}, and $\mN_1\neq\emptyset$ since $\Y_{1:T}^*\in\CMPO$. Moreover, for $i\in\mN_1$, let 
\[
\epsilon^{(i)}:= \rho_{0,T}^{(i)}(Y_{1:T}^{*(i)}) - 
    \rho_{0,T}^{(i)}(\tilde{Y}_{1:T}^{(i)}) <0.
\]

\noindent By \eqref{eq:inf_cmpo_proof}, there exists $\{\epsilon^{(i)}\}_{i\in\mN_3}$ such that 

\begin{enumerate}[label = $(\roman*)$]
    \item $\epsilon^{(i)}\leq 0$ for all $i\in\mN_3$;
    \item $\rho_{0,T}^{(i)}(\tilde{Y}_{1:T}^{(i)})-\epsilon^{(i)}\leq \rho_{0,T}^{(i)}(Y_{1:T}^{*(i)})$ for all $i\in\mN_3$, with at least one strict inequality; and 
    \item $\sum_{i\in\mN_1}\epsilon^{(i)}=\sum_{i\in\mN_3}\epsilon^{(i)}$. 
\end{enumerate}

\noindent Let 
\[
\hat{Y}_1^{(i)} := \begin{cases}
    \tilde{Y}_1^{(i)}+\epsilon^{(i)}, & i\in\mN_1 \\
    \tilde{Y}_1^{(i)}, & i\in\mN_2 \  \\
    \tilde{Y}_1^{(i)}-\epsilon^{(i)}, & i\in\mN_3   
\end{cases}
\quad \text{and} \quad 
\hat{\Y}_{2:T}=\tilde{\Y}_{2:T}.
\]

\noindent Note that $\hat{\Y}_{1:T}\in\A_{1:T}$ since $(iii)$ holds and $\tilde{\Y}_{1:T}\in\A_{1:T}^C\subseteq\A_{1:T}$. Moreover, $\hat{\Y}_{1:T}$ is comonotone since $\tilde{\Y}_{1:T}$ is comonotone and $\epsilon^{(i)}\in\R$ for all $i\in\mN_1\cup\mN_3$. Therefore, $\hat{\Y}_{1:T}\in\A_{1:T}^C$. Finally, by translation invariance, 
\begin{align}
 \rho_{0,T}^{(i)}(\hat{Y}_{1:T}^{(i)}) = 
    \rho_{0,T}^{(i)}(\tilde{Y}_{1:T}^{(i)}) + \epsilon^{(i)} = 
\rho_{0,T}^{(i)}(Y_{1:T}^{*(i)}) \leq 
    \rho_{0,T}^{(i)}(X_{1:T}^{(i)}),
\quad &\text{for all} \quad i\in\mN_1,
\nonumber \\
 \rho_{0,T}^{(i)}(\hat{Y}_{1:T}^{(i)}) = 
    \rho_{0,T}^{(i)}(Y_{1:T}^{*(i)}) \leq
\rho_{0,T}^{(i)}(X_{1:T}^{(i)}), \qquad\qquad
\quad &\text{for all} \quad i\in\mN_2, 
\nonumber \\
\rho_{0,T}^{(i)}(\hat{Y}_{1:T}^{(i)}) =  
    \rho_{0,T}^{(i)}(\tilde{Y}_{1:T}^{(i)}) - \epsilon^{(i)} \leq 
\rho_{0,T}^{(i)}(Y_{1:T}^{*(i)}) \leq
    \rho_{0,T}^{(i)}(X_{1:T}^{(i)}),
\quad &\text{for all} \quad i\in\mN_3. 
\label{eq:mir}
\end{align}

\noindent Thus, $\hat{\Y}_{1:T}\in\MIR$. Moreover, at least one inequality in \eqref{eq:mir} is strict, contradicting the assumption that $\Y_{1:T}^*\in\CMPO$.  
\end{proof}

\medskip

Similar to the case of myopic Pareto optimality in the unconstrained market, comonotone myopic Pareto optimality is a condition on the total allocation for each agent over time, and never defines a unique stochastic process. Hence, it is necessary to consider an alternative notion of optimality in the comonotone market. 

\medskip
\subsection{Comonotone Dynamic Pareto Optima} 
\label{sec:cDPO}

Next, we define the notions of comonotone Pareto optimality and comonotone Pareto optimality at time $t$. We further show that under Assumptions \ref{assumption:tc} and \ref{assumption:convex_order}, an allocation $\Y_{1:T}$ is $c$-PO if and only if the sub-allocation $\Y_{t+1:T}$ is $c$-PO at time $t$ for all $t\in\mfT$. This result allows us to construct comonotone Pareto optima recursively. 

\medskip

\begin{definition} [Comonotone Dynamic Pareto Optimality] \label{definition:cdpo} An allocation $\Y_{1:T}^*\in\A_{1:T}$ is comonotone dynamic PO if 
\begin{enumerate}[label = $(\roman*)$]
    \item $\Y_{1:T}^*\in\DIR\cap\A_{1:T}^C$
    \item there does not exist an allocation process $\Y_{1:T}\in\DIR\cap\A_{1:T}^C$ such that for all $i\in\mN$ and $t\in\mfT$,   
    \[
    \rho_{t,T}^{(i)}(Y_{t+1:T}^{(i)})\leq \rho_{t,T}^{(i)}(Y_{t+1:T}^{*(i)}),
    \] 
    and $\P\left(\rho_{t,T}^{(i)}(Y_{t+1:T}^{(i)})<\rho_{t,T}^{(i)}(Y_{t+1:T}^{*(i)})\right)>0$ for some $i\in\mN$ and $t\in\mfT$.
\end{enumerate}

\smallskip
 
\noindent We denote by $\CDPO$ the set of all comonotone dynamic PO allocations. 
\end{definition}

\medskip

\begin{definition} [Comonotone Pareto Optimality at Time $t$] \label{definition:cpo_t} For $t\in\mfT$, an allocation $\Y_{t+1:T}^*\in\A_{t+1:T}$ is comonotone PO at time $t$ if
\begin{enumerate}[label = $(\roman*)$]
    \item $\Y^*_{t+1:T}\in\IR_t\cap\A_{t+1:T}^C$
    \item there does not exist a random vector $\Y_{t+1}$ with $\{\Y_{t+1}, \Y_{t+2:T}^*\}\in\IR_t\cap\A_{t+1:T}^C$ such that for all $i\in\mN$,   
    \[
    \rho_{t}^{(i)}\big(Y_{t+1}^{(i)}+R_{t+1}^{(i)}(Y_{t+2:T}^{*(i)})\big)
        \leq 
    \rho_{t}^{(i)}\big(Y_{t+1}^{*(i)}+R_{t+1}^{(i)}(Y_{t+2:T}^{*(i)})\big),
    \]    
    and $\P\left(\rho_{t}^{(i)}\big(Y_{t+1}^{(i)}+R_{t+1}^{(i)}(Y_{t+2:T}^{*(i)})\big)<\rho_{t}^{(i)}\big(Y_{t+1}^{*(i)}+R_{t+1}^{(i)}(Y_{t+2:T}^{*(i)})\right)>0$ for some $i\in\mN$.
\end{enumerate}

\smallskip

\noindent We denote by $\CPO_t$ the set of all processes $\Y_{t+1:T}$ that are $c$-PO at time $t$. 
\end{definition}

The following result is the analogue of Lemma \ref{lemma:dpo} for comonotone processes. This result allows us to characterize $c$-PO allocations by characterizing allocations that are $c$-PO at time $t$ instead.  

\medskip

\begin{lemma}\label{lemma:cdpo} If $\Y_{t+1:T}^*\in\CPO_t$ for all $t\in\mfT$, then $\Y_{1:T}^*\in\CDPO$. Moreover, the converse also holds under Assumption \ref{assumption:convex_order}. 
\end{lemma}

\medskip

\begin{proof}

Suppose that $\Y_{t+1:T}^*\in\CPO_t$ for all $t\in\mfT$. Since $\Y_{t+1:T}^*\in\IR_t\cap\A_{t+1:T}^C$ for all $t\in\mfT$, it follows that $\Y_{1:T}^*\in\DIR\cap\A_{1:T}^C$. Hence, $\Y_{1:T}^*$ satisfies the first condition of comonotone Pareto optimality. Next, let $\Y_{1:T}\in\DIR\cap\A_{1:T}^C$ such that for all $t\in\mfT$ and $i\in\mN$,
\begin{equation*}
R_{t}^{(i)}(Y_{t+1:T}^{(i)})\leq R_{t}^{(i)}(Y_{t+1:T}^{*(i)}).
\end{equation*}

\noindent To verify the second condition of comonotone Pareto optimality, it suffices to show that for all $i\in\mN$ and $s\in\mfT$,
\begin{equation} \label{eq:lemma_cdpo_2}
R_{s}^{(i)}(Y_{s+1:T}^{(i)})=R_{s}^{(i)}(Y_{s+1:T}^{*(i)}). 
\end{equation}

\noindent The case when $s=T-1$ holds by the same arguments as in the proof of Lemma \ref{lemma:dpo}. Next, we proceed by induction backwards in time. For this, assume \eqref{eq:lemma_cdpo_2} holds for all $s\geq t$. It suffices to show that \eqref{eq:lemma_cdpo_2} also holds for $t-1$. By the arguments in Lemma \ref{lemma:dpo}, we have
\[
\rho_{t-1}^{(i)}\big(Y_{t}^{(i)}+R_{t}^{(i)}(Y_{t+1:T}^{*(i)})\big)
    \leq 
\rho_{t-1}^{(i)}\big(X_{t}^{(i)}+R_{t}^{(i)}(X_{t+1:T}^{(i)})\big), 
    \ \text{for all} \ i\in\mN.
\]

\noindent Thus, $\{\Y_t, \Y_{t+1:T}^*\}\in\IR_{t-1}$. Next, we show that $\{\Y_t, \Y_{t+1:T}^*\}\in\A_{t:T}^C$. Since $\Y_{t+1:T}^*\in\CPO_{t}$, the random vector 
\[
\left(
    Y_s^{*(1)}+R_s^{(1)}(Y_{s+1:T}^{*(1)}), \ldots, 
    Y_s^{*(n)}+R_s^{(n)}(Y_{s+1:T}^{*(n)})
\right),
\] 

\noindent is comonotone for all $s\geq t+1$. Moreover, by the induction hypothesis,
\[
\left(Y_t^{(1)}+R_t^{(1)}(Y_{t+1:T}^{*(1)}), \ldots,        
    Y_t^{(n)}+R_t^{(n)}(Y_{t+1:T}^{*(n)})\right) = 
\left(Y_t^{(1)}+R_t^{(1)}(Y_{t+1:T}^{(1)}), \ldots, 
    Y_t^{(n)}+R_t^{(n)}(Y_{t+1:T}^{(n)})\right).
\]

\noindent Since $\Y_{1:T}\in\A_{1:T}^C$, it follows that $\{\Y_t, \Y_{t+1:T}^*\}\in\A_{t:T}^C$. The remainder of the induction argument follows from the arguments in the proof of Lemma \ref{lemma:dpo}. 

\smallskip

For the converse, let $\Y_{1:T}^*\in\CDPO$ and suppose that Assumption \ref{assumption:convex_order} holds. Assume, by way of contradiction, that $\Y_{t+1:T}^*\notin\CPO_t$ for some $t\in\mfT$. Since $\Y_{1:T}^*\in\CDPO$, we have $\Y_{t+1:T}^*\in\IR_t\cap\A_{t+1:T}^C$ for all $t\in\mfT$. Thus, $\Y_{t+1:T}^*\notin\CPO_t$ implies that there exists $\Y_{t+1}$ with $\{\Y_{t+1}, \Y_{t+2:T}^*\}\in\IR_t\cap\A_{t+1:T}^C$ such that
\begin{align}
&\P\left(\rho_{t}^{(i)}
    \big(Y_{t+1}^{(i)}+R_{t+1}^{(i)}(Y_{t+2:T}^{*(i)})\big)
        \leq \rho_{t}^{(i)}
    \big(Y_{t+1}^{*(i)}+R_{t+1}^{(i)}(Y_{t+2:T}^{*(i)})\big)
\right)=1 \ \text{for all} \ i\in\mN \ \text{and},
\nonumber \\
&\P\left(\rho_{t}^{(i)}
    \big(Y_{t+1}^{(i)}+R_{t+1}^{(i)}(Y_{t+2:T}^{*(i)})\big) 
         < \rho_{t}^{(i)}
    \big(Y_{t+1}^{*(i)}+R_{t+1}^{(i)}(Y_{t+2:T}^{*(i)})\big)
\right)>0 \ \text{for some} \ i\in\mN.
\label{eq:lemma_cdpo_4} 
\end{align}

\noindent Let $\hat{\Y}_{1:T}$ be the process $\Y_{1:T}^*$ with the $(t+1)$-st vector replaced by $\Y_{t+1}$. Let $\tilde{\Y}_{1:T}\in\A_{1:T}^C$ be a comonotone improvement of $\hat{\Y}_{1:T}$ (which exists by Lemma \ref{lemma:com_improvement}). Next, we want to show that for all $s\in\mfT$ and $i\in\mN$, 
\begin{equation}\label{eq:lemma_cdpo_5}
R_{s}^{(i)}(\tilde{Y}_{s+1:T}^{(i)})
    \leq 
R_{s}^{(i)}(\hat{Y}_{s+1:T}^{(i)})\leq R_{s}^{(i)}(Y_{s+1:T}^{*(i)}).
\end{equation}

\noindent The second inequality follows from Equation \eqref{eq:lemma_dpo_9} in the proof of Lemma \ref{lemma:dpo}, so it remains to prove the first inequality. We proceed by induction. For $s=T-1$, \eqref{eq:lemma_cdpo_5} holds since $\rho_{T-1}^{(i)}$ preserves the convex order for all $i\in\mN$. Next, assume that \eqref{eq:lemma_cdpo_5} holds for $s\geq t$. We want to show that \eqref{eq:lemma_cdpo_5} holds for $s=t-1$. Note that for all $i\in\mN$,
\begin{align*}
R_{t-1}^{(i)}(\tilde{Y}_{t:T}^{(i)}) &= \rho_{t-1}^{(i)}
    \big(\tilde{Y}_t^{(i)}+R_t^{(i)}(\tilde{Y}_{t+1:T}^{(i)})\big) \\
& \leq \rho_{t-1}^{(i)}
    \big(\hat{Y}_t^{(i)}+R_t^{(i)}(\tilde{Y}_{t+1:T}^{(i)})\big) \\
& \leq \rho_{t-1}^{(i)}
    \big(\hat{Y}_t^{(i)}+R_t^{(i)}(\hat{Y}_{t+1:T}^{(i)})\big) \\
& = R_{t-1}^{(i)}(\hat{Y}_{t:T}^{(i)}),
\end{align*}

\noindent where the first inequality holds since $\rho_{t-1}^{(i)}$ preserves the convex order for all $i\in\mN$ and the second inequality follows from monotonicity of $\rho_{t-1}^{(i)}$ and the induction hypothesis. Hence, \eqref{eq:lemma_cdpo_5} holds for all $s\in\mfT$ and $i\in\mN$. Finally, for $i\in\mN$ satisfying \eqref{eq:lemma_cdpo_4}, we have
\begin{align*}
\P\left(\rho_{t}^{(i)}\big(
    \tilde{Y}_{t+1}^{(i)}+R_{t+1}^{(i)}(\tilde{Y}_{t+2:T}^{(i)})\big) 
        < R_{t}^{(i)}(Y_{t+1:T}^{*(i)})
    \right) 
    \geq
\P\left(\rho_{t}^{(i)}\big(
    Y_{t+1}^{(i)}+R_{t+1}^{(i)}(Y_{t+2:T}^{*(i)})\big)
        < R_{t}^{(i)}(Y_{t+1:T}^{*(i)})
    \right) 
 > 0,
\end{align*}

\noindent where the first inequality follows from \eqref{eq:lemma_cdpo_5} and the last inequality follows from \eqref{eq:lemma_cdpo_4}. Finally, \eqref{eq:lemma_cdpo_5} implies that $\tilde{\Y}_{1:T}\in\DIR$ and hence $\tilde{\Y}_{1:T}\in\DIR\cap\A_{1:T}^C$. This contradicts the assumption that $\Y_{1:T}^*\in\CDPO$, so $\Y_{t+1:T}^*\in\CPO_t$. 
\end{proof}

\medskip

\subsection{Characterization of Comonotone Dynamic Pareto Optimal Allocations} \label{sec:cdpo_char} 

In this section, we derive characterization results for $c$-PO allocations. We start with an analogue of Lemma \ref{lemma:pot_rep}, which gives a useful representation of allocations that are $c$-PO at time $t$.  

\medskip

\begin{lemma} \label{lemma:cpot_rep} For $t\in\mfT$, let
\begin{align*}
\Gamma:=\biggl\{                   
    \Z_{t+1:T}\biggr. &\in\IR_t\cap\A_{t+1:T}^C ~\big|~ \P\left(
        \sum_{i\in\mN}\rho_{t}^{(i)}
            \big(Y_{t+1}^{(i)}+R_{t+1}^{(i)}(Z_{t+2:T}^{(i)})\big)
                \leq
        \sum_{i\in\mN}R_{t}^{(i)}(Z_{t+1:T}^{(i)})
    \right)<1 \ \text{or} \\
    & \P\left(
        \sum_{i\in\mN}\rho_{t}^{(i)}
            \big(Y_{t+1}^{(i)}+R_{t+1}^{(i)}(Z_{t+2:T}^{(i)})\big)
                <
        \sum_{i\in\mN}R_{t}^{(i)}(Z_{t+1:T}^{(i)})
    \right)=0 \\
     & \biggl. \qquad\qquad\qquad \text{for all} \ 
        \Y_{t+1}\neq \Z_{t+1} \ 
    \text{such that} \ \{\Y_{t+1}, \Z_{t+2:T}\}\in\IR_t \cap\A_{t+1:T}^C
\biggr\}.
\end{align*}
Then, $\Gamma\subseteq\CPO_t$. If, in addition, Assumption \ref{assumption:convex_order} holds (i.e., agent preferences preserve the convex order), then $\Gamma=\CPO_t$.  
\end{lemma}

\medskip

\begin{proof}
The statement $\Gamma\subseteq\CPO_t$ follows from similar arguments as in the proof of Lemma \ref{lemma:pot_rep} (with $\IR_t\cap\A_{t+1:T}^C$ instead of $\IR_t$ and $\CPO_t$ instead of $\PO_t$). 

\smallskip

To show that $\CPO_t\subseteq\Gamma$, suppose that Assumption \ref{assumption:convex_order} holds and let $\Y_{t+1:T}^*\in\CPO_t$. Assume, by way of contradiction, that $\Y_{t+1:T}^*\notin \Gamma$. Then, by definition of $\Gamma$, there exists a random vector $\tilde{\Y}_{t+1}$ with $\{\tilde{\Y}_{t+1}, \Y_{t+2:T}^*\}\in\IR_t\cap\A_{t+1:T}^C$ such that
\begin{align*}
&\P\left(\sum_{i\in\mN}\rho_{t}^{(i)}
    \big(\tilde{Y}_{t+1}^{(i)}+R_{t+1}^{(i)}(Y_{t+2:T}^{*(i)})\big)
        \leq \sum_{i\in\mN}\rho_{t}^{(i)}
    \big(Y_{t+1}^{*(i)}+R_{t+1}^{(i)}(Y_{t+2:T}^{*(i)})\big)
\right)=1, \ \text{and} \\
&\P\left(\sum_{i\in\mN}\rho_{t}^{(i)}
    \big(\tilde{Y}_{t+1}^{(i)}+R_{t+1}^{(i)}(Y_{t+2:T}^{*(i)})\big) 
        <\sum_{i\in\mN}\rho_{t}^{(i)}
    \big(Y_{t+1}^{*(i)}+R_{t+1}^{(i)}(Y_{t+2:T}^{*(i)})\big)
\right)>0.
\end{align*}

\noindent Let $\hat{Y}_{t+1}^{(i)}:=\tilde{Y}_{t+1}^{(i)}+\epsilon^{(i)}$, where $\epsilon^{(i)}$ is defined as in the proof of Lemma \ref{lemma:pot_rep}. As $\{\hat{\Y}_{t+1}, \Y_{t+2:T}^*\}\in\A_{t+1:T}$ and $\Y_{t+2:T}^*\in\A_{t+2:T}^C$, by Corollary \ref{cor:com_improvement_2}, there exists $\Y'_{t+1}\in\A_{t+1}$ such that $\{\Y'_{t+1}, \Y_{t+2:T}^*\}\in\A_{t+1:T}^C$ and
\begin{equation*}
Y_{t+1}'^{(i)}+R_{t+1}^{(i)}(Y_{t+2:T}^{*(i)})\cvx \hat{Y}_{t+1}^{(i)}+R_{t+1}^{(i)}(Y_{t+2:T}^{*(i)}),
\end{equation*}

\noindent for all $i\in\mN$. Furthermore, for all $i\in\mN$, 
\begin{equation*}
\rho_{t}^{(i)}
    \big(Y_{t+1}'^{(i)}+R_{t+1}^{(i)}(Y_{t+2:T}^{*(i)})\big) \leq 
\rho_{t}^{(i)}
    \big(\hat{Y}_{t+1}^{(i)}+R_{t+1}^{(i)}(Y_{t+2:T}^{*(i)})\big) \leq
R_{t}^{(i)}(Y_{t+1:T}^{*(i)}) \leq 
R_{t}^{(i)}(X_{t+1:T}^{(i)}), 
\end{equation*}

\noindent where the first inequality holds since $\rho_{t}^{(i)}$ is convex order preserving, the second inequality follows from similar arguments as in the proof of Lemma \ref{lemma:pot_rep}, and the last inequality holds since $\Y_{t+1:T}^*\in\CPO_t\subseteq\IR_t$. Moreover, for some $i\in\mN$,
\begin{align*}
\P\left(\rho_{t}^{(i)}
    \big(Y_{t+1}'^{(i)}+R_{t+1}^{(i)}(Y_{t+2:T}^{*(i)})\big)
        < R_{t}^{(i)}(Y_{t+1:T}^{*(i)})
    \right) &\geq 
\P\left(\rho_{t}^{(i)}
    \big(\hat{Y}_{t+1}^{(i)}+R_{t+1}^{(i)}(Y_{t+2:T}^{*(i)})\big) 
        < R_{t}^{(i)}(Y_{t+1:T}^{*(i)})
    \right) \\
& > 0,
\end{align*}

\noindent where the last inequality follows from similar arguments as in the proof of Lemma \ref{lemma:pot_rep}. Finally, since $\{\Y'_{t+1}, \Y_{t+2:T}^*\}\in\IR_t\cap\A_{t+1:T}^C$, this contradicts the assumption that $\Y_{t+1:T}^*\in\CPO_t$. 
\end{proof}

\medskip

Lemma \ref{lemma:cpot_rep} states that if $\Z_{t+1:T}\in\CPO_t$, then we cannot obtain a comonotone strict improvement in the total risk-to-go at time $t$ by changing only the allocation of $S_{t+1}$. Furthermore, Lemma \ref{lemma:cpot_rep} allows us to derive an inf-convolution result for allocations that are $c$-PO at time $t$. The result is similar to Proposition \ref{prop:pot_inf}, except the infimum is now taken over all comonotone IR processes instead of all IR processes. 

\medskip

\begin{proposition} \label{prop:cpot_inf} Suppose that Assumption \ref{assumption:convex_order} holds. Let $t\in\mfT$ and $\Y_{t+2:T}^*\in\A_{t+2:T}^C$. Then, $\{\Z_{t+1}, \Y_{t+2:T}^*\}\in\CPO_{t}$ if and only if $\Z_{t+1}$ attains the infimum
\begin{equation} \label{eq:inf_t_cdpo}
\underset{\left\{
    \W_{t+1}~|~\{\W_{t+1}, \Y_{t+2:T}^{*}\}\in\IR_{t}\cap\A_{t+1:T}^C
\right\}} \inf \,                           
\psi\left(\sum_{i\in\mN}\rho_{t}^{(i)}
        \left(W_{t+1}^{(i)}+R_{t+1}^{(i)}
            \left(Y_{t+2:T}^{*(i)}\right)
        \right)
    \right).
\end{equation}

\medskip

\noindent for some monotone $\psi:\L_{t}^{\infty}\to\R$ that is strictly monotone on the set
\begin{equation*}
\aleph^C:=
    \left\{
        V\in\L_{t}^{\infty} ~|~ V = \sum_{i\in\mN}\rho_{t}^{(i)}
            \big(V_{t+1}^{(i)}+R_{t+1}^{(i)}(Y_{t+2:T}^{*(i)})\big) 
        \ \text{for some} \ 
            \{\V_{t+1}, \Y_{t+2:T}^*\} \in
        \IR_{t}\cap\A_{t+1:T}^C
    \right\}.
\end{equation*}
\end{proposition}

\medskip

\begin{proof}
This follows from similar arguments as in the proof of Proposition \ref{prop:pot_inf} (with Lemma \ref{lemma:cpot_rep} instead of Lemma \ref{lemma:pot_rep}). 
\end{proof}

\medskip

The set $\aleph^C$ is the set of attainable risks at time $t$ in the comonotone market. Combining Proposition \ref{prop:cpot_inf} and Lemma \ref{lemma:cdpo} yields the following characterization, which is the analogue of Theorem \ref{theorem:char_dpo} for comonotone allocations. Here, the argmin is again taken over comonotone IR allocations instead of all IR allocations. This key difference allows us to derive a more explicit characterization of $c$-PO allocations in Section \ref{sec:cdpo_ex_char}. 

\medskip

\begin{theorem} \label{theorem:char_cdpo} Suppose that Assumption \ref{assumption:convex_order} holds. Then, $\Y_{1:T}\in\CDPO$ if and only if for each $t\in\{T, \ldots, 1\}$, $\Y_{t}^*$ is defined recursively as a solution to
\begin{equation}\label{eq:inf_cdpo}
\underset{\left\{
    \Y_{t}~|~\{\Y_{t}, \Y_{t+1:T}^{*}\}\in\IR_{t-1}\cap\A_{t:T}^C
\right\}} \argmin \, 
\psi_{t-1}\left(\sum_{i\in\mN}\rho_{t-1}^{(i)}
    \big(Y_{t}^{(i)}+R_{t}^{(i)}(Y_{t+1:T}^{*(i)})\big)
\right),
\end{equation}

\medskip

\noindent for some $\psi_{t-1}:\L_{t-1}^{\infty}\to\R$ satisfying the conditions of Proposition \ref{prop:cpot_inf}. 
\end{theorem}

\medskip

\begin{proof}
The theorem follows from similar arguments as in the proof of Theorem \ref{theorem:char_dpo} (with Proposition \ref{prop:cpot_inf} instead of Proposition \ref{prop:pot_inf} and Lemma \ref{lemma:cdpo} instead of Lemma \ref{lemma:dpo}).
\end{proof}

\medskip

Theorem \ref{theorem:char_cdpo} states that an allocation is $c$-PO if and only if it is the solution to a series of recursive (backward in time) $\psi$ inf-convolution problems of the form \eqref{eq:inf_cdpo}. For $c$-PO allocations, we only require $\psi$ to be strictly monotone on the set of attainable risks in the comonotone market. Again, this condition on $\psi$ is always satisfied by the expected value.  

\medskip

\begin{corollary} \label{corollary:char_cdpo} Suppose that Assumption \ref{assumption:convex_order} holds. For $t\in\{T, \ldots, 1\}$, define $\Y_{t}^*$ recursively as a solution to
\begin{equation}\label{eq:inf_cdpo_ev}
\underset{\left\{
    \Y_{t}~|~\{\Y_{t}, \Y_{t+1:T}^{*}\}\in\IR_{t-1}\cap\A_{t:T}^C
\right\}} \argmin \, 
\E\left[\sum_{i\in\mN}\rho_{t-1}^{(i)}
    \big(Y_{t}^{(i)}+R_{t}^{(i)}(Y_{t+1:T}^{*(i)})\big)
\right].
\end{equation}

\medskip

\noindent Then, $\Y_{1:T}\in\CDPO$. 
\end{corollary}

\medskip

\subsection{Connection Between Comonotone and Unconstrained Markets} \label{sec:dpo_cdpo}

In this section, we discuss the connection between comonotone and unconstrained markets, i.e., markets without the comonotonicity constraint. We show that when agent preferences preserve the convex order (i.e., under Assumption \ref{assumption:convex_order}), the set of $c$-PO allocations coincides with the set of PO allocations that are comonotone. If we further assume that agent preferences are strictly convex order preserving, then any PO allocation must be comonotone, which justifies the study of the comonotone market. We refer to \citet{ghossoub2026efficiency} for similar results in the static case.

\medskip

\begin{theorem} \label{theorem:dpo_cdpo} Suppose that Assumption \ref{assumption:convex_order} holds. Then, all of the following hold: 
\begin{enumerate}[label = $(\roman*)$]
    \item $\DPO\neq\emptyset$ if and only if $\CDPO\neq\emptyset$.
    \item For any $\Y_{1:T}^*\in\DPO$, there exists $\tilde{\Y}_{1:T}\in\A_{1:T}^C$ such that for all $i\in\mN$ and $s\in\mfT$,     
    \begin{equation} \label{eq:r_equal}
    R_s^{(i)}(Y_{s+1:T}^{*(i)}) = R_s^{(i)}(\tilde{Y}_{s+1:T}^{(i)}).
    \end{equation}    
    Furthermore, $\tilde{\Y}_{1:T}\in\CDPO$. 
    \item $\CDPO\subseteq\DPO$
    \item For $\Y_{1:T}^*\in\DPO$, let $\tilde{\Y}_{1:T}\in\A_{1:T}^C$ such that $(ii)$ holds. Then, for all $t\in\T$,    
    \begin{align*} 
    \underset{\left\{
        \Y_{t}~|~\{\Y_{t}, \Y_{t+1:T}^{*}\}\in\IR_{t-1}
    \right\}} \inf \, 
    & \psi_{t-1}\left(
        \sum_{i\in\mN}\rho_{t-1}^{(i)}\big(Y_t^{(i)}+R_t^{(i)}(Y_{t+1:T}^{*(i)})\big)
    \right) = \\
    & \underset{\left\{
        \Y_{t}~|~\{\Y_{t}, \tilde{\Y}_{t+1:T}\}\in\IR_{t-1}\cap\A_{t:T}^C
    \right\}} \inf \, 
    \psi_{t-1}\left(
        \sum_{i\in\mN}\rho_{t-1}^{(i)}\big(Y_t^{(i)}+R_t^{(i)}(\tilde{Y}_{t+1:T}^{(i)})\big)
    \right), 
    \end{align*}   
    where $\psi_{t-1}:\L_{t-1}^{\infty}\to\R$ is the functional given by Proposition \ref{prop:pot_inf} for the allocation $\Y_{t:T}^*$.
    \item For any $\Y_{1:T}^*\in\CDPO$ and $t\in\T$,
    \begin{align} \label{eq:dpo_cdpo}
    \underset{\left\{
        \Y_{t}~|~\{\Y_{t}, \Y_{t+1:T}^{*}\}\in\IR_{t-1}
    \right\}} \inf \, 
    & \psi_{t-1}\left(
        \sum_{i\in\mN}\rho_{t-1}^{(i)}\big(Y_t^{(i)}+R_t^{(i)}(Y_{t+1:T}^{*(i)})\big)
    \right) = \\
    & \underset{\left\{
        \Y_{t}~|~\{\Y_{t}, \Y_{t+1:T}^{*}\}\in\IR_{t-1}\cap\A_{t:T}^C
    \right\}} \inf \, 
    \psi_{t-1}\left(
        \sum_{i\in\mN}\rho_{t-1}^{(i)}\big(Y_t^{(i)}+R_t^{(i)}(Y_{t+1:T}^{*(i)})\big)
    \right), \nonumber
    \end{align}
    where $\psi_{t-1}:\L_{t-1}^{\infty}\to\R$ is the functional given by Proposition \ref{prop:cpot_inf} for the allocation $\Y_{t:T}^*$.
\end{enumerate} 
\end{theorem}

\medskip

\begin{proof}
We divide the proof of the theorem into five parts. First, we prove the first statement in $(ii)$. Second, we prove $(iv)$. Third, we finish the proof of $(ii)$, which also proves the first implication of $(i)$. Fourth, we prove $(iii)$, which also proves the converse of $(i)$. Finally, we prove $(v)$. 

\medskip

\noindent \underline{Proof of $(ii)$ - Part 1:} For $\Y_{1:T}^*\in\DPO$, we show the existence of $\tilde{\Y}_{1:T}\in\A_{1:T}^C$ by induction. For $s=T-1$, a comonotone improvement $\tilde{\Y}_T\in\A_T^C$ exists by Corollary \ref{cor:com_improvement_2}. Furthermore, Assumption \ref{assumption:convex_order} implies that for all $i\in\mN$,
\begin{equation} \label{eq:proof_dpo_cdpo_2}
\rho_T^{(i)}(\tilde{Y}_T^{(i)}) \leq \rho_T^{(i)}(Y_T^{*(i)}).
\end{equation}

\noindent By Lemma \ref{lemma:dpo}, $\Y_{1:T}^*\in\DPO$ implies that $\Y_T^*\in\PO_{T-1}$. Therefore, \eqref{eq:proof_dpo_cdpo_2} holds with equality for all $i\in\mN$. For the induction step, suppose there exists $\tilde{\Y}_{t+1:T}\in\A_{t+1:T}^C$ such that \eqref{eq:r_equal} holds for all $i\in\mN$ and $s\geq t$. Since $\{\Y_t^*, \tilde{\Y}_{t+1:T}\}\in\A_{t:T}$ and $\tilde{\Y}_{t+1:T}\in\A_{t+1:T}^C$, by Corollary \ref{cor:com_improvement_2}, there exists $\tilde{\Y}_t\in\A_t$ such that $\tilde{\Y}_{t:T}\in\A_{t:T}^C$ and 
\[
\tilde{Y}_t^{(i)}+R_t^{(i)}(\tilde{Y}_{t+1:T}^{(i)}) \cvx 
    Y_t^{*(i)}+R_t^{(i)}(\tilde{Y}_{t+1:T}^{(i)}),
\]

\noindent for all $i\in\mN$. Hence, by Assumption \ref{assumption:convex_order} and the fact that \eqref{eq:r_equal} holds for $s=t$,
\begin{equation} \label{eq:proof_dpo_cdpo_3}
\rho_{t-1}^{(i)}
    \big(\tilde{Y}_t^{(i)}+R_t^{(i)}(\tilde{Y}_{t+1:T}^{(i)})\big)
        \leq
    \rho_{t-1}^{(i)}
        \big(Y_t^{*(i)}+R_t^{(i)}(\tilde{Y}_{t+1:T}^{(i)})\big) 
        = 
    \rho_{t-1}^{(i)}\big(Y_t^{*(i)}+R_t^{(i)}(Y_{t+1:T}^{*(i)})\big).
\end{equation}

\noindent By Lemma \ref{lemma:dpo}, $\Y_{1:T}^*\in\DPO$ implies that $\Y_{t:T}^*\in\PO_{t-1}$. Therefore, \eqref{eq:proof_dpo_cdpo_3} holds with equality for all $i\in\mN$.

\medskip

\noindent \underline{Proof of $(iv)$:} For $\Y_{1:T}^*\in\DPO$, let $\tilde{\Y}_{1:T}\in\A_{1:T}^C$ such that $(ii)$ holds. Since $\Y_{1:T}^*\in\DPO$, Lemma \ref{lemma:dpo} implies that $\Y_{t:T}^*\in\PO_{t-1}$ for all $t\in\T$. Hence, for $t\in\T$, there exists a functional $\psi_{t-1}$ satisfying the conditions of Proposition \ref{prop:pot_inf} for the allocation $\Y_{t:T}^*$. Thus, for any $t\in\T$,
{\small
\begin{align*} 
\ \psi_{t-1}\left(\sum_{i\in\mN}\rho_{t-1}^{(i)}
    \big(Y_t^{*(i)}+R_t^{(i)}(Y_{t+1:T}^{*(i)})\big)
\right)  = 
\underset{\left\{
        \Y_{t}~|~\{\Y_{t}, \Y_{t+1:T}^{*}\}\in\IR_{t-1}
    \right\}} \inf \,
& \psi_{t-1}\left(\sum_{i\in\mN}\rho_{t-1}^{(i)}
    \big(Y_t^{(i)}+R_t^{(i)}(Y_{t+1:T}^{*(i)})\big)
\right) \\
 = \underset{\left\{
        \Y_{t}~|~\{\Y_{t}, \tilde{\Y}_{t+1:T}\}\in\IR_{t-1}
    \right\}} \inf \,
& \psi_{t-1}\left(\sum_{i\in\mN}\rho_{t-1}^{(i)}
    \big(Y_t^{(i)}+R_t^{(i)}(\tilde{Y}_{t+1:T}^{(i)})\big)
\right) \\
 \leq \underset{\left\{
    \Y_{t}~|~\{\Y_{t}, \tilde{\Y}_{t+1:T}\}\in\IR_{t-1}
        \cap\A_{t:T}^C
    \right\}} \inf \,
& \psi_{t-1}\left(\sum_{i\in\mN}\rho_{t-1}^{(i)}
    \big(Y_t^{(i)}+R_t^{(i)}(\tilde{Y}_{t+1:T}^{(i)})\big)
\right),
\end{align*}
}

\noindent where the second equality holds since \eqref{eq:r_equal} holds for $s=t$. Moreover,
{\small
\begin{align*}
\underset{\left\{
    \Y_{t}~|~\{\Y_{t}, \tilde{\Y}_{t+1:T}\}\in\IR_{t-1}
        \cap\A_{t:T}^C
    \right\}} \inf \,
 \psi_{t-1}\left(\sum_{i\in\mN}\rho_{t-1}^{(i)}
    \big(Y_t^{(i)}+R_t^{(i)}(\tilde{Y}_{t+1:T}^{(i)})\big)
\right)
\leq \psi_{t-1}\left(\sum_{i\in\mN}\rho_{t-1}^{(i)}
    \big(\tilde{Y}_t^{(i)}+R_t^{(i)}(\tilde{Y}_{t+1:T}^{(i)})\big)
\right)& \\
 = \psi_{t-1}\left(\sum_{i\in\mN}\rho_{t-1}^{(i)}
    \big(Y_t^{*(i)}+R_t^{(i)}(Y_{t+1:T}^{*(i)})\big)
\right)&,
\end{align*}
}

\noindent where the equality holds since \eqref{eq:r_equal} holds for $s=t-1$. Thus, both infima are equal to 
\[
\psi_{t-1}\left(
    \sum_{i\in\mN}\rho_{t-1}^{(i)}
        \big(Y_t^{*(i)}+R_t^{(i)}(Y_{t+1:T}^{*(i)})\big)
    \right).
\]

\medskip

\noindent \underline{Proof of $(ii)$ - Part 2:} For any $\Y_{1:T}^*\in\DPO$, let $\tilde{\Y}_{1:T}\in\A_{1:T}^C$ such that \eqref{eq:r_equal} holds for all $i\in\mN$ and $s\in\mfT$. Then, for all $t\in\T$, 
\[
\psi_{t-1}\left(
    \sum_{i\in\mN}\rho_{t-1}^{(i)}
        \big(Y_t^{*(i)}+R_t^{(i)}(Y_{t+1:T}^{*(i)})\big)
    \right) = 
\psi_{t-1}\left(
    \sum_{i\in\mN}\rho_{t-1}^{(i)}
        \big(\tilde{Y}_t^{(i)} + R_t^{(i)}(\tilde{Y}_{t+1:T}^{(i)})\big)
    \right),
\]

\noindent where $\psi_{t-1}$ is as defined in $(iv)$. Furthermore, $(iv)$ implies that 
\begin{align*}
\psi_{t-1}\bigg(
    \sum_{i\in\mN}\rho_{t-1}^{(i)}
        \big(\tilde{Y}_t^{(i)} + &  R_t^{(i)}(\tilde{Y}_{t+1:T}^{(i)})\big)
    \bigg) = \\
&\underset{\left\{
    \Y_{t}~|~\{\Y_{t}, \tilde{\Y}_{t+1:T}\}\in\IR_{t-1}
        \cap\A_{t:T}^C
    \right\}} \inf \,
\psi_{t-1}\left(
    \sum_{i\in\mN}\rho_{t-1}^{(i)}
        \big(Y_t^{(i)}+R_t^{(i)}(\tilde{Y}_{t+1:T}^{(i)})\big)
    \right),
\end{align*}

\noindent for all $t\in\T$. Since $\psi_{t-1}$ is strictly monotone on $\aleph$ for all $t\in\T$, it is strictly monotone on $\aleph^C$. Hence, by Proposition \ref{prop:cpot_inf}, $\tilde{\Y}_{t:T}\in\CPO_t$ for all $t\in\T$. Therefore, $\tilde{\Y}_{1:T}\in\CDPO$ by Lemma \ref{lemma:cdpo}. 

\medskip

\noindent \underline{Proof of $(iii)$:} Let $\Y_{1:T}^*\in\CDPO$. Assume, by way of contradiction, that $\Y_{1:T}^*\notin\DPO$. Since $\Y_{1:T}^*\in\CDPO\subseteq\DIR$, this implies that there exists $\hat{Y}_{1:T}\in\DIR$ such that 
\begin{align}
&\P\left(R_{t}^{(i)}(\hat{Y}_{t+1:T}^{(i)}) 
    \leq R_{t}^{(i)}(Y_{t+1:T}^{*(i)})
\right) = 1
    \ \text{for all} \ i\in\mN \ \text{and} \ t\in\mfT \ \text{and} \
\nonumber \\
& \P\left(R_{t}^{(i)}(\hat{Y}_{t+1:T}^{(i)}) 
    < R_{t}^{(i)}(Y_{t+1:T}^{*(i)})
\right) > 0 
    \ \text{for some} \ i\in\mN \ \text{and} \ t\in\mfT.
\label{eq:proof_dpo_cdpo_5}
\end{align}

\noindent Let $\tilde{\Y}_{1:T}\in\A_{1:T}^C$ be a comonotone improvement of $\hat{\Y}_{1:T}$ (which exists by Lemma \ref{lemma:com_improvement}). By the arguments in Lemma \ref{lemma:cdpo}, we have, 
\begin{equation*}
R_{s}^{(i)}(\tilde{Y}_{s+1:T}^{(i)})
    \leq 
R_{s}^{(i)}(\hat{Y}_{s+1:T}^{(i)})\leq R_{s}^{(i)}(Y_{s+1:T}^{*(i)}),
\end{equation*}

\noindent for all $s\in\mfT$ and $i\in\mN$. Furthermore, for $i\in\mN$ and $t\in\mfT$ such that \eqref{eq:proof_dpo_cdpo_5} holds, 
\[
\P\left(R_{t}^{(i)}(\tilde{Y}_{t+1:T}^{(i)}) 
    < R_{t}^{(i)}(Y_{t+1:T}^{*(i)})
\right) \geq \ 
\P\left(R_{t}^{(i)}(\hat{Y}_{t+1:T}^{(i)}) 
    < R_{t}^{(i)}(Y_{t+1:T}^{*(i)})
\right) 
> \ 0.
\]

\noindent Since $\tilde{\Y}_{1:T}\in\DIR\cap\A_{1:T}^C$, this contradicts the assumption that $\Y_{1:T}^*\in\CDPO$.  

\medskip

\noindent \underline{Proof of $(v)$:} Let $\Y_{1:T}^*\in\CDPO$. Then, by Lemma \ref{lemma:cdpo}, $\Y_{t:T}^*\in\CPO_t$ for all $t\in\T$. For $t\in\T$, let $\psi_{t-1}$ be functional given by Proposition \ref{prop:cpot_inf} for the allocation $\Y_{t:T}^*$. Since $\IR_{t-1}\cap\A_{t:T}^C\subseteq\IR_{t-1}$ for all $t\in\T$, we have 
\begin{align*}
\underset{\left\{
    \Y_{t}~|~\{\Y_{t}, \Y_{t+1:T}^{*}\}\in\IR_{t-1}
\right\}} \inf \, 
    & \psi_{t-1}\left(\sum_{i\in\mN}\rho_{t-1}^{(i)}
        \big(Y_t^{(i)}+R_t^{(i)}(Y_{t+1:T}^{*(i)})\big)
    \right) \leq \\
& \underset{\left\{
    \Y_{t}~|~\{\Y_{t}, \Y_{t+1:T}^{*}\}\in\IR_{t-1}\cap\A_{t:T}^C
\right\}} \inf \, 
    \psi_{t-1}\left(\sum_{i\in\mN}\rho_{t-1}^{(i)}
        \big(Y_t^{(i)}+R_t^{(i)}(Y_{t+1:T}^{*(i)})\big)
    \right).
\end{align*}

\noindent Assume, by way of contradiction, that there exists $t\in\T$ such that the inequality is strict. Then, there exists $\{\hat{\Y}_t, \Y_{t+1:T}^*\}\in\IR_{t-1}$ such that 
\begin{equation} \label{eq:proof_dpo_cdpo_4}
\psi_{t-1}\left(\sum_{i\in\mN}\rho_{t-1}^{(i)}
    \big(\hat{Y}_t^{(i)}+R_t^{(i)}(Y_{t+1:T}^{*(i)})\big)\right) < 
\psi_{t-1}\left(\sum_{i\in\mN}\rho_{t-1}^{(i)}
    \big(Y_t^{*(i)}+R_t^{(i)}(Y_{t+1:T}^{*(i)})\big)\right).
\end{equation}

\noindent Moreover, as $\{\hat{\Y}_t, \Y_{t+1:T}^*\}\in\A_t$ and $\Y_{t+1:T}^*\in\A_{t+1:T}^C$, by Corollary \ref{cor:com_improvement_2}, there exists $\tilde{\Y}_t\in\A_t$ such that $\{\tilde{\Y}_{t}, \Y_{t+1:T}^*\}\in\A_{t:T}^C$ and 
\[
\tilde{Y}_t^{(i)}+R_t^{(i)}(Y_{t+1:T}^{*(i)}) \cvx \hat{Y}_t^{(i)}+R_t^{(i)}(Y_{t+1:T}^{*(i)}) \ \text{for all} \ i\in\mN.
\]

\noindent By Assumption \ref{assumption:convex_order}, for all $i\in\mN$, we have
\[
\rho_{t-1}^{(i)}\big(\tilde{Y}_t^{(i)}+R_t^{(i)}(Y_{t+1:T}^{*(i)})\big)     \leq
\rho_{t-1}^{(i)}\big(\hat{Y}_t^{(i)}+R_t^{(i)}(Y_{t+1:T}^{*(i)})\big).
\]

\noindent Hence,
\begin{align*}
\underset{\left\{
    \Y_{t}~|~\{\Y_{t}, \Y_{t+1:T}^{*}\}\in\IR_{t-1}\cap\A_{t:T}^C
\right\}} \inf \, 
& \psi_{t-1}\left(\sum_{i\in\mN}\rho_{t-1}^{(i)}
    \big(Y_t^{(i)}+R_t^{(i)}(Y_{t+1:T}^{*(i)})\big)
\right) \\
\leq \ & \psi_{t-1}\left(\sum_{i\in\mN}\rho_{t-1}^{(i)}
    \big(\tilde{Y}_t^{(i)}+R_t^{(i)}(Y_{t+1:T}^{*(i)})\big)
\right) \\
\leq \ & \psi_{t-1}\left(\sum_{i\in\mN}\rho_{t-1}^{(i)}
    \big(\hat{Y}_t^{(i)}+R_t^{(i)}(Y_{t+1:T}^{*(i)})\big)
\right) \\
< \ & \psi_{t-1}\left(\sum_{i\in\mN}\rho_{t-1}^{(i)}
    \big(Y_t^{*(i)}+R_t^{(i)}(Y_{t+1:T}^{*(i)})\big)
\right) \\
= \underset{\left\{
    \Y_{t}~|~\{\Y_{t}, \Y_{t+1:T}^{*}\}\in\IR_{t-1}\cap\A_{t:T}^C
\right\}} \inf \, 
& \psi_{t-1}\left(\sum_{i\in\mN}\rho_{t-1}^{(i)}
    \big(Y_t^{(i)}+R_t^{(i)}(Y_{t+1:T}^{*(i)})\big)
\right),
\end{align*}

\noindent where the first inequality holds since $\{\tilde{\Y}_{t}, \Y_{t+1:T}^*\}\in\IR_{t-1}\cap\A_{t:T}^C$, the second inequality holds since $\psi_{t-1}$ is monotone, and the strict inequality follows from \eqref{eq:proof_dpo_cdpo_4}. This is a contradiction, so \eqref{eq:dpo_cdpo} holds for all $t\in\T$.
\end{proof}

\medskip

Statement $(i)$ of Theorem \ref{theorem:dpo_cdpo} states that comonotone Pareto optima exist if and only if Pareto optima exist. Moreover, if we view $\psi$ as a welfare function, then $(ii)$, $(iv)$, and $(v)$ implies that agents are not worse-off when they are restricted to the comonotone market. Furthermore, $(iii)$ states that all comonotone Pareto optima are PO, and hence we can find Pareto optima by searching for comonotone Pareto optima instead. In the next result, we show that all Pareto optima are comonotone when agent preferences strictly preserve the convex order. 

\medskip

\begin{corollary} \label{cor:dpo_cdpo}
Suppose that Assumption \ref{assumption:convex_order} holds. Then, $\DPO\cap\A_{1:T}^C=\CDPO$. Furthermore, if $\rho_t^{(i)}$ is strictly convex order preserving for all $i\in\mN$ and $t\in\mfT$, then $\DPO=\CDPO$. 
\end{corollary}

\medskip

\begin{proof}
We first prove that $\DPO\cap\A_{1:T}^C\subseteq\CDPO$. Let $\Y_{1:T}^*\in\DPO\cap\A_{1:T}^C$. Then, there does not exist $\Y_{1:T}\in\DIR$ such that for all $i\in\mN$ and $t\in\T$, 
\[
\rho_{t-1,T}^{(i)}(Y_{t:T}^{(i)}) \leq 
    \rho_{t-1,T}^{(i)}(Y_{t:T}^{*(i)}),
\] 

\noindent and $\P\big(\rho_{t-1,T}^{(i)}(Y_{t:T}^{(i)})<\rho_{t-1,T}^{(i)}(Y_{t:T}^{*(i)})\big)>0$ for some $i\in\mN$ and $t\in\T$. Hence, no such allocation exists in $\DIR\cap\A_{1:T}^C$ either. Since $\Y_{1:T}^*\in\DIR\cap\A_{1:T}^C$, it follows that $\Y_{1:T}^*\in\CDPO$.  

\smallskip

Next, we prove that $\CDPO\subseteq\DPO\cap\A_{1:T}^C$. The fact that $\CDPO\subseteq\DPO$ follows from Theorem \ref{theorem:dpo_cdpo} $(iii)$ and any $c$-PO allocation is comonotone by definition. Hence, $\CDPO\subseteq\DPO\cap\A_{1:T}^C$. 

\smallskip

Finally, we prove that if $\rho_t^{(i)}$ is strictly convex order preserving for all $i\in\mN$ and $t\in\mfT$, then $\DPO\cap\A_{1:T}^C=\DPO$. Let $\Y_{1:T}\in\DPO$ and assume, by way of contradiction, that $\Y_{1:T}\in\A_{1:T}\setminus\A_{1:T}^C$. Since $\Y_{1:T}\notin\A_{1:T}^C$, there exists $\tau\in\T$ such that $\Y_{\tau:T}\notin\A_{\tau:T}^C$. Let $t\in\T$ be the largest time index such that $\Y_{t:T}\notin\A_{t:T}^C$. By Corollary \ref{cor:com_improvement_2}, there exists $\tilde{\Y}_{t}\in\A_t$ such that $\{\tilde{\Y}_t, \Y_{t+1:T}\}\in\A_{t:T}^C$ and
\begin{equation*}
\tilde{Y}_t^{(i)}+R_t^{(i)}(Y_{t+1:T}^{(i)}) \cvx 
    Y_t^{(i)}+R_t^{(i)}(Y_{t+1:T}^{(i)}),
\end{equation*}

\noindent for all $i\in\mN$. Moreover, as $\Y_{t:T}\notin\A_{t:T}^C$, $\tilde{\Y}_{t}$ can be chosen such that there exists $j\in\mN$ with
\begin{equation} \label{eq:strict_ineq}
\tilde{Y}_t^{(j)}+R_t^{(j)}(Y_{t+1:T}^{(j)}) \scvx 
    Y_t^{(j)}+R_t^{(j)}(Y_{t+1:T}^{(j)}).
\end{equation}

\noindent Let $\hat{\Y}_{1:T}$ be the process $\Y_{1:T}$ with the $t$-th random vector replaced by $\tilde{\Y}_t$. Then, for all $i\in\mN$ and $s\geq t+1$, we have,
\[
\rho_{s-1}^{(i)}
    \big(\hat{Y}_s^{(i)}+R_s^{(i)}(\hat{Y}_{s+1:T}^{(i)})\big) =
\rho_{s-1}^{(i)}\big(Y_s^{(i)}+R_s^{(i)}(Y_{s+1:T}^{(i)})\big).
\]

\noindent Furthermore, for all $i\in\mN$ and $1\leq s\leq t$, we have
\[
\rho_{s-1}^{(i)}
    \big(\hat{Y}_s^{(i)}+R_s^{(i)}(\hat{Y}_{s+1:T}^{(i)})\big) \leq 
\rho_{s-1}^{(i)}\big(Y_s^{(i)}+R_s^{(i)}(Y_{s+1:T}^{(i)})\big),
\]

\noindent where the inequality follows from Assumption \ref{assumption:convex_order} for $s=t$ and monotonicity of $\rho_s^{(i)}$ for $s\leq t-1$. Moreover, as $\Y_{1:T}\in\DPO$, it follows that $\hat{\Y}_{1:T}\in\DIR$. Finally, for $j\in\mN$ such that \eqref{eq:strict_ineq} holds, $\P\big(\rho_{t-1,T}^{(j)}(\hat{Y}_{t:T}^{(j)})<\rho_{t-1,T}^{(j)}(Y_{t:T}^{(j)})\big)>0$, so $\Y_{1:T}\notin\DPO$, which is a contradiction. 
\end{proof}

\medskip

Corollary \ref{cor:dpo_cdpo} states that any PO allocation that is also comonotone is $c$-PO. Furthermore, if all agent preferences strictly preserve the convex order, then all PO allocations must be comonotone. Thus, it suffices to characterize comonotone Pareto optima when agent preferences preserve the convex order. 

\medskip

\subsection{Explicit Characterization of Comonotone Dynamic Pareto Optima} \label{sec:cdpo_ex_char}

In this section, we derive an explicit characterization of $c$-PO allocations. Before stating our main result, we introduce some additional notation. For an allocation process $\Y_{1:T}$, we let $\bR_t(\Y_{t+1:T}):=\sum_{i\in\mN}R_t^{(i)}(Y_{t+1:T}^{(i)})$ denote the total risk-to-go at time $t$ and denote by $\underline{r}_t:=\essinf\{\bR_t\}$ its essential infimum. Whenever there is no confusion, we omit the dependence on $\Y_{t+1:T}$ and simply write $\bR_t$. Moreover, we define the set $\G$ as follows:
\[
\G:=\left\{\{g^{(i)}\}_{i\in\mN}~|~g^{(i)}:\R^+\to\R^+ 
    \ \text{is non-decreasing and}\ \sum_{i\in\mN} g^{(i)} = \text{Id}
\right\},
\]
where Id denotes the identity function. In the static case, a comonotone allocation can always be written as the sum of functions in $\G$ and a constant (see Lemma 3.13 in \citet{ghossoub2026efficiency}). We show in the following lemma that a similar representation exists in the dynamic case for the sum of an allocation and the risk-to-go. 

\medskip

\begin{lemma} \label{lemma:com_allocation} For $t\in\T$, let $\Y_{t:T}\in\A_{t:T}^C$ and $\underline{s}_{t}:=\essinf\{S_{t}\}$. Then, there exists functions $\{g_{t}^{(i)}\}_{i\in\mN}\in\G$ and constants $\{c_{t}^{(i)}\}_{i\in\mN}\in\R^n$ such that for all $i\in\mN$, 
\begin{equation} \label{eq:com_allocation}
Y_{t}^{(i)}+R_{t}^{(i)}(Y_{t+1:T}^{(i)}) = 
    g_{t}^{(i)}\left(S_{t}+\bR_{t}-\underline{s}_{t}-
        \underline{r}_{t}\right) 
    + c_{t}^{(i)},
\end{equation}

\noindent and $\sum_{i\in\mN}c_{t}^{(i)}=\underline{s}_{t}+\underline{r}_{t}$. 
\end{lemma}

\medskip

\begin{proof}
By definition, the random vector $\big(Y_{t}^{(1)}+R_{t}^{(1)}(Y_{t+1:T}^{(1)}), \ldots, Y_{t}^{(n)}+R_{t}^{(n)}(Y_{t+1:T}^{(n)})\big)$ is a comonotone allocation (in the classical sense) of $S_{t}+\bR_{t}$. Thus, the desired result follows from Lemma 3.13 from \citet{ghossoub2026efficiency}. 
\end{proof}

\medskip

We refer to the functions $g_{t}^{(i)}$ in Lemma \ref{lemma:com_allocation} as retention functions since they represent the portion of the risk $S_{t}$ and the total risk-to-go (at time $t$) that is retained by the $i$-th agent. To explicitly characterize the retention functions for a PO allocation, we first recall the definition of a Choquet integral.

\medskip

\begin{definition} 
\label{definition:choquet}
A Choquet integral with respect to a distortion function $k$ of the probability measure $\P$ is the mapping  $I_k: L_{t}^{\infty} \to \R$ given by 
\begin{equation} \label{eq:choquet}
I_k(X) 
:=
\int_{0}^{\infty} k\left(\mathbb{P}(X\geq x)\right)\, dx
+
\int_{-\infty}^{0}\left[k\left(\mathbb{P}(X\geq x)\right)-1\right]\, dx, \ \text{for all} \ X \in L_{t}^{\infty},
\end{equation}
where $k: [0,1]\to [0,1]$ is a non-decreasing function that satisfies $k(0)=0$ and $k(1) = 1$. The function $k$ is called the distortion function associated with $I_k$. 
\end{definition}

Many properties of a Choquet integral are determined by its distortion function $k$. In particular, they are law-invariant, and if $k$ is concave, then $I_k$ is a coherent risk measure. We refer to \cite{denneberg1994non} or \cite{marinacci2004introduction} for a detailed discussion of Choquet integrals and their properties. Furthermore, a law-invariant coherent risk measure admits a convenient representation as the supremum of Choquet integrals. We refer to \citet{ghossoub2026efficiency} Lemma 3.12 for a proof of this result in the context of law invariant and positively homogeneous monetary utilities. 

\medskip

\begin{lemma} \label{lemma:rep} Let $\rho:\L_t^{\infty}\to\R$ be a law invariant coherent risk measure. Then, for all $Z\in\L_t^{\infty}$, we have 
\[
\rho(Z)=\sup_{k\in\K}I_k(Z),
\]

\noindent where $\K$ is a convex set of concave distortion functions that is sequentially closed under pointwise convergence.
\end{lemma}

Lemma \ref{lemma:rep} states that any law invariant coherent risk measure can be written as the supremum of Choquet integrals and we refer to \citet{ghossoub2026efficiency} for a more explicit characterization of the set $\K$. We are now ready to discuss the solutions to the $\psi$ inf-convolution problem in Proposition \ref{prop:cpot_inf}. As discussed earlier, since the expected value is strictly monotone, we will focus on expected value inf-convolution problems of the form
\begin{equation}\label{eq:proof_char_1}
\underset{
    \left\{\Y_{t}~\big|~\{\Y_{t}, \Y_{t+1:T}^{*}\}       
        \in\IR_{t-1}\cap\A_{t:T}^C \right\}} 
    \inf \, 
\E\left[\sum_{i\in\mN}\rho_{t-1}^{(i)}
    \big(Y_{t}^{(i)}+R_{t}^{(i)}(Y_{t+1:T}^{*(i)})\big)
\right],
\end{equation}

\medskip

\noindent for $t\in\T$ and $\Y_{t+1:T}^*\in\A_{t+1:T}^C$. 

\medskip

\begin{theorem} \label{theorem:char_cdpo2} Suppose that Assumption \ref{assumption:convex_order} holds and for each $i\in\mN$, the sequence of one-step conditional risk measures $\{\rho^{(i)}_t\}_{t\in\mfT}$ is coherent and equidistribution-preserving. For $t\in\T$, let $\underline{s}_{t}:=\essinf\{S_{t}\}$. Then, the following hold:

\begin{enumerate}[label = $(\roman*)$]
\item If the infimum \eqref{eq:proof_char_1} is attained by $\Y_{t}^*$, then for all $i\in\mN$,
\begin{equation}\label{eq:proof_char_com}
Y_{t}^{*(i)} = g_{t}^{*(i)}\left(
    S_{t}+\bR_{t}-\underline{r}_{t}-\underline{s}_{t}
\right) 
    - R_{t}^{(i)}(Y_{t+1:T}^{*(i)}) 
+ c_{t}^{*(i)},
\end{equation}

\noindent where $\{g_{t}^{*(i)}\}_{i\in\mN}\in\G$ solves
\begin{equation}\label{eq:proof_char_4}
\underset{\left\{\{g_{t}^{(i)}\}_{i\in\mN}\in\G
    \right\}} \inf \,
\E\left[\rho_{t-1}^{(i)}\big(
    g_{t}^{(i)}(S_{t}+\bR_{t}-\underline{r}_{t}-                \underline{s}_{t})\big)
    \right], 
\end{equation}

\noindent and $\{c_{t}^{*(i)}\}_{i\in\mN}\in\R^n$ is chosen such that $\sum_{i\in\mN}c_{t}^{*(i)}=\underline{r}_{t}+\underline{s}_{t}$ and for all $i\in\mN$,
\begin{equation} \label{eq:char_c}
c_{t}^{*(i)} \leq 
    \essinf\left\{\rho_{t-1}^{(i)}\Big(
        X_{t}^{(i)}+R_{t}^{(i)}(X_{t+1:T}^{(i)})
    \Big)-\rho_{t-1}^{(i)}\Big(
        g_{t}^{*(i)}\big(S_{t}+\bR_{t}-\underline{r}_{t}-\underline{s}_{t}\big)
    \Big)\right\}.
\end{equation}
\item A necessary condition for $\{g_{t}^{*(i)}\}_{i\in\mN}$ to solve \eqref{eq:proof_char_4} is $g_{t}^{*(i)}(x)=\int_{0}^{x}h_{t}^{(i)}(z)\,dz$ for $h_{t}^{(i)}:\R^+\to [0,1]$ satisfying
\[
\sum_{i\in L\left(x, k_{t}^{*(1)}, \ldots, k_{t}^{*(n)}\right)}
    h_{t}^{(i)}(x)=1 \quad \text{and} \quad
\sum_{i\in \mN\setminus L\left(
    x, k_{t}^{*(1)}, \ldots, k_{t}^{*(n)}\right)}
h_{t}^{(i)}(x)=0,
\]

\noindent where 
\begin{align*}
 L(x, k_{t}^{*(1)}, \ldots, k_{t}^{*(n)}):=
\biggl\{i\in\mN~\Big|~k_{t}^{*(i)}\Big(
    & \P\left(S_{t}+            
        \bR_{t}>x+\underline{r}_{t}+\underline{s}_{t}\right)
    \Big) = \\
 & \underset{k\in\mN} \min \,
    \left\{k_{t}^{*(k)}\Big(
        \P\left(S_{t}+\bR_{t}>x+\underline{r}_{t}+
            \underline{s}_{t}\right)
        \Big)
    \right\}
\biggr\},
\end{align*}

\noindent and $(k_{t}^{*(1)}, \ldots, k_{t}^{*(n)})$ solves 
\begin{equation}\label{eq:opt_problem_cdpo}
\underset{
    \{k_{t}^{(i)}\}_{i\in\mN}\in\prod_{i\in\mN}\K_{t}^{(i)}
} \max \,
\int_{0}^{\infty}
    \underset{i\in\mN} \min \,
\left\{k_{t}^{(i)}\left(
    \P\left(                
        S_{t}+\bR_{t}>x+\underline{r}_{t}+\underline{s}_{t}
    \right)
\right)\right\}\,dx, 
\end{equation}

\noindent where for any $i\in\mN$, $\K_{t}^{(i)}$ is the set of distortion functions from Lemma \ref{lemma:rep} for $\E[\rho_{t-1}^{(i)}(\cdot)]$. 
\end{enumerate}
\end{theorem}

\medskip

\begin{proof}
Let $t\in\T$ and $\varrho_{t}^{(i)}(\cdot):=\E[\rho_{t-1}^{(i)}(\cdot)]$ for all $i\in\mN$.

\smallskip

\noindent \underline{Proof of $(i)$:} By linearity of expectation, we can rewrite \eqref{eq:proof_char_1} as
\begin{equation}\label{eq:proof_char_2}
\underset{
    \left\{\Y_{t}~\big|~\{\Y_{t}, \Y_{t+1:T}^{*}\}       
        \in\IR_{t-1}\cap\A_{t:T}^C \right\}} 
    \inf \, 
\sum_{i\in\mN}\varrho_{t}^{(i)}
    \left(Y_{t}^{(i)}+R_{t}^{(i)}(Y_{t+1:T}^{*(i)})\right).
\end{equation}

\noindent By Lemma \ref{lemma:com_allocation}, for any $\Y_{t}$ such that $\{\Y_{t}, \Y_{t+1:T}^*\}\in\A_{t:T}^C$, there exists functions $\{g_{t}^{(i)}\}_{i\in\mN}\in\G$ and constants $\{c_{t}^{(i)}\}_{i\in\mN}\in\R^n$ such that for all $i\in\mN$,
\begin{equation*}
Y_{t}^{(i)}+R_{t}^{(i)}(Y_{t+1:T}^{*(i)}) = 
    g_{t}^{(i)}(
        S_{t}+\bR_{t}-\underline{r}_{t}-\underline{s}_{t})
    + c_{t}^{(i)}, 
\end{equation*}

\noindent and $\sum_{i\in\mN}c_{t}^{(i)}=\underline{r}_{t}+\underline{s}_{t}$. Conversely, if $\{g_{t}^{(i)}\}_{i\in\mN}\in\G$ and $\{c_{t}^{(i)}\}_{i\in\mN}\in\R^n$ where $\sum_{i\in\mN}c_{t}^{(i)}=\underline{r}_{t}+\underline{s}_{t}$, then the random vector 
\[
\Z_{t}' := \left(
    g_{t}^{(1)}(S_{t}+\bR_{t}-\underline{r}_{t}-        
        \underline{s}_{t})+c_{t}^{(1)}, \ldots, 
    g_{t}^{(n)}(S_{t}+\bR_{t}-\underline{r}_{t}-
        \underline{s}_{t})+c_{t}^{(n)}\right),
\] 
is comonotone (in the static sense) and the random vector
\[
\Y_{t}' := \left(
    Z_{t}'^{(1)}-R_{t}^{(1)}(Y_{t+1:T}^{*(1)}), 
        \ldots, 
    Z_{t}'^{(n)}-R_{t}^{(n)}(Y_{t+1:T}^{*(n)})
\right),
\]

\noindent is an allocation of $S_{t}$. Hence $\{\Y_{t}', \Y_{t+1:T}^*\}\in\A_{t:T}^C$. Thus, we can rewrite \eqref{eq:proof_char_2} as
\begin{equation}\label{eq:proof_char_3}
\underset{
    \left\{\{g_{t}^{(i)}\}_{i\in\mN}\in\G, \, 
        \{c_{t}^{(i)}\}_{i\in\mN}\in 
            \B\left(\{g_{t}^{(i)}\}_{i\in\mN}\right)
        \right\}
    } \inf \, 
\sum_{i\in\mN}\varrho_{t}^{(i)}
    \left(g_{t}^{(i)}(S_{t}+\bR_{t}-\underline{r}_{t}-
        \underline{s}_{t})+c_{t}^{(i)}\right),
\end{equation}

\noindent where 
\begin{align*}
\B\left(\{g_{t}^{(i)}\}_{i\in\mN}\right) & := 
    \biggl\{\{c_{t}^{(i)}\}_{i\in\mN}~\Big|~
        \sum_{i\in\mN}c_{t}^{(i)} =          
            \underline{r}_{t}+\underline{s}_{t} 
        \ \text{and} \ \\
    &\biggl\{\left(
        g_{t}^{(1)}(S_{t}+\bR_{t}-\underline{r}_{t}-
            \underline{s}_{t})+c_{t}^{(1)}-R_{t}^{(1)}
                (Y_{t+1:T}^{*(1)}), \ldots, \right. \\
& \left. \qquad\qquad 
    g_{t}^{(n)}(S_{t}+\bR_{t}-\underline{r}_{t}-    
        \underline{s}_{t})+c_{t}^{(n)}-R_{t}^{(n)}
            (Y_{t+1:T}^{*(n)})\right), 
         \, \Y_{t+1:T}^*\biggr\}
    \in\IR_{t-1}\biggr\}.
\end{align*}

\noindent By translation invariance of $\varrho_{t}^{(i)}$, we can further rewrite \eqref{eq:proof_char_3} as 
\begin{equation*}
\underset{
    \left\{\{g_{t}^{(i)}\}_{i\in\mN}\in\G, \, 
        \{c_{t}^{(i)}\}_{i\in\mN}\in 
            \B\left(\{g_{t}^{(i)}\}_{i\in\mN}\right)
        \right\}
    } \inf \, 
\sum_{i\in\mN}\varrho_{t}^{(i)}
    \left(g_{t}^{(i)}(S_{t}+\bR_{t}-\underline{r}_{t}-      
        \underline{s}_{t})\right) +  
    \sum_{i\in\mN}c_{t}^{(i)}. 
\end{equation*}

\noindent Since the first summation is independent of $c_{t}^{(i)}$, \eqref{eq:proof_char_4} holds. To prove \eqref{eq:char_c}, note that by definition of $\IR_{t-1}$, for any $i\in\mN$, $c_{t}^{*(i)}$ satisfies
\[
\rho_{t-1}^{(i)}\big(
    g_{t}^{*(i)}\big(S_{t}+\bR_{t}-\underline{r}_{t}-
        \underline{s}_{t}\big)+c_{t}^{*(i)}\big)
    \leq
\rho_{t-1}^{(i)}\big(
    X_{t}^{(i)}+R_{t}^{(i)}(X_{t+1:T}^{(i)})\big).
\]

\noindent Hence, \eqref{eq:char_c} follows by translation invariance of $\rho_{t-1}^{(i)}$. 

\smallskip

\noindent \underline{Proof of $(ii)$:} Since $\rho_{t-1}^{(i)}$ is equidistribution-preserving and the expected value is law invariant, $\varrho_{t}^{(i)}$ is law invariant for all $i\in\mN$. Furthermore, the coherence of $\rho_{t-1}^{(i)}$ and the expected value implies that $\varrho_{t}^{(i)}$ is also coherent for all $i\in\mN$. Thus, the necessary condition for $\{g_{t}^{*(i)}\}_{i\in\mN}$ follows from \citet{ghossoub2026efficiency} Corollary 4.2 $(ii)$.
\end{proof}

\medskip

Note that the existence of a solution to \eqref{eq:opt_problem_cdpo} is guaranteed by \citet{ghossoub2026efficiency} Corollary 4.2 $(i)$. Furthermore, the structure of allocations given by \eqref{eq:proof_char_com} have a natural interpretation in an insurance context. Each function $g_{t}^{*(i)}$ is the portion of the loss $S_{t}$ and the total risk-to-go at time $t$ that is retained by agent $i$. Similar to the static case, the constants $c_{t}^{*(i)}$ represent the premia paid by agent $i$ at time $t$ to participate in the reallocation scheme and these premia ensure that allocations are IR. Moreover, since the set 
\[
\left\{\left(
    g_{t}^{(1)}(S_{t}+\bR_{t}-\underline{r}_{t}-
        \underline{s}_{t}), \ldots, 
    g_{t}^{(n)}(S_{t}+\bR_{t}-\underline{r}_{t}-
        \underline{s}_{t})\right) ~\big |~ 
    \{g_{t}^{(i)}\}_{i\in\mN}\in\G\right\},
\]
is the set of comonotone allocations (in the static sense) of $S_{t}+\bR_{t}-\underline{r}_{t}-\underline{s}_{t}$ and the risk measures $\E[\rho_{t-1}^{(1)}(\cdot)], \ldots, \E[\rho_{t-1}^{(n)}(\cdot)]$ are law invariant and coherent, a solution to \eqref{eq:proof_char_4} always exists by \citet{filipovic2008optimal} Theorem 2.5. In other words, if we ignore any individual rationality constraints, then we can always find a comonotone allocation process $\Y_{1:T}^*$ that minimizes the expected total risk-to-go at all time points. Furthermore, condition $(ii)$ is not sufficient and we refer to \citet{ghossoub2026efficiency} Appendix B for a counterexample in the static case. However, since \eqref{eq:proof_char_4} always has a solution under the conditions of Theorem \ref{theorem:char_cdpo2}, we can find a solution to \eqref{eq:proof_char_4} by finding all of the allocations satisfying the necessary condition and computing the expected total risk-to-go for each candidate solution. The candidate solution that minimizes the expected total risk-to-go must be a solution to \eqref{eq:proof_char_4}. We discuss this algorithm in more detail in Section \ref{sec:algorithm}. 

\bigskip

\section{Numerical Example} \label{sec:num_ex}

In this section, we consider an illustrative two-period example. We start with a summary of the algorithmic approach used to find comonotone Pareto optima in Section \ref{sec:algorithm}. Then, we discuss a two-agent example in Section \ref{sec:examples} and explicitly solve for the PO retention structure. Finally, in Section \ref{sec:dpo_mpo}, we revisit the notion of myopic Pareto optimality from Section \ref{sec:myopic_po} and discuss its shortfalls.  

\medskip

\subsection{An Algorithmic Approach to Finding Dynamic Pareto Optima} \label{sec:algorithm}

In summary, we propose the following recursive approach for finding $c$-PO allocations:

\begin{enumerate}[label = $(\roman*)$]
\item For the terminal time $t=T$, find all solutions $(k_{t}^{*(1)}, \ldots, k_{t}^{*(n)})$ to \eqref{eq:opt_problem_cdpo}. 
\item For each solution $(k_{t}^{*(1)}, \ldots, k_{t}^{*(n)})$ to \eqref{eq:opt_problem_cdpo}, define the set 
\begin{align*}
L(x, k_{t}^{*(1)}, \ldots, k_{t}^{*(n)}):=
\biggl\{i\in\mN~\Big|~k_{t}^{*(i)}\Big(
    & \P\left(S_{t}+            
        \bR_{t}>x+\underline{r}_{t}+\underline{s}_{t}\right)
    \Big) = \\
 & \underset{k\in\mN} \min \,
    \left\{k_{t}^{*(k)}\Big(
        \P\left(S_{t}+\bR_{t}>x+\underline{r}_{t}+
            \underline{s}_{t}\right)
        \Big)
    \right\}
\biggr\}.
\end{align*}

This set represents the agents with the most optimistic assessment of the likelihood of the tail event $[S_{t}+\bR_{t}>x+\underline{r}_{t}+\underline{s}_{t}]$. Note that in the dynamic setting, the tail event depends on both the aggregate endowment at time $t$ and the risk-to-go at time $t$. 

\item Identify all $\{h_{t}^{(i)}\}_{i\in\mN}$ satisfying the necessary condition in Part $(ii)$ of Theorem \ref{theorem:char_cdpo2} for all solutions $(k_{t}^{*(1)}, \ldots, k_{t}^{*(n)})$. For each set of functions $\{h_{t}^{(i)}\}_{i\in\mN}$, define the candidate solution 
$$g_{t}^{(i)}(x):=\int_{0}^{x}h_{t}^{(i)}(z)\,dz, \ i\in\mN.$$ 

\item Choose the candidate solution $\{g_{t}^{(i)}\}_{i\in\mN}$ that minimizes the expected total risk-to-go \eqref{eq:proof_char_4}, and choose constants $\{c_{t}^{*(i)}\}_{i\in\mN}$ satisfying \eqref{eq:char_c}. 

\item Repeat the above steps for each of the previous time steps. 
\end{enumerate}

\medskip

Recall that the quantities $\bR_{t}$ and $\underline{r}_{t}$ (which appear in Step $(ii)$ of the algorithm) depend on the optimal allocation $\Y_{t+1:T}^*$. Hence, the proposed algorithm is recursive, and the optimal allocation process must be determined starting at the terminal time. Furthermore, by Theorem \ref{theorem:dpo_cdpo} $(iii)$, any allocation determined by this algorithm is also PO. 

\medskip

\subsection{Pareto-Optimal Allocations in a Two-Period Setting} \label{sec:examples} 

In this section, we study an example with $n=2$ agents and $T=2$ periods. Assume that $S_1$ follows an exponential distribution with mean $200$, and that the conditional distribution of $S_2$ given $S_1$ is exponential with mean $S_1$. We are interested in the retention structure of the agent allocations (i.e., the functions $g_{t}^{*(i)}$ given by Theorem \ref{sec:cdpo_ex_char}). Since the constants $c_{t}^{*(i)}$ only shift the retention functions and do not change the shape of the retention structure, we will not discuss them in our example.

To evaluate the risk at time 1, the first agent uses the Expected Shortfall (ES) at level $\alpha=0.9$ while the second agent uses the expected value. Recall that the ES at level $\alpha\in(0,1)$ is given by 
\[
\ES_{\alpha}(X):=\frac{1}{1-\alpha}\int_{\alpha}^{1}F_X^{-1}(t) \, dt,
\]

\noindent where $F_X^{-1}(\alpha):=\inf\{x\in\R| F_X(x)\geq\alpha\}$, $\alpha\in (0,1)$, is the (left-continuous) quantile function of the random variable $X$. Thus, the ES is a distortion risk measure with distortion weight function $\gamma(u):=\frac{1}{1-\alpha}\Id_{u>\alpha}$. To evaluate the risk at time 2, both agents will use the ES, but the parameter $\alpha$ will depend on the realization of $S_1$. In other words, the risk measure of agent 1 is given by
\[
\rho_{0:2}^{(1)}(Y_{1:2})=\ES_{0.9}\big(Y_1+\ES_{\alpha^{(1)}(S_1)}(Y_2)\big),
\]

\noindent where $\alpha^{(1)}(S_1)=\begin{cases} 0, & S_1\leq F_{S_1}^{-1}(0.2) \\ 0.9, & S_1>F_{S_1}^{-1}(0.2)\end{cases}.$ 

\smallskip

The risk measure of agent 2 is given by
\[
\rho_{0:2}^{(2)}(Y_{1:2})=\E[Y_1+\ES_{\alpha^{(2)}(S_1)}(Y_2)],
\]

\noindent where $\alpha^{(2)}(S_1)=\begin{cases} 0, & S_1\leq F_{S_1}^{-1}(0.6) \\ 0.99, & S_1>F_{S_1}^{-1}(0.6)\end{cases}.$

\medskip

In this example, agent 2 uses a less conservative risk measures to evaluate the risk at time 1, and hence we expect that they will retain all of the risk at time 1. For both agents, their assessment of the time 2 risk depends on the realized value of $S_1$. For the first agent, if the realization of $S_1$ is less than the $0.2$ quantile, then they will use the expected value to evaluate the risk at time 2. In contrast, if the realization of $S_1$ is large, then they will continue to use the ES with level $\alpha=0.9$. For the second agent, if the realization of $S_1$ is less than the $0.6$ quantile, then they will continue using the expected value. Otherwise, they will switch to the ES with level $\alpha=0.99$. In other words, agent 2 becomes more conservative than agent 1 if $S_1$ is large. Since $S_1$ and $S_2$ are positively correlated, we expect agent 1 to retain the tranche of the aggregate endowment $S_2$ when $S_2$ is large. The following calculations show that the PO retention structure aligns with this intuition.

We begin with Step $(i)$ of the algorithm for the terminal time $T=2$. That is, we calculate the distortion functions of the risk measures $\varrho_2^{(i)}(\cdot):=\E[\rho_{1}^{(i)}(\cdot)]$, for $i \in \{1,2\}$. For the first agent,
\begin{align*}
\varrho_2^{(1)}(X) &= \E[\ES_{\alpha^{(1)}(S_1)}(X)] \\
&= 0.2 \, \E[X] + 0.8 \, \ES_{0.9}(X) \\
&= 0.2\int_{0}^{1}F_{X}^{-1}(u) \, du+0.8\int_{0}^{1}10\, \Id_{u>0.9} \, F_{X}^{-1}(u) \, du \\
&= \int_{0}^{1}(0.2+8\, \Id_{u>0.9}) \, F_{X}^{-1}(u) \, du. 
\end{align*}

\noindent Hence, $\varrho_2^{(1)}(X)$ is a distortion risk measure with distortion weight function $\gamma_2^{(1)}(u)=0.2+8 \, \Id_{u>0.9}$. Since $\gamma_2^{(1)}(1-u)$ is the left derivative of the distortion function $k_2^{*(1)}$, we have 
\[
k_2^{*(1)}(u)=\begin{cases} 8.2 \, u, & 0\leq u\leq 0.1 \\ 0.2 \, u + 0.8, & 0.1<u\leq 1\end{cases}.
\]
A similar calculation for agent 2 yields
\[
k_2^{*(2)}(u)=\begin{cases} 40.6 \, u, & 0\leq u\leq 0.01 \\ 0.6 \, u+0.4, & 0.01<u\leq 1\end{cases}.
\]

\noindent The distortion functions for the two agents at time 2 are illustrated in \Cref{fig:1}. Note that since $\varrho_2^{(1)}(X)$ and $\varrho_2^{(2)}(X)$ are distortion risk measures, the sets $\K_{2}^{(1)}$ and $\K_{2}^{(2)}$ in Theorem \ref{theorem:char_cdpo2} are singletons. Moreover, the two functions intersect at $u^*:=\frac{0.4}{7.6}$, which corresponds to the point when the ordering of agent preferences changes.  

\begin{figure} [!htbp]
\centering
\includegraphics[scale=0.35]{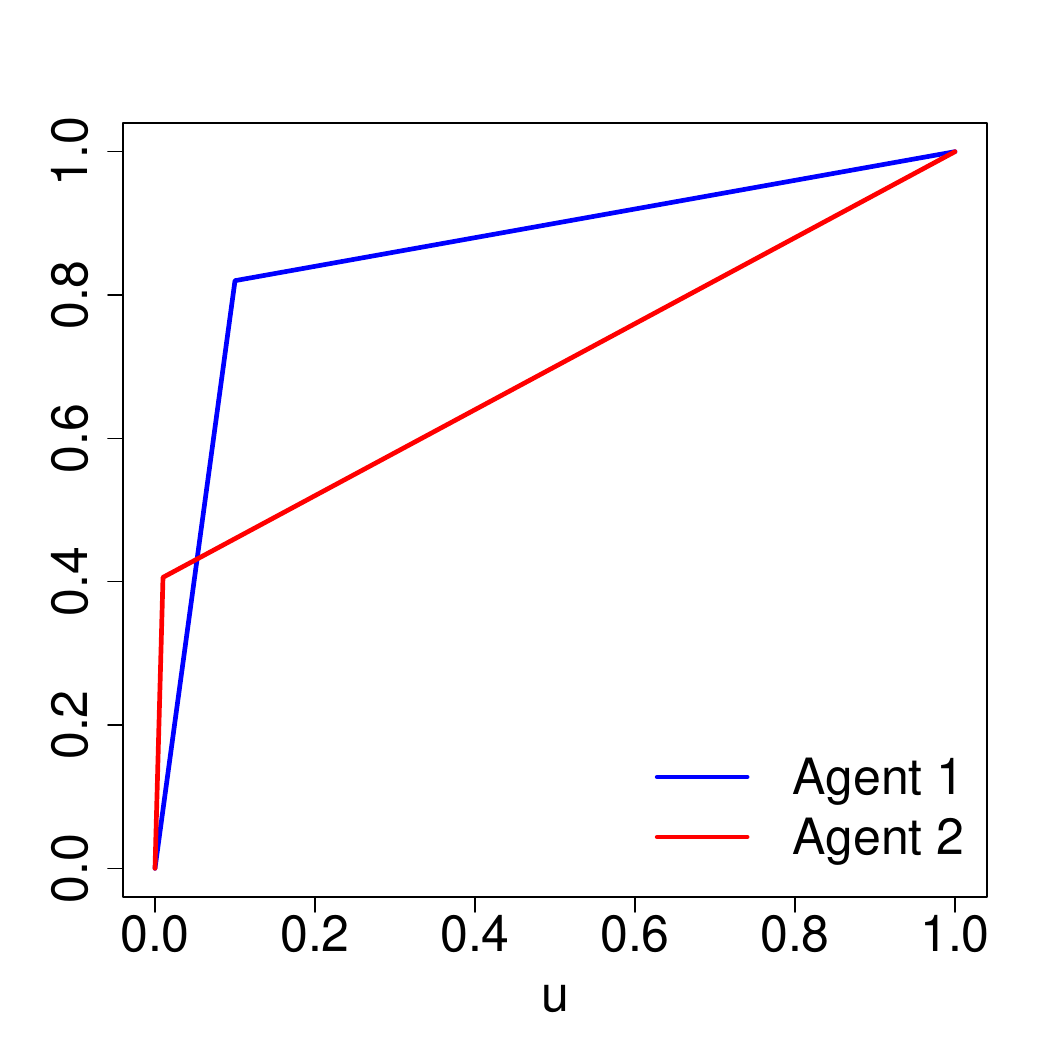}
\caption{Distortion functions $k_2^{*(1)}$ (blue) and $k_2^{*(1)}$ (red) at time 2.}\label{fig:1}
\end{figure}

Next, we study each agent's view towards the likelihood of the tail risk at time 2, which corresponds to Step $(ii)$ of the algorithm. Since time 2 is the terminal time, there is no risk-to-go, and the tail risk is simply the event $[S_2>x]$. We estimate the probabilities $\P(S_2>x)$ empirically using $100,000$ samples of $S_2$ since the marginal distribution of $S_2$ does not have a closed form expression for the survival function. From our sample, we obtain that $\P(S_2>780) \approx \frac{0.4}{7.6} = u^*$. Hence,
\[
L(x, k_2^{*(1)}, k_2^{*(2)}) = 
    \begin{cases} \{2\}, & x\leq 780 \\ 
        \{1\}, & x>780 
    \end{cases}.
\]

\noindent Each agent's assessment of tail risks at time 2 is illustrated in \Cref{fig:2}. 

Moreover, $h_2^{(1)}(z)=\Id_{z>780}$ and $h_2^{(2)}(z)=\Id_{z\leq 780}$, and the PO at time 2 retention functions are 
$$
g_2^{*(1)}(S_2) = \max\{S_2-780, 0\} 
\ \ \hbox{and} \ \ 
g_2^{*(2)}(S_2) = \min\{S_2, 780\}.
$$

\noindent Both retention functions are illustrated in \Cref{fig:3}. As expected, agent 1 retains the tranche of the aggregate endowment $S_2$ when $S_2$ is large. Moreover, like the static case, the retention function for agent $i$ increases linearly with slope 1 whenever $L(x, k_1^{*(1)}, k_1^{*(2)})=\{i\}$. That is, the retention function for agent $i$ is strictly increasing at $x$ if and only if agent $i$ is most optimistic about the likelihood of the tail risk $[S_2>x]$. Finally, since $L(x, k_1^{*(1)}, k_1^{*(2)})=\{i\}$ is a singleton for all $x\in\R$, the retention functions plotted are unique. That is, an allocation is $c$-PO at time 2 if and only if it has the retention structure shown in \Cref{fig:3}.  

\begin{figure} [!htbp]
\centering
\includegraphics[scale=0.35]{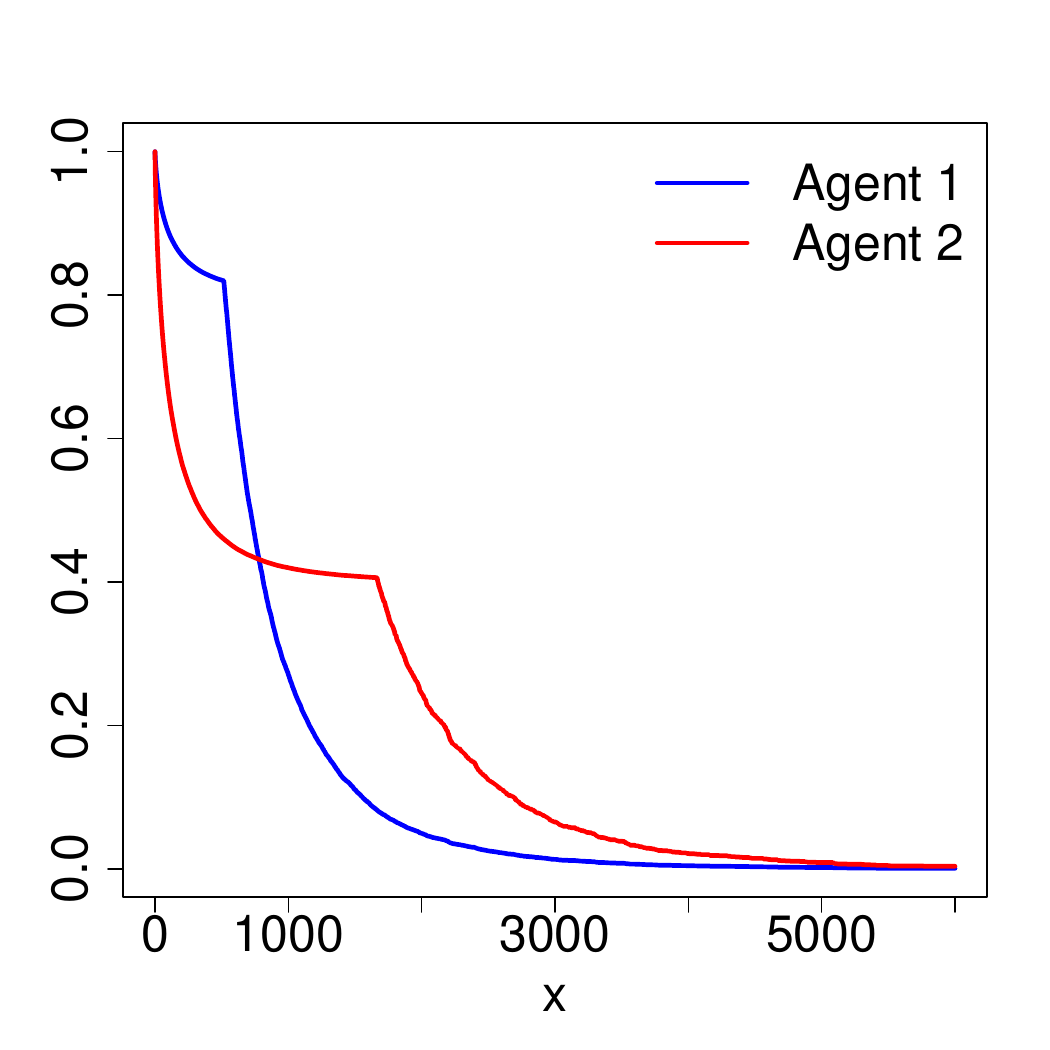}
\caption{Assessments of tail risks $k_2^{*(1)}\big(\P(S_2>x)\big)$ (blue) and $k_2^{*(2)}\big(\P(S_2>x)\big)$ (red) at time 2.}\label{fig:2}
\end{figure}

\begin{figure} [!htbp]
\centering
\includegraphics[scale=0.35]{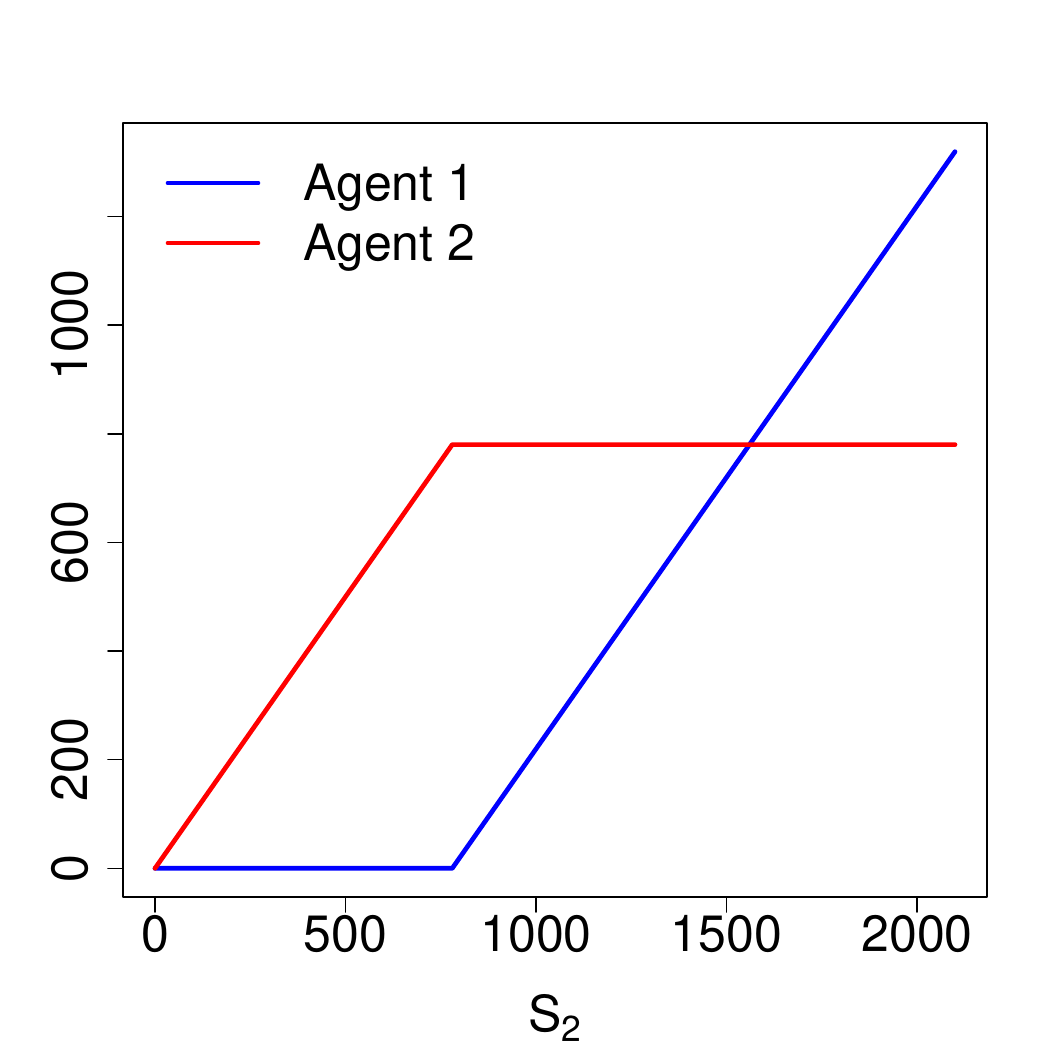}
\caption{PO at time 2 retention functions $g_2^{*(1)}(S_2)$ (blue) and $g_2^{*(2)}(S_2)$ (red).}\label{fig:3}
\end{figure}

We repeat the above steps for time 1. At time 1, the distortion functions for the agent 1 and 2 are $k_1^{*(1)}(u)=\min\{\frac{u}{0.1}, 1\}$ and $k_1^{*(2)}(u)=u$, respectively. Since $k_1^{*(1)}(u)>k_1^{*(2)}(u)$ for all $u\in (0,1)$, this implies that
\[
L(x, k_1^{*(1)}, k_1^{*(1)})=\{2\},  \ \text{for all} \ x\in\R.
\] 

\noindent Note that at time 1, the tail risk is $[S_1+\bR_1>x]$, where $\bR_1$ is the total risk-to-go at time 1. This contrasts with the tail risk in the static case, which is $[S_1>x]$. In the dynamic case, agents want to account for both the aggregate endowment at the current time point and the risk-to-go. To plot each agent's assessment of the likelihood of the tail risk, we first calculate the total risk-to-go at time 1. By definition of the risk-to-go, we have
\begin{align*}
\bR_1 &= \ES_{\alpha^{(1)}(S_1)}\left(g_2^{*(1)}(S_2)+c_2^{*(1)}\right)     + \ES_{\alpha^{(2)}(S_1)}\left(g_2^{*(2)}(S_2)+c_2^{*(2)}\right) \\
&=\begin{cases} 
    \E[g_2^{*(1)}(S_2)]+\E[g_2^{*(2)}(S_2)], & S_1\leq F_{S_1}^{-1}(0.2)      \\ 
    \ES_{0.9}\big(g_2^{*(1)}(S_2)\big)+\E[g_2^{*(2)}(S_2)], 
        & F_{S_1}^{-1}(0.2)< S_1\leq F_{S_1}^{-1}(0.6) \\
    \ES_{0.9}\big(g_2^{*(1)}(S_2)\big) + 
        \ES_{0.99}\big(g_2^{*(2)}(S_2)\big), & S_1>F_{S_1}^{-1}(0.6)
\end{cases} \\
&= \begin{cases} 200, & S_1\leq F_{S_1}^{-1}(0.2) 
        \\ 
    464.7, & F_{S_1}^{-1}(0.2)< S_1\leq F_{S_1}^{-1}(0.6) 
        \\
    1074.1, & S_1>F_{S_1}^{-1}(0.6)
\end{cases},
\end{align*}

\noindent where the first equality follows from the fact that $c_2^{*(1)}+c_2^{*(2)}=\underline{r}_2+\underline{s}_2$ and $\underline{r}_2=\underline{s}_2=0$, and the third equality is computed numerically using the same $100,000$ samples of $S_2$ used in the time 2 calculation. Using the above risk-to-go calculation and $100,000$ samples of $S_1$, we can numerically estimate the probabilities $\P(S_1+\bR_1>x)$, which we use to plot each agent's assessment of the likelihood of the tail risk in \Cref{fig:4}. 

\begin{figure} [!htbp]
\centering
\includegraphics[scale=0.35]{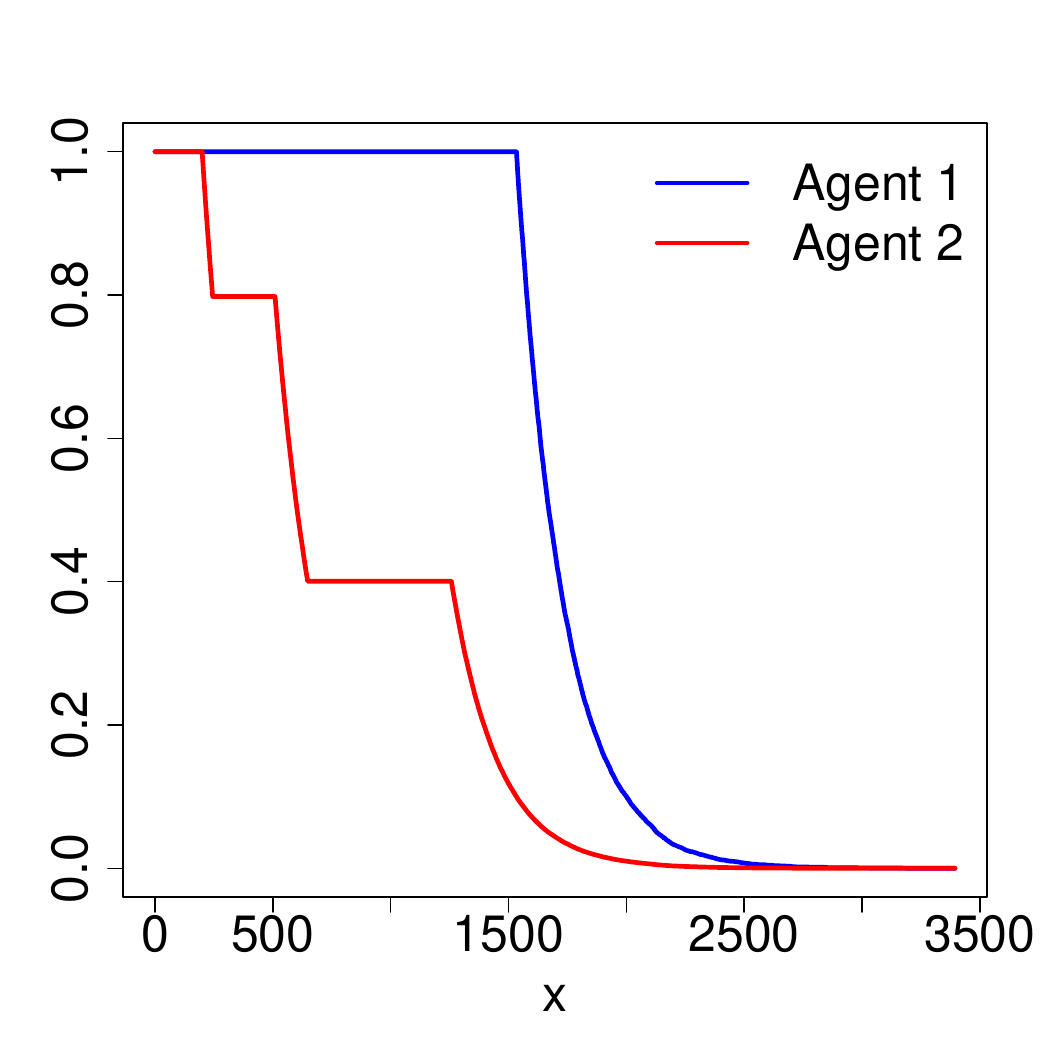}
\caption{Assessments of tail risks $k_1^{*(1)}\big(\P(S_1+\bR_1>x)\big)$ (blue) and $k_1^{*(2)}\big(\P(S_1+\bR_1>x)\big)$ (red) at time 1.}\label{fig:4}
\end{figure}

Finally, we plot the PO retention structure at time 1 in \Cref{fig:5}. The PO at time 1 retention functions (left plot) are 
$$
g_1^{*(1)}\big(S_1+\bR_{1}\big) = 0 
\ \ \hbox{and} \ \ 
g_1^{*(2)}\big(S_1+\bR_{1}\big) = S_1+\bR_{1}.
$$

\noindent As expected, agent 2 retains all of $S_1+\bR_1$ since they use a less conservative risk measure at time 1. Furthermore, as $g_1^{*(i)}$ is a non-decreasing function of $S_1+\bR_1$ for all $i\in\{1, 2\}$, the allocation process is comonotone in the sense of Definition \ref{definition:com_process}. On the right, we plot $\mfg_1^{*(i)}:=g_1^{*(i)}-R_1^{(i)}(Y_2^{*(i)})$, which is the total risk retained by the agent minus the risk-to-go of the agent. In this example, $\mfg_1^{*(1)}$ is negative, which means that the total risk retained by agent 1 at time 1 (which in this case is zero) is less than their risk-to-go. This difference is retained by the second agent, which acts as an insurer for agent 1 at time 1. We also note that $\mfg_1^{*(1)}$ is decreasing, so agent 1 has an incentive to over-report losses at time 1. However, it does not have the incentive to over-report the sum of the loss and the risk-to-go.      

\begin{figure} [!htbp]
\centering
\includegraphics[width=0.4\textwidth]{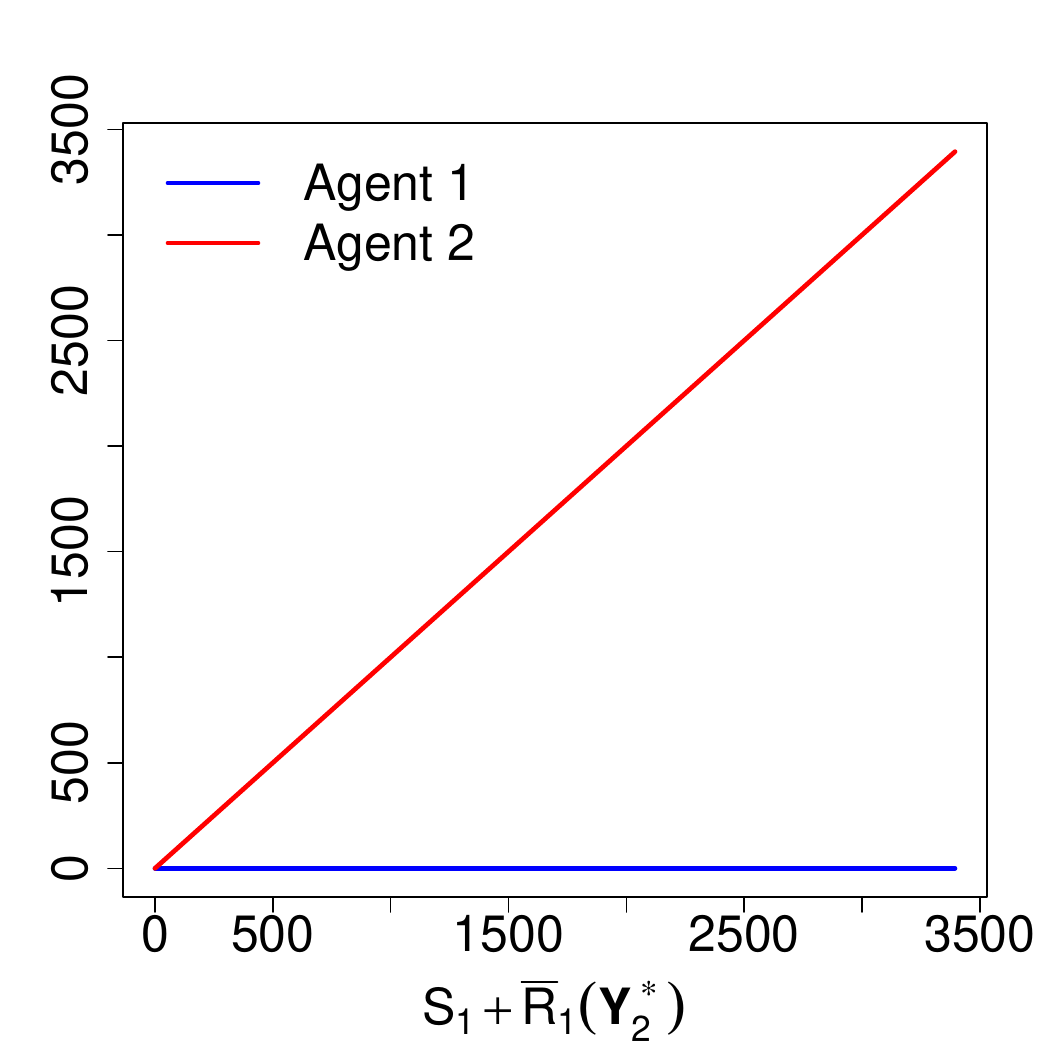}
\includegraphics[width=0.4\textwidth]{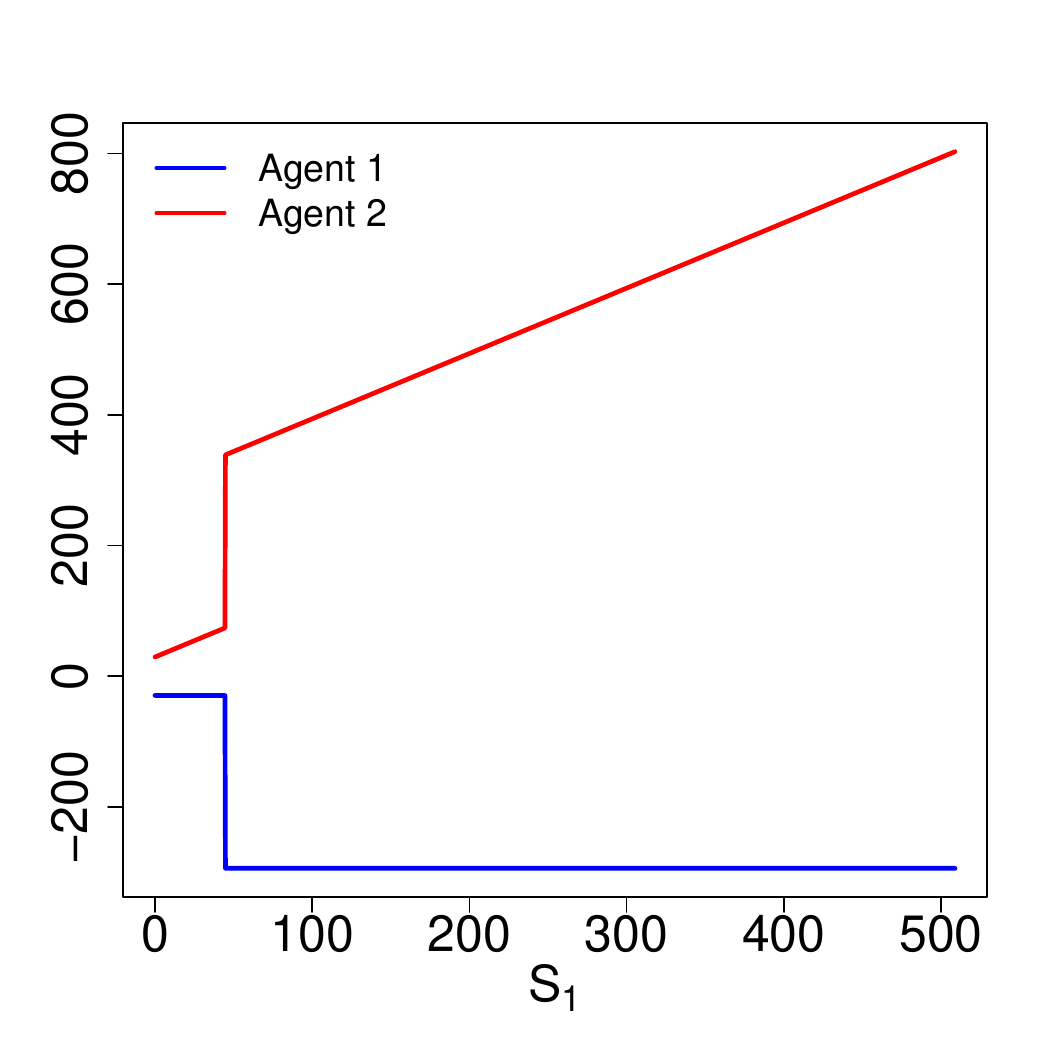}
\caption{The PO at time 1 retention functions $g_1^{*(1)}(S_1+\bR_1)$ (blue) and $g_1^{*(2)}(S_1+\bR_1)$ (red) are plotted on the left. On the right, we plot the total risk retained by the agent minus the risk-to-go of the agent.}\label{fig:5}
\end{figure}

\medskip

\subsection{Comparison with Myopic Pareto Optimality} \label{sec:dpo_mpo} 

In this section, we revisit the notion of myopic Pareto optimality defined in Section \ref{sec:myopic_po}. We first show that the PO allocation in Section \ref{sec:examples} is also comonotone myopic PO. Then, we consider alternative allocations that are comonotone myopic PO and argue that myopic Pareto optimality is not an appropriate notion of optimality when agent preferences are translation invariant.

The next result shows that comonotone Pareto optima are also comonotone myopic Pareto optima in a two-period setting when at least one agent uses the expected value at time $0$.  

\medskip

\begin{proposition} \label{prop:mpo_2} Suppose that agent preferences satisfy the assumptions of Theorem \ref{theorem:char_cdpo2} and at least one agent uses the expected value at time $0$. If $\Y_{1:2}^*$ is $c$-PO, then it is also comonotone myopic PO. 
\end{proposition}

\begin{proof} Let $\Y_{1:2}^*\in\CDPO$. By Proposition \ref{prop:cmpo}, it suffices to show that $\Y_{1:2}^*$ attains the infimum \eqref{eq:inf_cmpo}. Assume, without loss of generality, that $\underline{r}_1=\underline{s}_1=\underline{r}_2=\underline{s}_2=0$ and the first $j$ agents use the expected value at time $0$. By Lemma \ref{lemma:com_allocation}, any comonotone allocation has representation \eqref{eq:com_allocation}. Conversely, any allocation with representation \eqref{eq:com_allocation} is comonotone. Hence, we can rewrite \eqref{eq:inf_cmpo} as 
\[
\underset{\left\{
    \{g_{1}^{(i)}\}_{i\in\mN}\in\G, \  
        \{c_{1}^{(i)}\}_{i\in\mN}\in 
            \B(\{g_{1}^{(i)}\}_{i\in\mN}), \   
    \Y_2\in\A_2^C
    \right\}} \inf \ 
\sum_{i\in\mN}\rho_{0}^{(i)}\left(
    g_1^{(i)}\big(S_1+\bR_1(\Y_2)\big)
\right) + c_1^{(i)},
\]

\noindent where $\B\big(\{g_{1}^{(i)}\}_{i\in\mN}\big)$ is as defined in the proof of Theorem \ref{theorem:char_cdpo2}. Moreover, as $\underline{r}_1=\underline{s}_1=0$, we can further rewrite the infimum as 
\[
\underset{\left\{
    \{g_{1}^{(i)}\}_{i\in\mN}\in\G, \    
    \Y_2\in\A_2^C
    \right\}} \inf \ 
\sum_{i\in\mN}\rho_{0}^{(i)}\left(
    g_1^{(i)}\big(S_1+\bR_1(\Y_2)\big)
\right). 
\]

\medskip

\noindent Since any concave distortion function is lower bounded by the identity function (i.e., the distortion function of the expected value), we have $\rho_0^{(i)}(\cdot)\geq\E[\cdot]$ for all $i\in\mN$, and thus the infimum \eqref{eq:inf_cmpo} is lower bounded by
\begin{equation} \label{eq:inf_cmpo_2}
\underset{\left\{
    \{g_{1}^{(i)}\}_{i\in\mN}\in\G, \   
    \Y_2\in\A_2^C
    \right\}} \inf \ 
\sum_{i=1}^{j}\E\left[
    g_1^{(i)}\big(S_1+\bR_1(\Y_2)\big)
\right] = 
\underset{\Y_2\in\A_2^C} \inf \ 
    \E[S_1+\bR_1(\Y_2)],
\end{equation}

\noindent where the equality follows from the definition of $\G$. Moreover, we can further rewrite \eqref{eq:inf_cmpo_2} as
\begin{align}
& \ \E[S_1] \, + \underset{\Y_2\in\A_2^C} \inf \
    \E\left[\sum_{i\in\mN}\rho_1^{(i)}(Y_2^{(i)})\right] 
\nonumber \\
= & \ \E[S_1] \, + \underset{\left\{\{g_{2}^{(i)}\}_{i\in\mN}\in\G, \ 
    \sum_{i\in\mN}c_2^{(i)}=\underline{r}_2+\underline{s}_2
        \right\}} \inf \,
    \E\left[\sum_{i\in\mN} \rho_{1}^{(i)}\big(g_{2}^{(i)}(S_{2})\big) + 
        c_2^{(i)}
    \right] 
\nonumber \\
= & \ \E[S_1] \, + \underset{\{g_{2}^{(i)}\}_{i\in\mN}\in\G
    } \inf \,
\E\left[\sum_{i\in\mN} \rho_{1}^{(i)}\big(g_{2}^{(i)}(S_{2})\big)\right] 
\label{eq:inf_cmpo_3}, 
\end{align}

\noindent where the first equality follows from Lemma \ref{lemma:com_allocation} and the second equality holds since $\underline{r}_2=\underline{s}_2=0$.

\smallskip

It remains to show that the lower bound \eqref{eq:inf_cmpo_3} is attained by $\Y_{1:2}^*$. Since $L(x, k_2^{*(1)}, \ldots, k_2^{*(n)})=\{i\in\mN~|~1\leq i\leq j\}$, we have
\begin{align*}
\sum_{i\in\mN}\rho_0^{(i)}\big(Y_1^{*(i)}+\rho_1^{(i)}(Y_2^{*(t)})\big)     & = \sum_{i=1}^{j} \E\big[Y_1^{*(i)}+\rho_1^{(i)}(Y_2^{*(t)})\big] \\
& = \sum_{i=1}^{j} \E\big[g_1^{*(i)}\big(S_1+\bR_1(\Y_2^{*})\big)\big] \\
& = \E[S_1+\bR_1(\Y_2^{*})] 
 = \E[S_1] + \E\left[\sum_{i\in\mN}\rho_1^{(i)}(Y_2^{*(i)})\right],
\end{align*}

\noindent where the second equality follows from Lemma \ref{lemma:com_allocation}, the third equality holds since $\{g_1^{*(i)}\}_{i\in\mN}\in\G$, and the last equality follows by definition of $\bR_1(\Y_2^{*})$. Furthermore, since $\Y_{1:2}^*\in\CDPO$, by Theorem \ref{theorem:char_cdpo2}, we have 
\[
\E\left[\sum_{i\in\mN}\rho_1^{(i)}(Y_2^{*(i)})\right] = 
    \underset{\{g_{2}^{(i)}\}_{i\in\mN}\in\G
    } \inf \,
\E\left[\sum_{i\in\mN} \rho_{1}^{(i)}
    \big(g_{2}^{(i)}(S_{2})\big)\right],
\]

\noindent and thus $\Y_{1:2}^*$ attains the desired bound. 
\end{proof}

\medskip

Since agent 2 uses the expected value at time $0$ in our example, the allocation derived in Section \ref{sec:examples} is myopic PO. Furthermore, for any myopic PO allocation $\Y_{1:2}^*$, the allocation 
\[
\left\{(Y_1^{*(1)}-c, \, Y_1^{*(2)}+c), \, (Y_2^{*(1)}+c, \, Y_2^{*(2)}-c)\right\},
\]

\noindent is also myopic PO for any $c\in\R$. This implies that agents are willing to exchange any amount of risk-free capital (i.e., cash) as long as it is repaid at the terminal time. Furthermore, agents do not have a preference for smaller or larger values of $c$, which is an unreasonable assumption in practice. This phenomenon does not happen with dynamic PO allocations since the dynamic individual rationality constraint limits the amount of cash that the agents are willing to exchange at any period, leading to a more reasonable set of feasible allocations.   

\bigskip

\section{Conclusion} \label{sec:conclusion}

In this paper, we study Pareto optima in a discrete-time multi-period pure-exchange economy, where agents have preferences over stochastic endowment processes that are represented by strongly time-consistent dynamic risk measures, which have been extensively studied in the risk management literature. Whereas the literature on dynamic optimal allocations studies Pareto efficiency from the perspective of time-$0$ preferences, which we refer to as \emph{myopic Pareto optimality}, we introduce the notion of \emph{dynamic Pareto-optimal} allocation processes that takes into account the dynamic structure of risk profiles in the decision making process. We show that such dynamic Pareto-optimal allocation processes can be constructed recursively backwards in time, starting with the allocation at the terminal time, and that the allocation at each time point is the solution to a social welfare optimization problem.

\medskip

Furthermore, we extend the classical notion of comonotonicity to the dynamic setting, and we derive a dynamic comonotone improvement theorem for allocation processes. This motivates the study of dynamic Pareto optima in a comonotone dynamic market. We provide a crisp characterization of Pareto optima when agent preferences are coherent and satisfy a property that we call \emph{equidistribution-preserving}, and we provide a clear and easily implementable recursive algorithm for finding these optima. In the special case of the single period problem, we recover the characterization of \citet{ghossoub2026efficiency}. Finally, we demonstrate our algorithmic approach in a two-period numerical example.

\newpage

\setlength{\parskip}{0.5ex}
\hypertarget{LinkToAppendix}{\ }
\appendix

\vspace{-0.4cm}


\section[Allocations in a Single-Period Model]{Allocations in a Single-Period Model}\label{appendix:single_period}

In this appendix, we briefly discuss the optimal-allocation problem in the static setting, which is well studied in the literature. In the static setting, we work with the set of essentially bounded random variables $\L^{\infty}$ on a non-atomic probability space $(\Omega, \F, \P)$. We consider a problem with $n\in\N^{+}$ agents, where the $i$-th agent has an initial endowment $X^{(i)}\in\L^{\infty}$, $i\in\mN$ that is realized at the end of the period. The $n$ agents then pool their endowments to the aggregate endowment $S=\sum_{i\in\mN} X^{(i)}\in\L^{\infty}$ and wish to reallocate the aggregate endowment $S$ based on their preferences characterized by a (static) risk measure. That is, agent-$i$ uses the risk measure $\rho_0^{(i)}$, $i \in \mN$. We recall below some classical definitions in the static setting. 

\medskip

\begin{definition} [Allocation - Static] \label{definition:allocation}
A random vector $\Y:=(Y^{(1)}, \ldots, Y^{(n)})$ is an allocation of $S$ if each component is essentially bounded and $\sum_{i\in\mN}Y^{(i)}=S$. We denote by $\A$ the set of allocations, i.e., 
\begin{equation*}
    \A:= \Big\{\Y ~\big|~  Y^{(i)} \in \L^\infty \ \text{for all} \ i \in \mN
    \ \text{and} \ \sum_{i\in\mN}Y^{(i)}=S \Big\}\,.
\end{equation*}
\end{definition}

\medskip

\begin{definition} [Individual Rationality - Static] \label{definition:ir_single} An allocation $\Y\in\A$ is IR if for all $i\in\mN$,
\[
\rho_0^{(i)}(Y^{(i)})\leq \rho_0^{(i)}(X^{(i)}).
\]

\medskip

\noindent We denote by $\IR$ the set of all IR allocations, i.e., 
\begin{equation*}
    \IR:=\Big\{ \Y \in \A ~\big| ~  \rho_0^{(i)}(Y^{(i)})\leq \rho_0^{(i)}(X^{(i)}), \; i\in\mN \Big\}\,.
\end{equation*}
\end{definition}

\medskip

Note that the set $\IR$ is always non-empty since $\X\in\IR$. 

\medskip

\begin{definition} [Pareto Optimality - Static] \label{definition:po_single} An allocation $\Y^*\in\A$ is PO if 
\smallskip
\begin{enumerate}[label = $(\roman*)$]
    \item $\Y^*\in\IR$.
    \medskip
    \item there does not exist an allocation $\Y\in\IR$ that satisfies     
        $\rho_0^{(i)}(Y^{(i)})\leq \rho_0^{(i)}(Y^{*(i)})$
        for all $i\in\mN$, and at least one inequality is strict. 
\end{enumerate}

\medskip

\noindent We denote by $\PO$ the set of all PO allocations. 
\end{definition}

\medskip

If we restrict the market to the set of comonotone allocations, denoted $\A^C$, then we have the following notion of comonotone Pareto optimality. 

\medskip

\begin{definition} [Comonotone Pareto Optimality - Static] \label{definition:cpo_single} 
An allocation $\Y^*$ is $c$-PO if 
\smallskip
\begin{enumerate}[label = $(\roman*)$]
    \item $\Y^*\in\IR\cap\A^C$.
    \medskip
    \item there does not exist an allocation $\Y\in\IR\cap\A^C$ that satisfies     
        $\rho_0^{(i)}(Y^{(i)})\leq \rho_0^{(i)}(Y^{*(i)})$
        for all $i\in\mN$, and at least one inequality is strict. 
\end{enumerate}

\medskip

\noindent We denote by $\CPO$ the set of all $c$-PO allocations. 
\end{definition}

\medskip

The following proposition summarizes two seminal results for Pareto optima. Part $(i)$ characterizes Pareto optima for translation-invariant risk measures, and part $(ii)$ is the classical comonotone improvement theorem. We omit the proof of the proposition and refer to \citet{ghossoub2026efficiency} Proposition 3.3 for a proof of $(i)$ and \citet{carlier2012pareto} Theorem 3.1 for a proof of $(ii)$.

\medskip

\begin{proposition}[Static Case Results] \hfill \label{prop:static} 
\smallskip
\begin{enumerate}[label = $(\roman*)$]
    \item Suppose $\rho_0^{(i)}$ is translation invariant for all $i\in\mN$. Then $\Y\in\PO$ if and only if $\Y$ is optimal for the problem
    \smallskip
    \[
    \underset{
    \Y\in\IR} 
    \inf \,
        \sum_{i\in\mN}\rho_0^{(i)}(Y^{(i)}).
    \]

  \medskip
  
    \item For each $\Y\in\A$, there exists $\tilde{\Y}\in\A^C$ such that $\tilde{Y}^{(i)}\cvx Y^{(i)}$ for all $i\in\mN$. Moreover, if $\Y\notin\A^C$, then $\tilde{\Y}$ can be chosen such that $\tilde{Y}^{(j)}\scvx Y^{(j)}$, for some $j\in\mN$. 
\end{enumerate}
\end{proposition}

\vspace{0.7cm}

{\small
\bibliographystyle{ecta}
\bibliography{Refs}
}

\vspace{0.6cm}

\end{document}